\PassOptionsToPackage{unicode}{hyperref}
\PassOptionsToPackage{hyphens}{url}
\PassOptionsToPackage{dvipsnames,svgnames,x11names}{xcolor}
\documentclass[12pt]{article}

\usepackage{amsmath,amsfonts,bm}

\def\1{\bm{1}}

\def\vg{{\bm{g}}}

\DeclareMathAlphabet{\mathsfit}{\encodingdefault}{\sfdefault}{m}{sl}
\SetMathAlphabet{\mathsfit}{bold}{\encodingdefault}{\sfdefault}{bx}{n}

\def\gD{{\mathcal{D}}}

\def\gH{{\mathcal{H}}}

\def\gM{{\mathcal{M}}}

\def\gU{{\mathcal{U}}}
\def\gV{{\mathcal{V}}}

\def\gX{{\mathcal{X}}}
\def\gY{{\mathcal{Y}}}

\def\sP{{\mathbb{P}}}

\def\sR{{\mathbb{R}}}

\newcommand{\E}{\mathbb{E}}

\DeclareMathOperator*{\argmax}{arg\,max}
\DeclareMathOperator*{\argmin}{arg\,min}

\newcommand{\LRs}[1]{\left(#1\right)}

\newcommand{\LRl}[1]{\left\{#1\right\}}

\newcommand{\Prob}{\mathbb{P}}

\usepackage{algorithm}
\usepackage{algorithmic}
\usepackage{bbm}
\usepackage{amsmath,amssymb}
\usepackage{makecell}
\usepackage{bibunits}
\defaultbibliographystyle{asa}
\defaultbibliography{references}
\makeatletter
\newcommand{\setbibunitprefix}[1]{%
  \def\@extra@binfo{#1}%
  \def\@extra@b@citeb{#1}}
\makeatother

\usepackage{diagbox}
\usepackage{caption}
\usepackage{amsthm}
\newtheorem{theorem}{Theorem}[section]
\newtheorem{lemma}{Lemma}[section]

\newtheorem{assumption}{Assumption}
\newtheorem{proposition}{Proposition}[section]
\newtheorem{remark}{Remark}[section]

\newtheorem{definition}{Definition}

\usepackage{soul}

\usepackage{iftex}
\ifPDFTeX
  \usepackage[T1]{fontenc}
  \usepackage[utf8]{inputenc}
  \usepackage{textcomp} 
\else 
  \usepackage{unicode-math}
  \defaultfontfeatures{Scale=MatchLowercase}
  \defaultfontfeatures[\rmfamily]{Ligatures=TeX,Scale=1}
\fi
\usepackage{lmodern}
\ifPDFTeX\else  
\fi
\IfFileExists{upquote.sty}{\usepackage{upquote}}{}
\IfFileExists{microtype.sty}{
  \usepackage[]{microtype}
  \UseMicrotypeSet[protrusion]{basicmath} 
}{}
\makeatletter
\@ifundefined{KOMAClassName}{
  \IfFileExists{parskip.sty}{%
    \usepackage{parskip}
  }{
    \setlength{\parindent}{0pt}
    \setlength{\parskip}{6pt plus 2pt minus 1pt}}
}{
  \KOMAoptions{parskip=half}}
\makeatother
\usepackage{xcolor}
\makeatletter
\ifx\textbf\undefined\else
  \let\oldparagraph\textbf
  \renewcommand{\textbf}{
    \@ifstar
      \xxxParagraphStar
      \xxxParagraphNoStar
  }
  \newcommand{\xxxParagraphStar}[1]{\oldparagraph*{#1}\mbox{}}
  \newcommand{\xxxParagraphNoStar}[1]{\oldparagraph{#1}\mbox{}}
\fi
\ifx\subparagraph\undefined\else
  \let\oldsubparagraph\subparagraph
  \renewcommand{\subparagraph}{
    \@ifstar
      \xxxSubParagraphStar
      \xxxSubParagraphNoStar
  }
  \newcommand{\xxxSubParagraphStar}[1]{\oldsubparagraph*{#1}\mbox{}}
  \newcommand{\xxxSubParagraphNoStar}[1]{\oldsubparagraph{#1}\mbox{}}
\fi
\makeatother

\usepackage{longtable,booktabs,array}
\usepackage{calc} 
\usepackage{etoolbox}
\makeatletter
\patchcmd\longtable{\par}{\if@noskipsec\mbox{}\fi\par}{}{}
\makeatother
\IfFileExists{footnotehyper.sty}{\usepackage{footnotehyper}}{\usepackage{footnote}}
\makesavenoteenv{longtable}
\usepackage{graphicx}
\makeatletter
\def\maxwidth{\ifdim\Gin@nat@width>\linewidth\linewidth\else\Gin@nat@width\fi}
\def\maxheight{\ifdim\Gin@nat@height>\textheight\textheight\else\Gin@nat@height\fi}
\makeatother
\setkeys{Gin}{width=\maxwidth,height=\maxheight,keepaspectratio}
\makeatletter
\def\fps@figure{htbp}
\makeatother

\makeatletter
\@ifpackageloaded{caption}{}{\usepackage{caption}}
\AtBeginDocument{%
\ifdefined\contentsname
  \renewcommand*\contentsname{Table of contents}
\else
  \newcommand\contentsname{Table of contents}
\fi
\ifdefined\listfigurename
  \renewcommand*\listfigurename{List of Figures}
\else
  \newcommand\listfigurename{List of Figures}
\fi
\ifdefined\listtablename
  \renewcommand*\listtablename{List of Tables}
\else
  \newcommand\listtablename{List of Tables}
\fi
\ifdefined\figurename
  \renewcommand*\figurename{Figure}
\else
  \newcommand\figurename{Figure}
\fi
\ifdefined\tablename
  \renewcommand*\tablename{Table}
\else
  \newcommand\tablename{Table}
\fi
}
\@ifpackageloaded{float}{}{\usepackage{float}}
\floatstyle{ruled}
\@ifundefined{c@chapter}{\newfloat{codelisting}{h}{lop}}{\newfloat{codelisting}{h}{lop}[chapter]}
\floatname{codelisting}{Listing}

\makeatother
\makeatletter
\@ifpackageloaded{caption}{}{\usepackage{caption}}
\@ifpackageloaded{subcaption}{}{\usepackage{subcaption}}
\makeatother

\ifLuaTeX
  \usepackage{selnolig}  
\fi
\usepackage[]{natbib}
\usepackage{bookmark}
\usepackage{setspace}
\usepackage{booktabs}
\usepackage{multirow}

\IfFileExists{xurl.sty}{\usepackage{xurl}}{} 
\hypersetup{
  pdftitle={Conformalized Safe Feasible Sets in Uncertain Decision Systems},
  pdfauthor={Yajie Bao; Yinjie Min; Haojie Ren; Changliang Zou},
  pdfkeywords={Conformal inference; Constrained optimization; Data-driven decision; Inclusion guarantee; Setwise calibration.},
  colorlinks=true,
  linkcolor={blue},
  filecolor={Maroon},
  citecolor={Blue},
  urlcolor={Blue},
  pdfcreator={LaTeX via pandoc}}

\newcommand{\anon}{1}

\definecolor{HJblue}{RGB}{0,85,170}

\begin{document}

\def\spacingset#1{\renewcommand{\baselinestretch}%
{#1}\small\normalsize} \spacingset{1}

\makeatletter
\newcommand{\printfnsymbol}[1]{%
  \textsuperscript{\@fnsymbol{#1}}%
}
\makeatother

\begin{bibunit}
\setbibunitprefix{-main}

\if1\anon
{
  \title{\LARGE\bf Conformalized Safe Feasible Sets in Uncertain\\ Decision Systems}
  \author{Yajie Bao$^1$\thanks{All authors are listed in alphabetical order.}, Yinjie Min$^2$, Haojie Ren$^3$, and Changliang Zou$^2$\\
$^1$ {\normalsize School of Management, Fudan University}\\
$^2$ {\normalsize School of Statistics and Data Science, Nankai University}\\
$^3$ {\normalsize School of Mathematical Sciences, Shanghai Jiao Tong University}}
  \maketitle
} \fi

\if0\anon
{
  \bigskip
  \bigskip
  \bigskip
  \begin{center}
    {\LARGE\bf Conformalized Safe Feasible Sets in Uncertain\\ Decision Systems}
\end{center}
  \medskip
} \fi

\bigskip
\begin{abstract}
    {Safety-critical decision systems often require a downstream optimizer to choose from an unknown feasible set determined by an unobserved label $Y$. Given a context $X$, the
    goal is to construct a safe subset $D(X)$ contained in the oracle feasible set $A(X,Y)$ with probability at least $1-\alpha$. 
    Existing conformal approaches typically construct a prediction set of the unobserved label $Y$ and retain decisions that are safe for every value in this set. Although valid, this requires a stronger intermediate event than set inclusion.
    We propose \emph{Directed Inclusion Safety Calibration} (\texttt{DISC}), a conformal framework that directly controls the probability of this inclusion event by reducing its verification to a scalar critical-inclusion score. 
    Given a pretrained nested family of candidate feasible sets, \texttt{DISC} assigns each labeled observation the smallest nestedness level at which the corresponding subset is contained in $A(X,Y)$, and constructs the safe feasible set using an empirical quantile at test-time. 
    With data exchangeability, this yields a finite-sample, distribution-free inclusion guarantee. 
    {Under two practical set families, we show that \texttt{DISC} produces a safe feasible set containing that obtained by the corresponding calibration baseline.}
  We further develop optimization-based score computation and decision-aware procedures for learning subset families. Experiments across continuous and structured decision problems show that \texttt{DISC} achieves the target inclusion guarantee while producing larger feasible regions.}
\end{abstract}

\noindent%
{\it Keywords:} Conformal inference; Constrained optimization; Data-driven decision; Inclusion guarantee; Setwise calibration.
\vfill

\newpage
\spacingset{1.7} 


\section{Introduction}

Modern decision systems often act before safety-relevant uncertainty is fully resolved. For example, resource allocations are determined before future demand is realized, while an autonomous vehicle plans its motion without knowing the future trajectories of surrounding obstacles. {In such problems, the object passed to a downstream optimizer is often not a prediction of the unknown label itself, but a feasible set that includes a set of admissible decisions and typically depends on the unknown label. The key problem is therefore how to construct a {decision set} that is contained in the unknown feasible set with high probability.}

Let $u\in\gU$ denote a decision variable, where $\gU$ is the decision space that can be discrete or continuous. Let $ (X, Y) \in\gX\times\gY$ denote a pair of random variables associated with the system, where $X$ represents the observed context and $Y$ denotes the label that remains unobserved at the decision stage. {Let $\phi(\cdot): \gU \to \sR$ be the deterministic objective function, and let $f(\cdot;X,Y): \gU\to\sR$ be the data-dependent constraint function. We consider the following constrained optimization problem:
\[
\min_{u\in \gU} \phi(u) \quad \text{s.t.}\quad f(u;X,Y) \leq 0.
\]
The {\it oracle (true) feasible set} of the problem above is thus
\begin{equation}\label{eq:absolute_safety_region}
    A(X,Y)=\LRl{u\in\gU:f(u;X,Y)\leq 0}.
\end{equation}
When the true label $Y$ is unknown, a pretrained point predictor $\mathsf{ML}(\cdot):\gX\to\gY$ provides a natural plug-in estimate of $A(X,Y)$ by replacing the true label $Y$ with its prediction $\mathsf{ML}(X)$ in \eqref{eq:absolute_safety_region}. However, prediction error may cause the resulting set $A(X, \mathsf{ML}(X))$ to contain decisions that are infeasible under $Y$.}

Given independent and identically distributed (i.i.d.) data $\{(X_i,Y_i)\}_{i=1}^n$, together with the predictor $\mathsf{ML}(\cdot)$, our goal is to {construct a safe feasible set $D(X)\subseteq\gU$ satisfying the following safety inclusion guarantee}:
\begin{equation}\label{eq:set_inclusion}
    \sP\bigl\{D(X)\subseteq A(X,Y)\bigr\}\geq 1-\alpha,
\end{equation}
where $\alpha\in(0,1)$ is a prespecified safety level.
This is a {setwise feasibility guarantee}: on the inclusion event $D(X) \subseteq A(X,Y)$, every decision $u\in D(X)$ satisfies the feasibility condition under the true label $Y$, allowing an arbitrary objective $\phi(u)$ to be optimized over the same $D(X)$. Since an empty $D(X)$ satisfies \eqref{eq:set_inclusion} trivially, the key challenge is to construct a valid safe feasible set that is sufficiently large and flexible to support effective decision-making.

A canonical application is chance-constrained programming \citep{charnes1959chance,nemirovski2006convex}, which seeks to minimize the objective \(\phi(u)\) subject to a probabilistic constraint \(\sP\{f(u;X,Y)\leq 0\}\geq 1-\alpha\). Here, \(X\) denotes observed contextual information, and \(Y\) captures unresolved future uncertainty, such as demand or system disturbances. In this case, the safe feasible set \(D(X)\) satisfying \eqref{eq:set_inclusion} provides a data-dependent inner approximation to the unknown feasible set \(A(X,Y)\), 
and any downstream solution \(\widehat u(X)\in\argmin_{u\in D(X)}\phi(u)\) inherits the probabilistic feasibility. {For another example }of the control or planning problem, $Y$ may represent the future positions of a moving obstacle or another agent, and $f(u;X,Y)$ measures violation of a required separation between the decision trajectory and the obstacle \citep{lindemann2023safe,dixit2023adaptive}. {In this case, the safe feasible set $D(X)$ is a region where every trajectory can avoid collision with the obstacles with high probability.}
The same formulation can also extend to structured output problems, such as LLM reasoning problems, where $Y$ records the validity of generated claims and $f(u;X,Y)$ measures whether the claims retained by a candidate output satisfy the desired validity criterion \citep{wei2022chain,zhouleast}. Although the decision spaces and feasibility criteria differ across these applications, they share the same target: identifying a collection of decisions that are feasible under the unknown realization $Y$. 

\subsection{Our framework and contributions}

Classical predictive inference provides a {baseline} route to achieving the inclusion target \eqref{eq:set_inclusion}. Given labeled data and an arbitrary predictive model, conformal prediction \citep{vovk2005algorithmic,lei2018distribution} constructs a prediction set \(C(X)\subseteq\gY\) satisfying \(\sP\{Y\in C(X)\}\geq 1-\alpha\) under exchangeability. 
Robust optimization \citep{ben2009robust,chenreddy2022data} then retains only decisions that are safe for all \(y\in C(X)\), {producing the feasible set
\begin{equation}\label{eq:RO_safe_set}
    D^{\texttt{Base}}(X) = \{u\in\gU:f(u;X,y)\leq 0,\ \forall~y\in C(X)\}.
\end{equation}
Whenever \(Y\in C(X)\), this subset $D^{\texttt{Base}}(X)$ is contained in the oracle feasible set \(A(X,Y)\), and hence the label coverage,  $\sP\{Y\in C(X)\}\geq 1-\alpha$, implies the desired safety guarantee \eqref{eq:set_inclusion}.} However, the label coverage is sufficient but not necessary for \eqref{eq:set_inclusion}. Even when \(Y\notin C(X)\), the constructed subset may still be entirely safe under the realized \(Y\). {Controlling label coverage therefore controls a stronger intermediate event than the required feasibility.} This mismatch can unnecessarily shrink the feasible set and exclude useful decisions, or even produce an empty set. 

Motivated by the gap between label coverage and set inclusion, we propose a new framework \emph{Directed Inclusion Safety Calibration} (\texttt{DISC}) that directly controls the probability of the desired inclusion event. Specifically, we consider a nested family of candidate feasible sets $\{D(X;\lambda):\lambda\in\sR\}$, where $D(X;\lambda)\subseteq\gU$ and $D(X;\lambda_2)\subseteq D(X;\lambda_1)$ whenever $\lambda_1\leq\lambda_2$. Thus, larger values of $\lambda$ correspond to more conservative feasible sets. This formulation separates the construction of a safe feasible set into two components: specifying a nested set family that incorporates available predictive or structural information, and calibrating the nestedness level $\lambda$ using labeled data. 
{To calibrate $\lambda$, for each labeled observation \((X_i,Y_i)\), \texttt{DISC} defines the score \(V_i=\inf\{\lambda\in\sR:D(X_i;\lambda)\subseteq A(X_i,Y_i)\}\), the smallest nestedness level required for the candidate subset to be contained in the corresponding oracle feasible set. The test-time nestedness level is then chosen as an appropriate empirical quantile of these scores. 
By using the target inclusion event directly, rather than stronger sufficient conditions such as label coverage, \texttt{DISC} can reduce conservativeness while retaining feasibility.} 
Our contributions are summarized as follows.
\begin{itemize}
    \item 
    For any fixed nested feasible set family, \texttt{DISC} provides a finite-sample, distribution-free inclusion guarantee under exchangeability. 
    {Specifically, we consider two practical nested set families. The first is termed the \textit{prediction-set} family, which follows the robust optimization construction in \eqref{eq:RO_safe_set}, with $C(X)$ replaced by a sequence of nested prediction sets. The second is termed the \textit{additive residual} family, which adds a tunable margin to the predicted constraint value $f(u;X,\mathsf{ML}(X))$ to account for prediction error and retains decisions for which the adjusted value is nonpositive. For both families, we show that the \texttt{DISC} safe feasible sets contain those returned by the existing baseline calibration methods.}
    

    \item 
    {Computing the \texttt{DISC} score typically requires verifying set inclusion over a potentially continuous decision space. We formulate this verification as a constrained optimization problem at each fixed nestedness level $\lambda$ and use bisection over $\lambda$ to locate the score. We further provide family-specific reductions to simplify the optimization. When an exact solution of the constrained optimization is unavailable, we construct computable upper bounds on the \texttt{DISC} scores while preserving the finite-sample inclusion guarantee.}
    
    

    \item 
    We develop \texttt{O-DISC}, which optimizes a parameterized family of feasible sets according to downstream decision risk. We establish an upper bound on the inclusion error of the optimized subset that quantifies the effect of family optimization, together with a nonasymptotic excess decision-risk bound relative to a population benchmark. We further develop \texttt{FO-DISC}, which restores the exact finite-sample inclusion guarantee at an additional computational cost.
    

\end{itemize}

\subsection{Related work}\label{sec:related_work}

Conformal prediction \citep{vovk2005algorithmic,lei2018distribution} provides prediction sets satisfying finite-sample marginal coverage under exchangeability without requiring a correctly specified prediction model. Several extensions move beyond label coverage to more general risk criteria depending on the label $Y$. \citet{bates2021distribution} develop risk-controlling prediction sets that bound a population-level risk with target probability. The Learn-then-Test framework developed by \citet{angelopoulos2025learn} aims to output a post-processing of prediction values that satisfy some statistical error rate control. Given a loss $\ell_{\lambda}(X,Y) \leq 1$ that varies monotonically along the parameter $\lambda$, \citet{angelopoulos2024conformal} propose the conformal risk control (CRC) framework to choose the test-time parameter via $\hat\lambda = \inf\left\{\lambda\in \sR: \frac{1}{n+1}(\sum_{i=1}^n \ell_{\lambda}(X_i,Y_i) + 1) \le \alpha\right\}$, which enjoys $\E[\ell_{\hat\lambda}(X,Y)] \leq \alpha$ under exchangeability. In our setting, the inclusion target \eqref{eq:set_inclusion} is equivalent to controlling the expectation of the binary loss $\mathbbm{1}\{D(X;\lambda)\nsubseteq A(X,Y)\}$, so CRC provides the underlying calibration principle. However, evaluating this binary loss requires verifying the constraint over the entire candidate feasible set. The generic CRC framework does not itself resolve this verification problem, so computing $\hat\lambda$ remains challenging.

Next, we review related work in two application areas relevant to the inclusion target \eqref{eq:set_inclusion}. For chance-constrained programming, classical approaches \citep{ben2009robust,nemirovski2006convex} use structural and distributional information to construct tractable safe approximations of the chance constraint $\sP\{f(u;X,Y) \le 0\}\ge 1-\alpha$. More recently, \citet{zhao2024conformal} propose conformal predictive programming, which replaces chance constraints with empirical-quantile constraints and uses an independent calibration set to obtain a posteriori feasibility guarantees for the resulting decision. In planning and control problems, several works \citep{lindemann2023safe,dixit2023adaptive,yang2023safe} construct safe decision regions of the baseline form \eqref{eq:RO_safe_set} by enforcing feasibility uniformly over conformal prediction sets.
These approaches obtain simultaneous feasibility through label coverage. \texttt{DISC} instead targets the inclusion event itself, avoiding a sufficient condition that can be unnecessarily restrictive.

Several works apply conformal prediction to related decision-making tasks. 
For example, \citet{johnstone2021conformal}, \citet{sun2023predict} and \citet{patel2024conformal} use conformal prediction sets as the uncertainty set in contextual robust optimization problems \citep{chenreddy2022data}, where the robustness of decisions can be guaranteed by the coverage of prediction sets. In addition, another line tailors the construction of conformal prediction sets to the downstream decision cost \citep{kiyani2025decision,cortesgomez2025utility,yeh2024end,bao2025optimal,wang2026optimal}. Specifically, \citet{hu2026conformal} develop a new method to directly control the robustness of the solution of robust optimization, instead of label coverage of prediction sets.

\paragraph*{Organization.}
The remainder of the paper is organized as follows. Section~\ref{sec:method} introduces \texttt{DISC}, develops methods for computing and conservatively upper-bounding its calibration scores, and instantiates the framework under two nested feasible set families. Section~\ref{sec:family_optimization} develops efficiency-aware optimization of safe feasible sets and studies its inclusion and decision-risk guarantees. Sections~\ref{sec:simulation} and~\ref{sec:real_data} present simulation studies and real-data applications, respectively, and Section~\ref{sec:conclusion} concludes and gives future directions. 


\section{Conformal Construction for Safe Feasible Sets}\label{sec:method}

We now formalize \texttt{DISC} for a nested family of feasible sets. Let labeled data $\gD_n=\{(X_i,Y_i)\}_{i=1}^n$ denote the calibration data, and let $X_{n+1}$ be a test context with unknown label $Y_{n+1}$. Given a nested family \(    \{D(X;\lambda):\lambda\in\sR\}\) such that \(D(X;\lambda_2)\subseteq D(X;\lambda_1)\) whenever \(\lambda_1\leq\lambda_2,\) our goal is to calibrate the nestedness level such that \eqref{eq:set_inclusion} holds for the test data.

\subsection{Directed inclusion safety calibration method}\label{sec:dir_inclusion}

For any $i\in[n+1]$ and $\lambda\in\sR$, the candidate subset is safe under the realized label \(Y_i\) if and only if its worst-case safety violation is nonpositive:
\begin{align}
    D(X_i;\lambda)\subseteq A(X_i,Y_i)\Longleftrightarrow \sup_{u\in D(X_i;\lambda)}f(u;X_i,Y_i)\leq0.\nonumber
\end{align}
{Since \(D(X_i;\lambda)\) shrinks as \(\lambda\) increases, the inclusion event is monotone in \(\lambda\): once the candidate subset becomes safe, it remains safe at every larger nestedness level. Consequently, each observation \((X_i,Y_i)\) determines a single critical threshold separating unsafe and safe candidate subsets.}

This equivalence motivates the explicit form of the \texttt{DISC} score: $V_i$ is the smallest nestedness level at which the entire candidate subset is contained in the oracle feasible set, that is, 
\begin{equation}\label{eq:dir_score}
    V_i
    =
    \inf\left\{
        \lambda\in\sR:
        \sup_{u\in D(X_i;\lambda)}
        f(u;X_i,Y_i)\leq0
    \right\}.
\end{equation}
Throughout the paper, we assume that the infimum in the definition \eqref{eq:dir_score} is attained, namely, for each \(i\in[n+1]\), the score \(V_i\) is finite and $\sup_{u\in D(X_i;V_i)}f(u;X_i,Y_i)\leq0$. Together with nestedness, this implies that $D(X_i;\lambda)\subseteq A(X_i,Y_i)$ if and only if $V_i\leq\lambda$. 
{This is the central reduction underlying \texttt{DISC}, which translates the original set-valued inclusion exactly into a scalar threshold event. 

Thus, we can calibrate the nestedness level using the empirical quantile of these \texttt{DISC} scores.} Specifically, let $\hat\lambda=\mathsf{Q}_{1-\alpha}\LRs{\{V_i\}_{i=1}^n,\infty}$, where $\mathsf{Q}_{1-\alpha}(
\cdot)$ denotes the empirical $(1-\alpha)$-quantile. The resulting \texttt{DISC} safe feasible set is $\widehat D^{\texttt{DISC}}(X_{n+1})=D(X_{n+1};\hat\lambda)$.

\begin{theorem}\label{thm:inclusion}
If $\{(X_i,Y_i)\}_{i=1}^{n+1}$ are exchangeable, then $\sP\bigl\{\widehat D^{\texttt{DISC}}(X_{n+1})\subseteq A(X_{n+1},Y_{n+1})\bigr\} \geq1-\alpha.$
Additionally, if $\{V_i\}_{i=1}^{n+1}$ are distinct with probability one, then $1-\alpha
    \leq
    \sP\bigl\{
        \widehat D^{\texttt{DISC}}(X_{n+1})
        \subseteq A(X_{n+1},Y_{n+1})
    \bigr\}
    \leq
    1-\alpha+1/(n+1)$.
\end{theorem}

Theorem~\ref{thm:inclusion} provides a finite-sample inclusion guarantee analogous to the standard coverage guarantee of split conformal prediction \citep{lei2018distribution}. {Rather than conformalizing a label-space score such as a prediction residual, \texttt{DISC} conformalizes the amount of shrinkage required to make an entire decision subset safe. Once this setwise condition has been encoded by the scalar score \(V_i\), finite-sample validity follows directly from exchangeability.}

{Before proceeding, let us discuss a useful decision-space interpretation of the score $V_i$ in \eqref{eq:dir_score}. Since inclusion fails precisely when \(D(X_i;\lambda)\) contains at least one unsafe decision satisfying \(f(u;X_i,Y_i)>0\), $V_i$ is equivalently the nestedness level at which the last remaining unsafe decision is excluded. Explicitly, that is $V_i = \sup_{u\in A^c(X_i,Y_i)}\inf\LRl{\lambda\in\sR: u\notin D(X_i;\lambda)}$, where $A^c(X_i,Y_i)=\{u\in\gU:f(u;X_i,Y_i)>0\}$ denotes the oracle infeasible set.}

The nested structure also facilitates score computation. Since the subset $D(X_i;\lambda)$ shrinks as $\lambda$ increases, the map $\lambda\mapsto\sup_{u\in D(X_i;\lambda)}f(u;X_i,Y_i)$ is nonincreasing. Hence, whenever the inner supremum in \eqref{eq:dir_score} can be solved, we can compute $V_i$ using bisection. Section~\ref{sec:score_computation} develops exact and conservative procedures for the $V_i$'s computation.



\subsection{Instantiation of \texttt{DISC} under two nested families}\label{sec:instantiation}

The validity result above requires only a fixed nested family. To turn this general calibration rule into practical constructions, we specify the candidate subset family through a \emph{margin function} $L_\lambda(\cdot;X):\gU\to\sR$, that is
\begin{equation}\label{eq:safety_margin_representation}
    D(X;\lambda)
    =
    \{u\in\gU:L_\lambda(u;X)\leq0\}.
\end{equation}
Here $L_\lambda(u;X)$ is nondecreasing in $\lambda$, i.e., $L_{\lambda_1}(u;X) \leq L_{\lambda_2}(u;X)$ if $\lambda_1\leq \lambda_2$, ensuring the required nestedness of $D(X;\lambda)$. This representation is not required for the general validity of \texttt{DISC}, but it provides a convenient structure for constructing practical safe feasible sets and deriving tractable \texttt{DISC} scores. 
The representation \eqref{eq:safety_margin_representation} admits two practically useful instantiations. 
Under both families, the structure of $L_\lambda$ yields an explicit representation of the \texttt{DISC} score in \eqref{eq:dir_score} and facilitates comparison with existing calibration rules.

\subsubsection{Prediction-set margin family}\label{sec:pred_set_family}

{We first consider a family induced by prediction sets in the label space. Let \(C(X;\lambda)=\{y\in\gY:s(X,y)\leq\lambda\}\), where $s:\gX\times\gY\to\sR$ is a pretrained nonconformity score.} 
The prediction-set margin function is defined by  \(L_\lambda(u;X)=\sup_{y\in C(X;\lambda)}f(u;X,y)\). 
{The corresponding candidate feasible set is $D(X;\lambda)=\LRl{u\in\gU: \sup_{y\in C(X;\lambda)}f(u;X,y)\leq 0}$, which consists of decisions that are safe for every label in $C(X;\lambda)$. As a larger $\lambda$ produces a larger prediction set $C(X;\lambda)$, the margin function \(L_\lambda(u;X)\) is nondecreasing and \(D(X;\lambda)\) becomes smaller. The family therefore satisfies the required nestedness. }
Additionally, for many commonly used shapes of prediction sets, such as ellipsoids, boxes, and polyhedra, $L_\lambda(u;X)$ can be computed either in closed form or by solving a standard convex optimization problem. Detailed formulations are provided in Appendix~\ref{appen:score_computation}.

Proposition~\ref{prop:dir_score_PS} states that {the prediction-set margin family yields the following explicit form of }the \texttt{DISC} score. 

\begin{proposition}
\label{prop:dir_score_PS}
Under the prediction-set margin family, the \texttt{DISC} score \eqref{eq:dir_score} is
\begin{align}\label{eq:dir_score_PS}
    V_i
    =
    \sup_{u\in A^c(X_i,Y_i)}\,
    \LRl{\inf_{y\in\gY:\,f(u;X_i,y)>0}
    s(X_i,y)}.
\end{align}
\end{proposition}

{To interpret \eqref{eq:dir_score_PS}, consider a decision \(u\in A^c(X_i,Y_i)\) that is unsafe under the realized label \(Y_i\). Such a decision \(u\) should be excluded from  \(D(X_i;\lambda)\) once \(C(X_i;\lambda)\) contains at least one label \(y\) under which \(u\) is unsafe. The inner problem in \eqref{eq:dir_score_PS} is the smallest threshold of this prediction set required to contain such an unsafe witness $y$.
The outer supremum then identifies the unsafe decision that is hardest to exclude. 
Consequently, \(V_i\) is the smallest threshold at which every truly unsafe decision has an unsafe witness inside \(C(X_i;\lambda)\), and hence has been removed from the candidate subset.

This representation in \eqref{eq:dir_score_PS} also separates the score computation into two components: the inner optimization only over the label space and the outer optimization over the oracle infeasible set. The inner problem is tractable for many commonly used prediction scores \(s(X_i,\cdot)\) and constraint functions \(f(u;X_i,\cdot)\).} 
The following remark illustrates two cases in which the inner problem admits a closed-form or efficiently computable solution.

\begin{remark}
Consider the Mahalanobis score \(s(x,y)=(y-\mathsf{ML}(x))^\top\Sigma(y-\mathsf{ML}(x))\) with a positive definite symmetric matrix \(\Sigma\), which leads to an ellipsoidal prediction set \(C(X;\lambda)\). If \(f(u;X_i,y)=a_i(u)^\top y+b_i(u)\) with \(a_i(u)\neq0\), then the inner infimum of \eqref{eq:dir_score_PS} has a closed form: $\big[-\{a_i(u)^\top\mathsf{ML}(X_i)+b_i(u)\}\big]_+^2/\{a_i(u)^\top\Sigma^{-1} a_i(u)\}$, where $[\cdot]_+$ denotes the positive part. Another example is to suppose that $f(u;X_i,y)=r_i-\|y-c_i(u)\|_2$ with $r_i>0$. The unsafe label set is then an open Euclidean ball. The inner infimum is zero when $\|\mathsf{ML}(X_i)-c_i(u)\|_2\leq r_i$; otherwise, it is the squared Mahalanobis distance to the closed ball.

\end{remark}

\paragraph*{Label coverage calibration baseline.}
{The same prediction-set family naturally gives rise to the {standard predict-then-robustify} procedure from a robust optimization perspective. That is,} $C(X;\lambda)$ serves as an uncertainty set for the unknown label $Y$, and imposing $\sup_{y\in C(X;\lambda)}f(u;X,y)\leq 0$ gives the robust counterpart \citep{ben2009robust} of the unknown constraint $f(u;X,Y)\leq 0$. Existing baselines \citep{johnstone2021conformal,sun2023predict} calibrate the threshold of $C(X;\lambda)$ directly for label coverage. 
The resulting coverage-based score is defined as
\begin{align}\label{eq:score_coverage_PS}
    V_i^{\texttt{Base}} = \inf\{\lambda\in \sR: Y_i\in C(X_i;\lambda)\} = s(X_i,Y_i).
\end{align}
Let $\hat\lambda^{\texttt{Base}} =\mathsf{Q}_{1-\alpha} \LRs{\{V_i^{\texttt{Base}}\}_{i=1}^n,\infty}$. Then $C(X_{n+1};\hat\lambda^{\texttt{Base}})$ is precisely the usual split conformal prediction set \citep{lei2018distribution} and therefore satisfies a finite-sample label coverage guarantee \(\sP\LRl{Y_{n+1}\in C(X_{n+1};\hat\lambda^{\texttt{Base}})}\geq 1-\alpha\). The induced safe feasible set $\widehat D^{\texttt{Base}}(X_{n+1})=D(X_{n+1};\hat\lambda^{\texttt{Base}})$ satisfies the desired safety guarantee \eqref{eq:set_inclusion}. {However, the converse need not hold as the resulting decision subset may remain entirely safe even when the prediction set fails to cover \(Y_{n+1}\).} Thus, controlling label coverage controls a stronger event and can be more conservative than \texttt{DISC}. 

{The difference between the two procedures is transparent from their scores. For any unsafe decision \(u\in A^c(X_i,Y_i)\), the realized label \(Y_i\) is feasible for the inner optimization in \eqref{eq:dir_score_PS}. Hence, \(\inf_{y\in\gY:\,f(u;X_i,y)>0}s(X_i,y)\leq s(X_i,Y_i)\), and taking the supremum over \(u\in A^c(X_i,Y_i)\) gives the precise comparison in Proposition~\ref{prop:ps_scp_comparison}.} 



\begin{proposition}
\label{prop:ps_scp_comparison}
Under the prediction-set margin family, for every calibration point $i\in[n]$, we have
\(V_i\leq V_i^{\texttt{Base}}\). Consequently, 
\(\widehat D^{\texttt{Base}}(X_{n+1})\subseteq\widehat D^{\texttt{DISC}}(X_{n+1})\) holds almost surely.
\end{proposition}

\subsubsection{Additive residual-margin family}\label{sec:additive_margin}

We next consider a margin function family that operates directly on the predicted safety function. Let \(\hat f(u;X)\) be a pretrained surrogate safety function, such as the plug-in function $f(u;X,\mathsf{ML}(X))$. We define \(L_\lambda(u;X)=\hat f(u;X)+\mu(u)+\lambda\sigma(u)\), where $\mu(u)$ and $\sigma(u)$ are fixed before calibration. {The resulting candidate feasible set is \(D(X;\lambda)=\LRl{u\in\gU:~\hat f(u;X)+\mu(u)+\lambda\sigma(u)\leq 0}\). Since \(\sigma(u)>0\), increasing \(\lambda\) raises the margin function pointwise and therefore shrinks \(D(X;\lambda)\), as required by the nestedness condition.} This class is motivated by modeling the residual process \(f(u;X_i,Y_i)-\hat f(u;X_i)\) as a functional signal-plus-noise model, with mean function \(\mu(u)\) and scale function \(\sigma(u)>0\). This additive residual-margin family leads to an explicit form of the \texttt{DISC} score.

\begin{proposition}
\label{prop:dir_score_add}
Under the additive residual-margin family, the \texttt{DISC} score in \eqref{eq:dir_score} is
\begin{equation}\label{eq:dir_score_add}
    V_i=\sup_{u\in A^c(X_i,Y_i)}-\frac{\hat f(u;X_i)+\mu(u)}{\sigma(u)}\,.
\end{equation}
\end{proposition}


{The score $V_i$ in \eqref{eq:dir_score_add} can also be interpreted by the pointwise-exclusion representation. That means $V_i$ is the critical level at which the hardest decision that is unsafe under $Y_i$ is removed from the candidate subset.}

For the additive margin family, the representation \eqref{eq:dir_score_add} {brings an immediate computational advantage and} avoids the bisection over $\lambda$ required by the original formulation \eqref{eq:dir_score}. Computing $V_i$ is particularly convenient when the infeasible set \(A^c(X_i,Y_i)\) has tractable geometry. For example, suppose that the closure of $A^c(X_i,Y_i)$ is a finite union of compact polyhedra, which includes the case where $\gU$ is a compact polyhedron and $f(u;X_i,Y_i)$ is affine or max-affine in $u$. If $\hat f(u;X_i)+\mu(u)$ and $\sigma(u)$ are affine in $u$, with $\sigma(u)>0$, then the objective in \eqref{eq:dir_score_add} is linear-fractional and attains its maximum over each polyhedral component at a vertex \citep{boyd2004convex}. The representation \eqref{eq:dir_score_add} also enables a direct comparison with the uniform-residual benchmark below.

\paragraph*{Uniform residual calibration baseline.}
Uniform residual calibration is a natural alternative and is widely used in conformal methods for functional and trajectory-valued data~\citep{diquigiovanni2024importance,zhou24lconformalized}. It uses the largest standardized residual over the entire decision space, and the score is:
\begin{align}\label{eq:uniform_score}
    V_i^{\texttt{Base}}
    =
    \sup_{u\in\gU}
    \frac{f(u;X_i,Y_i) - \hat{f}(u;X_i) -\mu(u)}{\sigma(u)}.
\end{align}
Let \(\widehat D^{\texttt{Base}}(X_{n+1})=D(X_{n+1};\hat\lambda^{\texttt{Base}})\) be the safe feasible set output by this baseline method, where $\hat\lambda^{\texttt{Base}}=\mathsf{Q}_{1-\alpha}\LRs{\{V_i^{\texttt{Base}}\}_{i=1}^n,\infty}$. The validity of this procedure follows from a uniform upper-envelope argument. {For any $\lambda$, \(V_{n+1}^{\texttt{Base}}\leq\lambda\) implies \(f(u;X_{n+1},Y_{n+1})\leq L_{\lambda}(u;X_{n+1})\) for every \(u\in\gU\). In turn, every decision satisfying \(L_\lambda(u;X_{n+1})\leq0\) also satisfies \(f(u;X_{n+1},Y_{n+1})\leq0\), and hence \(D(X_{n+1};\lambda)\subseteq A(X_{n+1},Y_{n+1})\). The standard conformal inference applied to $\{V_{i}^{\texttt{Base}}\}_{i=1}^{n+1}$ therefore yields \(\sP\LRl{\widehat D^{\texttt{Base}}(X_{n+1})\subseteq A(X_{n+1},Y_{n+1})}\geq 1-\alpha\).}
In fact, this method outputs the safe feasible set analogue of simultaneous prediction bands for functional data.

Uniform calibration is valid but imposes a stronger requirement than the inclusion target. It controls the residual process over the entire \(\gU\), including decisions irrelevant to the target inclusion event. {By contrast, \texttt{DISC} based on \eqref{eq:dir_score_add} adjusts only the margin needed to exclude decisions that are unsafe under the realized label. This distinction yields a direct samplewise comparison. For every \(u\in A^c(X_i,Y_i)\), we have \(f(u;X_i,Y_i)>0\), and therefore \(-\frac{\hat{f}(u;X_i)+\mu(u)}{\sigma(u)}\leq \frac{f(u;X_i,Y_i) - \hat{f}(u;X_i) -\mu(u)}{\sigma(u)} \leq V_i^\texttt{Base}\). Taking the supremum over \(u\in A^c(X_i,Y_i)\),} the next proposition formally shows that our method is more efficient.

\begin{proposition}
\label{prop:additive_unif_comparison}
Under the additive residual-margin family, for every calibration point $i\in[n]$, we have
\(V_i\leq V_i^{\texttt{Base}}\). Consequently, \(\widehat D^{\texttt{Base}}(X_{n+1})\subseteq\widehat D^{\texttt{DISC}}(X_{n+1})\) holds almost surely.
\end{proposition}

\subsection{Computation of DISC scores}\label{sec:score_computation}

{Section~2.2 derives explicit representations of the \texttt{DISC} scores in two margin families. We now consider the computation of \(V_i\) for a general margin function family. Recall the score definition in \eqref{eq:dir_score}.} The margin representation \eqref{eq:safety_margin_representation} turns score computation into a verification problem at each nestedness level.  
For each labeled observation, write $f_i(u)\equiv f(u;X_i,Y_i)$ and $L_{i,\lambda}(u)\equiv L_\lambda(u;X_i)$, and define the certification value
\begin{equation}\label{eq:certification_value}
    g_i(\lambda) = \sup_{u\in D(X_i;\lambda)} f(u;X_i,Y_i),
\end{equation}
which represents the worst-case violation over the candidate subset $D(X_i;\lambda)$. 
Whenever $g_i(\lambda)\le 0$ can be verified for any $\lambda$, the score $V_i$ can be computed by bisection. {If \(\gU\) is finite, as in the LLM reasoning problem \citep{wei2022chain,zhouleast}}, \(g_i(\lambda)\) can be computed exactly by enumerating the
points satisfying \(L_{i,\lambda}(u)\leq0\). 
For continuous decision spaces, exact evaluation of $g_i(\lambda)$ remains possible in the following three standard optimization regimes. 

\noindent
\emph{(1) Concave safety function.} If \(f_i\) is concave and \(D(X_i;\lambda)\) is convex, then computing \(g_i(\lambda)\) in \eqref{eq:certification_value} is a standard convex optimization problem, {since it maximizes a concave objective over a convex set} \citep{boyd2004convex}. This case includes the collision-avoidance problem and can be handled using off-the-shelf algorithms.

\noindent
\emph{(2) Affine safety function.} If \(f_i\) is affine (or max-affine) and \(D(X_i;\lambda)\) is polyhedral, then \(g_i(\lambda)\) can be computed by solving a linear program. Such a polyhedral subset arises when \(\gU\) is polyhedral and \(L_\lambda(u;X_i)\) is affine or max-affine in \(u\), as in {the resource allocation problem and related chance-constrained applications} \citep{nemirovski2006convex}.

\noindent
\emph{(3) Convex safety function.} If \(f_i\) is continuous and convex and \(D(X_i;\lambda)\) is a nonempty compact polyhedron, then the maximum is attained at a vertex of $D(X_i;\lambda)$. Hence, \(g_i(\lambda)\) can be computed exactly by enumerating the vertices of \(D(X_i;\lambda)\), although the number of vertices may grow exponentially with the problem dimension.

    
    

When exact global optimization is intractable, we replace $g_i(\lambda)$ by a \emph{computable} upper bound $\widetilde g_i(\lambda)\geq g_i(\lambda)$. The condition $\widetilde g_i(\lambda)\leq0$ then implies $g_i(\lambda)\leq0$ and hence certifies $D(X_i;\lambda)\subseteq A(X_i,Y_i)$. We define the conservative \texttt{DISC} score as
\begin{equation}\label{eq:conservative_disc_score}
    \widetilde V_i
    =
    \inf\bigl\{
    \lambda\in\sR:
    \widetilde g_i(t)\leq0
    \text{ for every }t\geq\lambda
    \bigr\}.
\end{equation}
We assume that the conservative \texttt{DISC} score is also attained, in the sense that $\widetilde g_i(t)\leq0$ for every $t\geq\widetilde V_i$. If $\widetilde g_i$ is nonincreasing, the definition \eqref{eq:conservative_disc_score} reduces to $\widetilde V_i=\inf\{\lambda:\widetilde g_i(\lambda)\leq0\}$. {Moreover, every nestedness level admissible in \eqref{eq:conservative_disc_score} is also admissible for the original score, and hence \(\widetilde V_i\geq V_i\).
Further,} computing the empirical quantile in these conservative scores \(\widetilde V_i\) yields the subset $\widehat{D}^{\texttt{C-DISC}}(X_{n+1})=D(X_{n+1};\widetilde\lambda)$, where $\widetilde\lambda=\mathsf{Q}_{1-\alpha}(\{\widetilde V_i\}_{i=1}^n,\infty)$. Theorem~\ref{thm:inclusion_certificate} shows that exchangeability of the conservative scores preserves the finite-sample inclusion guarantee.

\begin{theorem}\label{thm:inclusion_certificate}
    If $\widetilde g_i(\lambda)\ge g_i(\lambda)$ for all $\lambda$, then $\left\{\widetilde V_{n+1} \leq \widetilde{\lambda}\right\} \subseteq \left\{\widehat D^{\texttt{C-DISC}}(X_{n+1}) \subseteq A(X_{n+1},Y_{n+1})\right\}$. 
    If the scores $\{\widetilde V_i\}_{i=1}^{n+1}$ are exchangeable, we have $\sP\LRl{\widehat{D}^{\texttt{C-DISC}}(X_{n+1}) \subseteq A(X_{n+1},Y_{n+1})}
        \geq 1-\alpha$.
\end{theorem}

{Appendix \ref{appen:Lagrangian_certificate} provides a general construction of \(\widetilde{g}_i(\lambda)\) for a class of nonconcave constraint functions through the Lagrangian dual formulation of \eqref{eq:certification_value}.}

\section{Efficiency-Aware Optimization of Safe Feasible Sets}\label{sec:family_optimization}


{Section~\ref{sec:method} establishes finite-sample validity for any fixed feasible set family, }
but does not distinguish among families with different downstream efficiency. Let $\phi:\gU\to\sR$ denote the downstream cost function, where smaller values indicate more efficient decisions. Given a safe feasible set $D(X)\subseteq\gU$, the corresponding decision solves the problem $\min_{u\in D(X)}\phi(u)$.
The optimal value measures the decision cost of $D(X)$ and therefore provides a natural criterion for optimizing the safe feasible set while controlling inclusion error.

We parameterize the candidate feasible set as $D_\theta(x;\lambda)=\{u\in\gU:L_{\theta,\lambda}(u;x)\le 0\}$,
where $L_{\theta,\lambda}$ is nondecreasing in $\lambda$ for each fixed $\theta$. {The parameter \(\theta\) controls the shape of the candidate set, whereas \(\lambda\) controls its nestedness level.} 
For example, the prediction-set family can use $L_{\theta,\lambda}(u;x)=\sup_{y\in C_{\theta}(x;\lambda)}f(u;x,y)$ with $C_{\theta}(x;\lambda)=\{y\in\gY:s_{\theta}(x,y)\leq\lambda\}$, while the additive residual-margin family can take $L_{\theta,\lambda}(u;x)=\hat f(u;x)+\mu_{\theta}(u)+\lambda\sigma_{\theta}(u)$.


\subsection{Optimized DISC via empirical risk minimization}

Applying \texttt{DISC} to this parameterized family yields the score
\begin{equation}\label{eq:opt_disc_score}
    V_\theta(X,Y)
    =
    \inf\left\{
        \lambda\in\sR:
        \sup_{u\in D_\theta(X;\lambda)}f(u;X,Y)\leq0
    \right\}.
\end{equation}
We assume that, with probability one, the infimum in \eqref{eq:opt_disc_score} is finite and attained simultaneously for all \(\theta\in\Theta\). Let \(\lambda_\theta^*=\inf\{q\in\sR:\sP\{V_\theta(X,Y)\leq q\}\geq1-\alpha\}\) denote the population \((1-\alpha)\)-quantile of \(V_\theta(X,Y)\), and let \(D_\theta^*(X)=D_\theta(X;\lambda_\theta^*)\) be the population subset.

\begin{definition}\label{def:oracle_optimal_model}
The population minimum risk over the feasible set family is
\(\Phi^*=\min_{\theta\in\Theta}\E[\min_{u\in D_\theta^*(X)}\phi(u)]\).
\end{definition}

{The population benchmark first specifies the safe feasible set using its own population score quantile and then selects the family with the smallest expected downstream cost.}
Since the population quantiles and risks are unknown, we replace them with empirical counterparts. For each \(\theta\in\Theta\), we take \(\hat\lambda_\theta=\mathsf{Q}_{1-\alpha}(\{V_\theta(X_i,Y_i)\}_{i=1}^n,\infty)\), and optimize the family parameter by solving the following empirical risk minimization (ERM) problem
\begin{align}\label{eq:empirical_model_selection}
    \hat\theta
    \in
    \argmin_{\theta\in\Theta}
    \frac{1}{n}\sum_{i=1}^n
    \min_{u\in D_\theta(X_i;\hat\lambda_\theta)}
    \phi(u).
\end{align}
We refer to this procedure as \texttt{O-DISC}, and define its test-time subset by $\widehat D^{\texttt{O-DISC}}(X_{n+1})=D_{\hat\theta}(X_{n+1};\hat\lambda_{\hat\theta})$. When $\Theta$ is finite, \eqref{eq:empirical_model_selection} can be solved by {comparing all candidate families}. 
For continuous $\Theta$, Section~\ref{sec:dc_cvar_erm} reformulates the ERM as a tractable constrained optimization problem.

{Unlike the fixed-family construction in Section~\ref{sec:dir_inclusion}, \texttt{O-DISC} uses the same labeled data both to select \(\hat\theta\) and to calibrate the corresponding nestedness level. Consequently,}
the optimized test score $V_{\hat\theta}(X_{n+1},Y_{n+1})$ is not generally exchangeable with the optimized calibration scores $\{V_{\hat\theta}(X_i,Y_i)\}_{i=1}^n$, and the exact finite-sample guarantee in Theorem~\ref{thm:inclusion} does not directly extend to \(\widehat D^{\texttt{O-DISC}}(X_{n+1})\). Section~\ref{sec:odisc_theory} bounds the resulting inclusion error, while Section~\ref{sec:augmented_erm} introduces an augmented ERM procedure that restores exact validity at higher computational cost.

\subsection{Reformulating the ERM via constrained optimization}\label{sec:dc_cvar_erm}

For a continuous parameter space $\Theta$, the conformal sample quantile $\hat\lambda_\theta$ makes the ERM problem \eqref{eq:empirical_model_selection} nonsmooth in $\theta$. Moreover, the value $\min_{u\in D_{\theta}(X_i;\hat\lambda_\theta)}\phi(u)$ has an implicit relation with $\theta$.
{To make these dependencies explicit, we adapt the lifting reformulation in the optimization literature \citep{balas2005projection} to the problem \eqref{eq:empirical_model_selection}. Specifically, we treat the nestedness level $\lambda$ and individual decisions $u_1,\ldots,u_n$ as joint optimization variables, replacing the inner minimizations with explicit feasibility constraints and expressing the quantile condition as $\hat\lambda_\theta\leq\lambda$.} Since $D_\theta(X;\lambda)$ decreases as $\lambda$ increases, this yields the following equivalent formulation:
\begin{equation}\label{eq:erm_quantile_constraint}
    \begin{aligned}
    \min_{\substack{\theta\in\Theta,\ \lambda\in\sR,\\ u_1,\ldots,u_n\in\gU}}\  \frac1n\sum_{i=1}^n \phi(u_i) \quad \mathrm{s.t.}\quad L_{\theta,\lambda}(u_i;X_i)\le 0,\ i\in[n],\quad \hat\lambda_\theta\le \lambda\,.
    \end{aligned}
\end{equation}
The first set of constraints enforces $u_i\in D_\theta(X_i;\lambda)$, while the second requires the chosen nestedness level to be no smaller than the conformal sample quantile. Because $D_\theta(X_i;\lambda)$ shrinks as $\lambda$ increases, the empirical downstream value $\min_{u\in D_\theta(X_i;\lambda)}\phi(u)$ is nondecreasing in $\lambda$. Thus, for each fixed $\theta$, an optimizer can always be chosen with $\lambda=\hat\lambda_\theta$.

To characterize the quantile constraint in \eqref{eq:erm_quantile_constraint}, 
define
\[
g_i(\lambda;\theta)=\sup_{u\in\gU} f(u;X_i,Y_i) \quad \text{s.t.}\quad L_{\theta,\lambda}(u;X_i)\le0.
\]
And let \(g_{(1)}(\lambda;\theta)\le\cdots\le g_{(n)}(\lambda;\theta)\) be their order statistics.
By definitions of $V_\theta(X_i,Y_i)$ and $g_i(\lambda;\theta)$, 
we know \(g_i(\lambda;\theta)\le0\) is equivalent to \(V_\theta(X_i,Y_i)\le\lambda\). 
Let \(k=\lceil(n+1)(1-\alpha)\rceil\) and assume \(k\le n\). Then \(\hat\lambda_\theta\le\lambda\) in \eqref{eq:erm_quantile_constraint} holds if and only if 
\(g_{(k)}(\lambda;\theta)\le0\). 
{Using the difference-of-top-sums representation: \(g_{(k)}(\lambda;\theta)=\sum_{j=k}^n g_{(j)}(\lambda;\theta)-\sum_{j=k+1}^n g_{(j)}(\lambda;\theta)\), where the second sum is zero when \(k=n\), {we obtain the equivalent constrained problem} 
\begin{equation}\label{eq:dc_cvar_erm}
    \begin{aligned}
    \min_{\substack{\theta\in\Theta,\ \lambda\in\sR,\\ u_1,\ldots,u_n\in\gU}}\quad
    \frac1n\sum_{i=1}^n \phi(u_i)\quad
    \mathrm{s.t.}\quad
    & L_{\theta,\lambda}(u_i;X_i)\le 0,\quad i\in[n],\\
    & \sum_{j=k}^n g_{(j)}(\lambda;\theta)
    -\sum_{j=k+1}^n g_{(j)}(\lambda;\theta)\le 0.
    \end{aligned}
\end{equation}
This formulation avoids the binary variables required by a direct encoding of the number of nonpositive verification values.
Both \(\sum_{j=k}^n g_{(j)}(\lambda;\theta)\) and \(\sum_{j=k+1}^n g_{(j)}(\lambda;\theta)\) are convex and piecewise-linear functions \citep{rockafellar2000optimization}} of $g_i(\lambda;\theta)$.
If each \(g_i(\lambda;\theta)\) is differentiable in \((\lambda,\theta)\), \eqref{eq:dc_cvar_erm} can be solved using generalized subgradient-based methods \citep{bolte2021nonsmooth}. Implementation details are deferred to Appendix~\ref{appen:SCA_erm}. The next proposition formalizes the equivalence between this constrained problem and the original ERM.


\begin{proposition}\label{prop:dc_cvar_equivalence}
The constrained problem \eqref{eq:dc_cvar_erm} has the same minimum objective value and the same optimal family parameters $\hat\theta$ as the original ERM problem \eqref{eq:empirical_model_selection}. 
\end{proposition}

\subsection{Theoretical results for the optimized safe feasible set}\label{sec:odisc_theory}

Optimizing the set family and calibrating the nestedness parameter on the same labeled data raises two statistical questions: how much inclusion validity may be lost and how far the selected family may lie from the population benchmark. This section derives the corresponding inclusion-error and excess-risk bounds.





\begin{assumption}\label{ass:threshold_complexity_inclusion}
The threshold class $\gH=\left\{ \mathbbm{1}\{V_\theta(x,y)\le t\}: \theta\in\Theta,\ t\in\sR \right\}$ has a Vapnik–Chervonenkis (VC) dimension $d_{\gH}$.
\end{assumption}
The threshold class $\gH$ controls the uniform estimation error between the empirical and population distribution functions of $V_\theta(X,Y)$. 


\begin{theorem}
\label{thm:direct_concentration_inclusion}
Under Assumption \ref{ass:threshold_complexity_inclusion}, for a universal constant $\mathfrak C>0$, with probability at least $1-\eta$ over the randomness of labeled data $\gD_n$, we have
\begin{gather*}
\sP\!\left\{ \widehat D^{\texttt{O-DISC}}(X_{n+1}) \subseteq A(X_{n+1},Y_{n+1}) \,\mid\, \gD_n\right\}\ge 1-\alpha-\mathfrak C \sqrt{\frac{d_{\gH}\log(en/d_{\gH})+\log(2/\eta)}{n}}.
\end{gather*}
\end{theorem}

Theorem \ref{thm:direct_concentration_inclusion} gives a distribution-free guarantee whose proof applies the VC inequality to the class $\gH$: uniformly over $\theta\in\Theta$ and $t\in\sR$, the empirical distribution of $V_\theta(X,Y)$ is close to its population distribution. Because the bound holds uniformly over the candidate families, it remains valid after the data-dependent selection of $\hat\theta$. 
Appendix \ref{appen:stability_bound} gives 
a potentially sharper inclusion-error bound under a replace-one stability condition, which is widely used in the recent conformal prediction literature \citep{bian2023training,liang2025algorithmic,Min2026enhanced}. 

We next evaluate the efficiency of the empirical optimizer in \eqref{eq:empirical_model_selection}, which is captured by the gap between the expected decision risk of the family selected by $\hat\theta$ and the population minimum risk $\Phi^*$ in Definition \ref{def:oracle_optimal_model}. 



\begin{assumption}\label{ass:uniform_quantile_regularity}
There exist positive constants $c_q$, $\epsilon_q$, $\Gamma$ and $B$. (i) For every $\theta\in\Theta$ and $0\le t\le \epsilon_q$, $\sP\{V_\theta(X,Y)\le\lambda_\theta^*+t\}\ge 1-\alpha+c_qt$ and $\sP\{V_\theta(X,Y)\le\lambda_\theta^*-t\}\le 1-\alpha-c_qt$. (ii) For every $\theta\in\Theta$ and $\lambda,\lambda'\in \sR$ satisfying $|\lambda-\lambda_\theta^*|\le \epsilon_q$ and $|\lambda'-\lambda_\theta^*|\le \epsilon_q$, we have $\left| \E\left[\min_{u\in D_\theta(X;\lambda)}\phi(u)\right] - \E\left[\min_{u\in D_\theta(X;\lambda')}\phi(u)\right] \right| \le \Gamma|\lambda-\lambda'|$, and $|\min_{u\in D_\theta(X;\lambda)}\phi(u)|\le B$. (iii) Define the function class $\gM_\delta=\{x\mapsto\min_{u\in D_\theta(x;\lambda)}\phi(u):\theta\in\Theta,\ |\lambda-\lambda_\theta^*|\le\delta\}$. For every $0<\delta\le\epsilon_q$, the pseudo-dimension \citep{anthony1999neural} of $\gM_\delta$ is at most $d_{\gM}$.


\end{assumption}

Assumption \ref{ass:uniform_quantile_regularity}(i) is a standard local regularity condition for quantile estimation. It holds when $V_\theta(X,Y)$ has a density uniformly bounded below near $\lambda_\theta^*$. Assumption~\ref{ass:uniform_quantile_regularity}(ii) controls how quantile-estimation error propagates to downstream cost and bounds the relevant decision values, while Assumption~\ref{ass:uniform_quantile_regularity}(iii) controls uniform convergence of the empirical objective. Further discussions are given in Appendices~\ref{appen:verification_lipschitz} and \ref{appen:verification_complexity}.

\begin{theorem}\label{thm:excess_risk_bound}
Suppose Assumptions~\ref{ass:threshold_complexity_inclusion} and
\ref{ass:uniform_quantile_regularity} hold.
There exists a universal constant $\mathfrak{C}>0$ such that, with probability at least $1-\eta$, if $\Delta_n(\eta) = \mathfrak{C}\sqrt{\frac{d_{\gH}\log(en/d_{\gH})+\log(4/\eta)}{n}}+\frac{2}{n} \leq c_q\epsilon_q$, then
\begin{align*}
&\E\left[
\min_{u\in\widehat D^{\texttt{O-DISC}}(X_{n+1})}
\phi(u)
\,\mid\,\gD_n
\right]
- \Phi^* \leq
2\mathfrak{C}B
\sqrt{
\frac{
d_{\gM}\log(en/d_{\gM})+\log(4/\eta)
}{n}
}
+
\frac{\Gamma}{c_q}\Delta_n(\eta),
\end{align*}
where $\Phi^*$ is the population minimum decision risk in Definition~\ref{def:oracle_optimal_model}.
\end{theorem}

Theorem \ref{thm:excess_risk_bound} characterizes the excess risk of the empirical minimizer \eqref{eq:empirical_model_selection}. The first term comes from uniform convergence of the downstream decision risk, while the second term arises from estimating the family-dependent population quantile. 

\subsection{Finite-sample inclusion guarantee via augmented ERM}
\label{sec:augmented_erm}


The inclusion gap above arises since $\gD_n$ is used for both family optimization and calibration. A full conformal procedure can restore exchangeability by augmenting the sample with each hypothesized test label and re-solving the augmented ERM \citep{liang2024conformal,bao2025optimal}. This recovers finite-sample validity at the cost of additional computation.

For each hypothesized label $y\in\gY$ and $\theta\in\Theta$, define the augmented sample quantile as $\hat\lambda_\theta^y = \mathsf{Q}_{1-\alpha} \left( \{V_\theta(X_i,Y_i)\}_{i=1}^n, V_\theta(X_{n+1},y) \right)$.
Using $\hat\lambda_\theta^y$ in the empirical objective over all $n+1$ augmented observations yields the label-specific optimizer:
\begin{gather*}
\hat\theta^y \in \argmin_{\theta\in\Theta} \frac{1}{n+1} \left\{ \sum_{i=1}^n \min_{u\in D_\theta(X_i;\hat\lambda_\theta^y)} \phi(u) + \min_{u\in D_\theta(X_{n+1};\hat\lambda_\theta^y)} \phi(u) \right\}\,.
\end{gather*}
The full-conformal optimized safe feasible set is
\begin{equation}
    \widehat D^{\texttt{FO-DISC}}(X_{n+1})=\bigcap_{y\in\gY}D_{\hat\theta^y}\left(X_{n+1};\hat\lambda_{\hat\theta^y}^y\right)\,.\nonumber
\end{equation}
For the true label $Y_{n+1}$, the augmented optimized parameter $\hat\theta^{Y_{n+1}}$ is obtained symmetrically from the $n+1$ exchangeable samples. 
Therefore, the inclusion scores $\{V_{\hat\theta^{Y_{n+1}}}(X_i,Y_i)\}_{i=1}^{n+1}$ are exchangeable, which yields the following inclusion guarantee.



\begin{theorem}
\label{thm:full_conformal_inclusion}
If $\{(X_i,Y_i)\}_{i=1}^{n+1}$ are exchangeable, then $\sP\{\widehat D^{\texttt{FO-DISC}}(X_{n+1}) \subseteq A(X_{n+1},Y_{n+1})\} \ge 1-\alpha$.
\end{theorem}

{The cost of exact validity is that \(\hat\theta^y\) and
\(\hat\lambda_{\hat\theta^y}^y\) have to be recomputed across all $y\in\gY$. When \(\Theta\) is finite, Appendix~\ref{appen:piecewise_constant_augmented_erm} exploits the piecewise-constant structure of the mapping $y\mapsto\hat\theta^y$ to partition $\gY$ into finitely many regions. A simpler alternative is sample splitting: one labeled subset selects $\theta$, and another performs the fixed-family \texttt{DISC} calibration. This approach also restores the finite-sample inclusion guarantee, but may sacrifice efficiency because each stage uses only part of the data. }

\section{Simulation Results}\label{sec:simulation}

We evaluate four fixed-family procedures built from the two margin function families in Section \ref{sec:instantiation}. \texttt{DISC-PS} applies directed inclusion calibration to the prediction-set margin family (\texttt{PS}) in \eqref{eq:dir_score_PS}, while \texttt{DISC-Add} applies it to the additive residual-margin family (\texttt{Add}) in \eqref{eq:dir_score_add}. The corresponding baselines \texttt{Base-PS} and \texttt{Base-Add} calibrate the nestedness level using the label coverage score in \eqref{eq:score_coverage_PS} and the uniformly standardized score in \eqref{eq:uniform_score}, respectively.

Across all simulations, we report two primary evaluation metrics. (i) The \textit{inclusion rate} is the empirical estimator of $ \sP\{D(X)\subseteq A(X,Y)\}$ computed from the test data. (ii) The \textit{decision cost} is the empirical estimator of $ \E[\min_{u\in D(X)}\phi(u)]$, 
where $\phi(u)$ is the task-specific decision cost function. The evaluation metrics are computed over 100 repetitions. In each repetition, an independent calibration sample \(\{(X_i,Y_i)\}_{i=1}^n\) is used to construct the safe feasible set, and an independent test sample \(\{(X_{n+j},Y_{n+j})\}_{j=1}^{n_{\rm test}}\) is used for evaluation.

\subsection{Chance-constrained resource allocation experiment}
Contextual resource allocation under uncertain demand is a continuous chance-constrained programming problem. 
The context is sampled as $ X\sim {\rm Unif}([-1,1]^5)$, and the demand $ Y=(Y_1,Y_2)$ is generated by $ Y_j=[m_j(X)+\varsigma_j(X)\xi_j]_+$, $ j=1,2$, where $ \xi_1,\xi_2$ are independent noise variables. The nonlinear mean function $ m_j$ and scale function $ \varsigma_j$ are reported in Appendix~\ref{appen:resource_dgp}. The demand predictor $ \mathsf{ML}(X)$ is a random forest regressor trained on an independent pre-training sample of size $5,000$, and held fixed across all repetitions. The test sample size is fixed at $n_{\rm test}=1,000$.

The decision variable is $ u=(u_1,u_2)\in[0,8]^2$, where $ u_j$ is the reserved resource for the $j$th demand component. We consider two constraint functions: the \textit{reserve-constraint function} is $f_{\rm res}(u;X,Y)=\max\left\{Y_1-u_1,\,Y_2-u_2,\,1.5(Y_1+Y_2)-(u_1+u_2)\right\}$; the \textit{buffered-constraint function} is $f_{\rm buf}(u;X,Y)=\max\left\{Y_1-u_1,\,Y_2-u_2,\,1.5(Y_1+Y_2)-(u_1+u_2),\,u_2-Y_2-3.5\right\}$.
The first $f_{\rm res}$ requires the allocation to cover both individual demands and an aggregate reserve requirement, whereas $f_{\rm buf}$ additionally imposes an upper buffer on $u_2$.


\paragraph*{Fixed margin function.}
For either constraint function $ f\in\{f_{\rm res},f_{\rm buf}\}$, the additive residual-margin family is constructed from a learned safety predictor $ \hat{f}$ as $ L_\lambda(u;X)=\hat{f}(u;X)+\lambda\sigma(u)$, where $ \sigma(u)=1+0.1(u_1+u_2)$ is fixed before calibration. The prediction-set margin family is defined as $ L_\lambda(u;X)=\sup_{y\in C(X;\lambda)} f(u;X,y)$, where $ C(X;\lambda)=\{y\in \sR^2:\|y-\mathsf{ML}(X)\|_2\le \lambda\}$. For \texttt{DISC-PS}, the score has a closed-form expression from the polyhedral resource constraints; for \texttt{DISC-Add}, the score can be computed exactly since $D(X;\lambda)$ is polyhedral, and both constraint functions are max-affine. We evaluate decision efficiency using a linear resource cost $ \phi_L(u)=u_1 + 1.25u_2$ and a quadratic resource cost $ \phi_Q(u)=0.1(u_1^2+1.5u_2^2)$. The linear cost represents unit resource prices, while the quadratic cost strongly penalizes over-allocation. Detailed computation is given in Appendix~\ref{appen:resource_score_computation}. 

Figure~\ref{fig:resource_loss} illustrates how the output safe feasible sets interact with the contours of two downstream cost functions under the reserve-constraint function $f_{\rm res}$. The \texttt{DISC} regions retain a larger low-cost portion of the oracle feasible set, whereas the baseline methods exclude many decisions that are in fact safe but fail to be included by the stronger events of label coverage or uniform residual bound. {As the downstream decision selects the lowest-cost point in the safe feasible set, the additional shrinkage of the baseline subset translates directly into a higher decision cost.}

\begin{figure}[H]
    \centering
    \includegraphics[width=.8\linewidth]{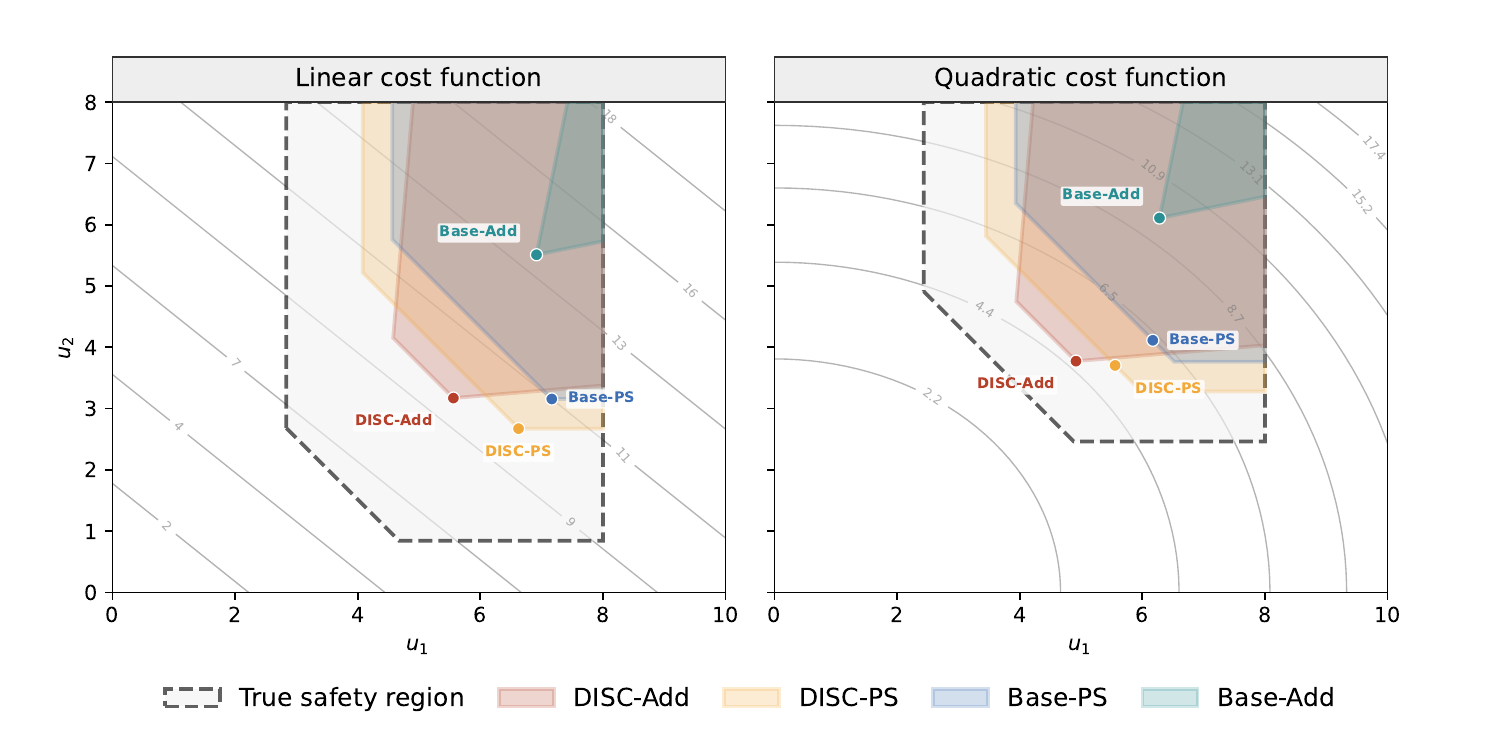}
    \caption{Safe feasible sets and oracle feasible sets for the fixed-family resource-allocation experiment under $ f_{\rm res}$, overlaid on the contours of linear and quadratic costs, where $ n=500$ and $ 1-\alpha = 90\%$. Colored dots refer to decision points.}
    \label{fig:resource_loss}
\end{figure}

{Figure~\ref{fig:resource_box} summarizes the results over repeated samples under both $f_{\rm res}$ and $f_{\rm buf}$}. 
Across all calibration sample sizes, the \texttt{DISC} methods attain inclusion rates close to the target level while achieving lower decision costs. The baseline methods produce substantially higher inclusion rates and are therefore more conservative. 
{The effect is particularly pronounced for \texttt{Base-Add} under \(f_{\rm res}\), whose inclusion rate is close to one. These findings agree with Propositions~\ref{prop:ps_scp_comparison} and \ref{prop:additive_unif_comparison}: even when the candidate family is fixed, controlling a stronger event may discard a substantial set of useful decisions.}

\begin{figure}[H]
    \centering
    \includegraphics[width=.9\linewidth]{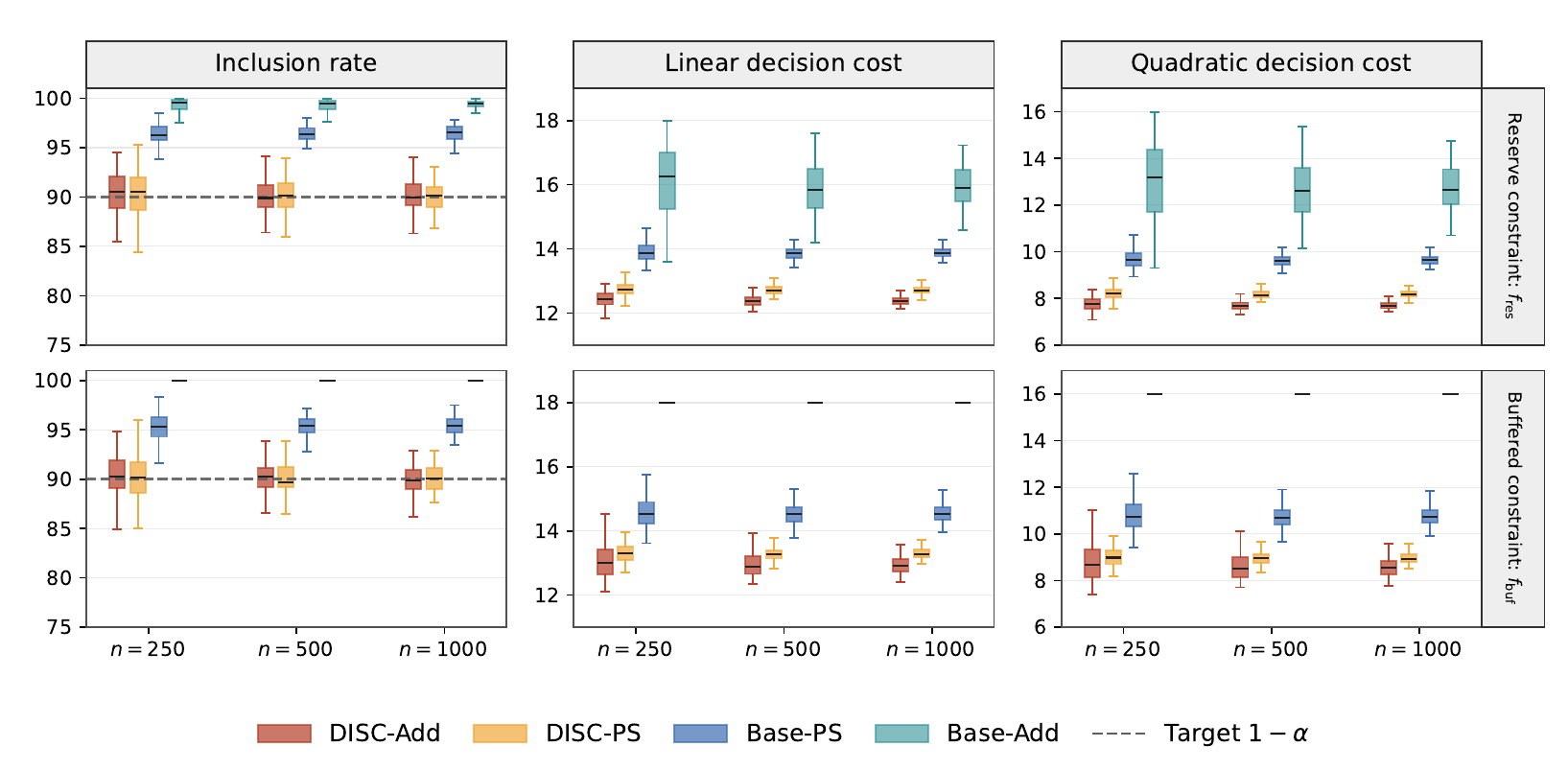}
    \caption{Empirical performance of each method under different calibration sample sizes in the resource allocation task for $f_{\rm res}$ and $f_{\rm buf}$, with target safety level $1-\alpha = 90\%$.}
    \label{fig:resource_box}
\end{figure}


\paragraph*{Optimized margin function.}
We next assess whether the procedure in Section~\ref{sec:family_optimization} can further improve downstream efficiency under one asymmetric linear objective $\phi_L(u)=u_1+10u_2$. {The optimized baseline procedures select \(\theta\) by minimizing the same empirical decision-loss objective over the same parameter spaces as their optimized \texttt{DISC} counterparts. They differ only in their calibration scores: baseline methods use the label-coverage score in \eqref{eq:score_coverage_PS}, or the uniform-residual score in \eqref{eq:uniform_score}, thereby retaining the stronger surrogate calibration targets of the corresponding fixed-family baselines.}


For the prediction-set margin family, we consider the parametric ellipsoidal prediction sets: $  C_\Sigma(X;\lambda) = \left\{y\in \sR^2: (y-\mathsf{ML}(X))^\top \Sigma (y-\mathsf{ML}(X))\le \lambda \right\},$ where $ \Sigma$ is a positive definite shape matrix. The concrete parameterization is given in Appendix~\ref{appen:resource_family_opt}. The fixed prediction-set family methods use the identity matrix $ \Sigma=I_2$, while the optimized methods select the shape matrix by solving the ERM \eqref{eq:empirical_model_selection} with $ \theta = \Sigma$ and $ \Theta = \{\Sigma = LL^{\top}: L \in \sR^{2\times 2}, \det(LL^{\top})=1\}$. Here we impose the constraint $ \det(\Sigma)=1$ to remove the scale non-identifiability between $ \Sigma$ and $ \lambda$.
For the additive residual-margin family, we consider a parameterized mean function 
while keeping a fixed scale function $\sigma(u)=1+0.1(u_1+u_2)$, i.e., \(L_{\theta, \lambda}(u;X) = \hat{f}(u;X)+\theta^\top u +\lambda \sigma(u)\) with a discrete space $\Theta = \{-2,-1,0,1,2\}^2$ {containing 25 two-dimensional parameters}.
The fixed-family procedures correspond to $ \theta=(0,0)^\top$. 

\begin{figure}[H]
    \centering
    \includegraphics[width=.8\linewidth]{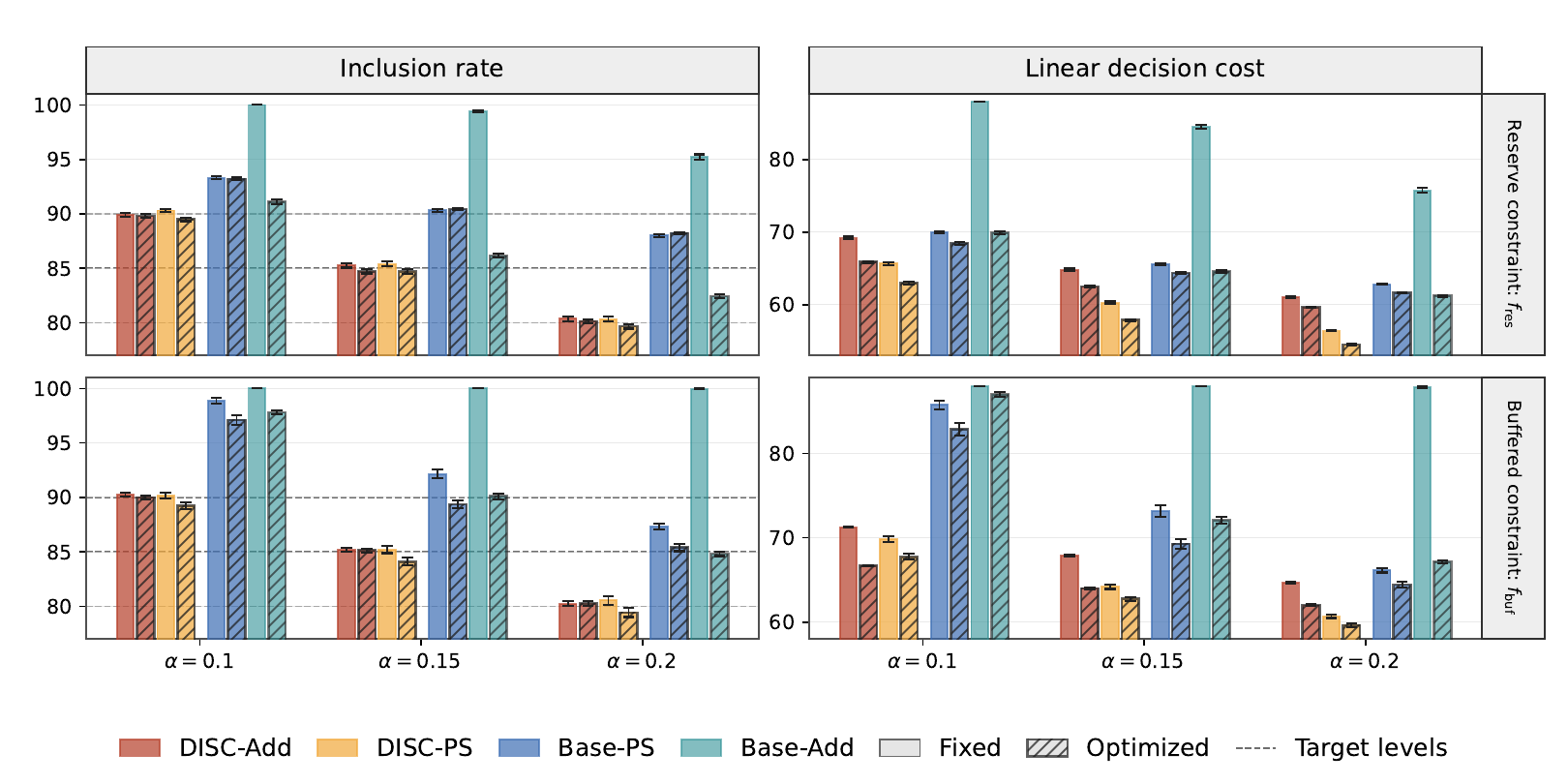}
    \caption{Empirical results of optimized margin function families in the resource-allocation task for $ f_{\rm res}$ and $ f_{\rm buf}$, with a linear decision cost function and $ n = 500$.}
    \label{fig:resource_opt}
\end{figure}

Figure~\ref{fig:resource_opt} shows that optimizing the margin function family effectively decreases the downstream cost in both \texttt{DISC} methods and \texttt{Base} methods. Under two margin families, the optimized \texttt{DISC} methods improve over their fixed-family counterparts while keeping inclusion close to the nominal target. The optimized \texttt{Base} methods also reduce cost relative to their fixed versions, but remain more conservative because their calibration targets are unchanged and are still stronger coverage or uniform residual bound events. 
{In sum, the experiments display inefficiency from two distinct sources: calibration of an unnecessarily strong surrogate event and an unsuitable candidate-family shape. Direct inclusion guarantee addresses the former, while family optimization addresses the latter.} Additional results of the \texttt{FO-DISC} safe feasible set are deferred to Appendix \ref{appen:add_results_resource}.


\subsection{Trajectory collision-avoidance experiment}
Trajectory-valued collision avoidance is motivated by safe-planning formulations for robots or autonomous vehicles moving among dynamic agents \citep{lindemann2023safe,dixit2023adaptive}, where a learned trajectory predictor is used to reason about future obstacle motion. This setting involves the following two trajectory components.

\noindent\textit{(1) Obstacle trajectory.} The unobserved label is the future trajectory of a surrounding traffic participant,
    $ Y=(Y_t)_{1\leq t\leq T}\in(\mathbb R^2)^T$ with $ T=10$. It is generated as $ Y_t=m_t(X)+\varepsilon_t(X)+J_t$ for $ t=1,\ldots,T$.
    Here $ m_t(X)$ is a smooth context-dependent mean path, $ \varepsilon_t(X)$ is heteroscedastic heavy-tailed smooth trajectory noise, and $ J_t$ is a rare lane-level excursion term. The excursion term can occasionally produce large prediction errors that do not necessarily imply a large safety risk within the ego-control box. The complete data-generating mechanism is given in Appendix~\ref{appen:trajectory_dgp}.

\noindent\textit{(2) Decision ego trajectory.} The decision variable is a two-dimensional control $u=(u_1,u_2)\in[-3,3]^2$, which characterizes an intermediate avoidance maneuver around a context-dependent ego route. Specifically, \(u_1\) controls an arch-shaped adjustment, while \(u_2\) controls the lateral avoidance component. Given a context $X$, the induced ego trajectory is $ Z(u;X)=\{Z_t(u;X)\}_{t=1}^T$ with $ Z_t(u;X)=a_t(X)+B_tu$. Here $ a_t(X)\in\mathbb R^2$ and $ B_t\in\mathbb R^{2\times2}$ are specified in Appendix~\ref{appen:trajectory_decision_setting}. The control basis functions in $ B_t$ vanish at both endpoints, so $ u$ affects only the intermediate maneuver. In particular, the \textit{nominal ego trajectory} is a zero-control route $ Z^{\rm nom}(X)=Z(0;X)$ when there are no obstacles.


When $ D(X)\neq\emptyset$, the decision control is obtained by $ u_D(X)\in\argmin_{u\in D(X)}\phi_Q(u)$ with $ \phi_Q(u)=0.2u_1^2+0.3u_2^2$, and the resulting decision ego trajectory is $ Z(u_D(X);X)$. If $ D(X)=\emptyset$, no trajectory is returned, and it receives a cost value of $4.5$. The quadratic cost measures the control effort needed to depart from the zero-control route, and the larger coefficient on \(u_2\) penalizes lateral maneuvers more heavily. It implies that lower quadratic cost corresponds to a safe trajectory that stays closer to the nominal trajectory. {In this simulation, we also report two additional metrics: the \textit{root-mean-square (RMS) to nominal} is the RMS distance between the decision ego trajectory and the nominal ego trajectory; the \textit{empty-set ratio} is the frequency of empty safe feasible sets.}

We consider two constraint functions: the \textit{average distance} constraint \(f_{\rm avg}(u;X,Y)=\omega-T^{-1}\sum_{t=1}^T\|Z_t(u;X)-Y_t\|_2^2\), and the \textit{minimum distance} constraint \(f_{\min}(u;X,Y)=\omega-\min_{1\le t\le T}\|Z_t(u;X)-Y_t\|_2^2\), where $\omega$ is the safety budget.
Before calibration, we fit a long short-term memory (LSTM) model $\{\mathsf{ML}_t(\cdot)\}_{1\leq t\leq T}$ to predict the obstacle trajectory $Y = (Y_t)_{1\leq t\leq T}$, in an independent pre-training sample of size $3000$. 
The fitted predictor is held fixed across all repetitions and methods. The prediction-set margin family uses the norm-based nonconformity scores. The additive residual-margin family is $ L_\lambda(u;X)=\hat{f}(u;X)+\lambda\sigma(u)$, where $\sigma(u)=1+0.2(u_1+1)^2+0.05(u_2-0.5)^2$. Because exact computation of the \texttt{DISC-Add} and \texttt{DISC-PS} scores requires solving nonconvex optimization problems, we use the conservative certification method in Section~\ref{sec:score_computation}; its implementation is described in Appendix~\ref{appen:trajectory_box_certificate}. {Since the certificate upper bounds the exact \texttt{DISC} score, Theorem~\ref{thm:inclusion_certificate} preserves the finite-sample inclusion guarantee, although the resulting subsets may be smaller than those obtained from exact scores.}

\begin{figure}[ht]
    \centering
    \includegraphics[width=.8\linewidth]{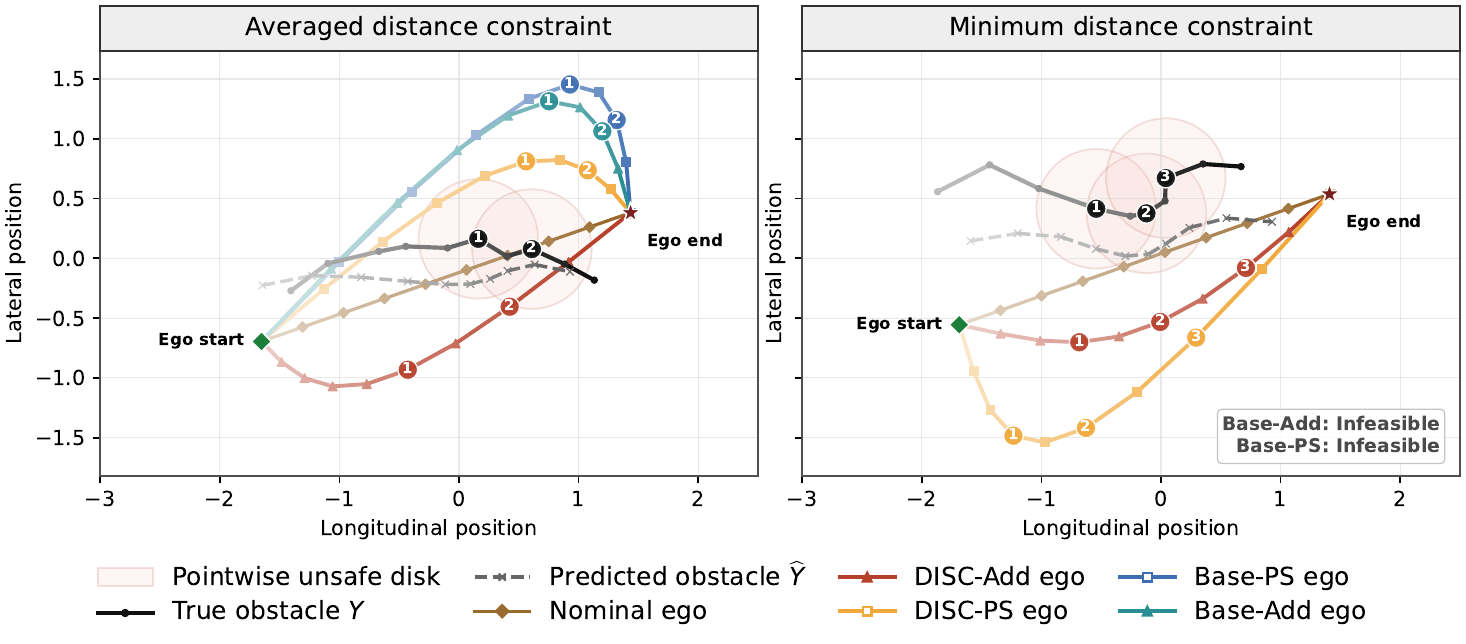}
    \caption{Decision ego trajectories of four methods with the safety budget $\omega = 0.25$. Numbered disks mark the displayed pointwise unsafe regions $ \{z\in \sR^2: \|z-Y_t\|_2^2 \leq \omega\}$ along the true obstacle trajectory, and matching numbers on the ego trajectories indicate the corresponding time indices. }
    \label{fig:collision_trajectory_paths}
\end{figure}

Figure~\ref{fig:collision_trajectory_paths} illustrates the decision ego trajectories produced by the four methods under both safety settings. Under the average-distance constraint function $f_{\rm avg}$, the safe feasible sets of all methods are nonempty. The \texttt{DISC} methods select ego trajectories that stay closer to the nominal path, while the baseline methods take more conservative detours. 
Under the minimum-distance constraint function $f_{\rm min}$, 
the \texttt{DISC} methods still return feasible low-cost avoidance maneuvers, whereas two baseline subsets are empty. 

\begin{figure}[H]
    \centering
    \includegraphics[width=.9\linewidth]{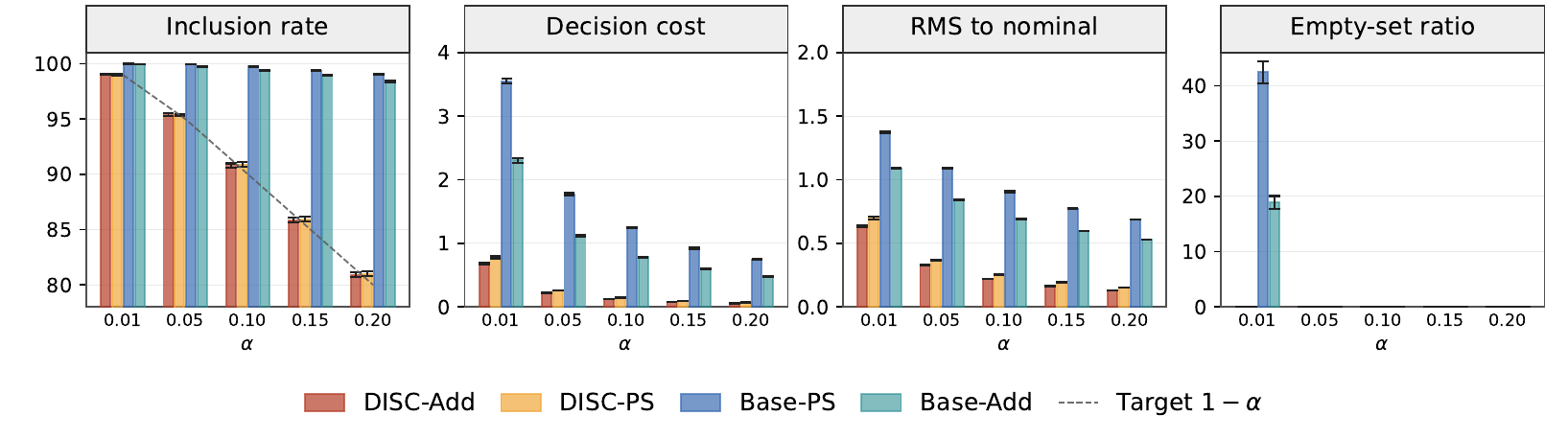}
    \caption{Sensitivity to the safety level $ \alpha$ in the average-distance safety setting with $ n=500$, and $ \omega=0.25$. Empty sets receive cost $ 4.5$; RMS is averaged over nonempty sets.} 
    \label{fig:collision_trajectory_barplot}
\end{figure}

\begin{figure}[H]
    \centering
    \includegraphics[width=.9\linewidth]{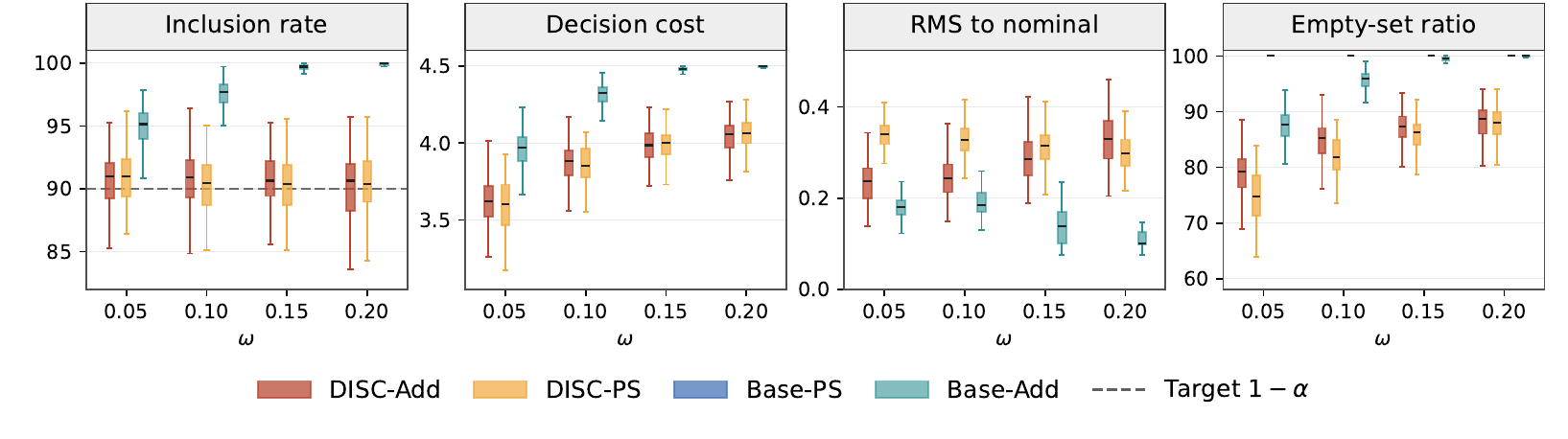}
    \caption{Sensitivity to the safety budget $ \omega$ in the minimum-distance safety setting with $ n=200$ and $ \alpha=0.1$. \texttt{Base-PS} is omitted because it always returns an empty subset. }
    \label{fig:trajectory_min_omega_sensitivity_main}
\end{figure}

Figure~\ref{fig:collision_trajectory_barplot} summarizes the performance metrics in the average-distance setting as the target safety level varies. The proposed \texttt{DISC} tracks the nominal inclusion targets across different values of $\alpha$, while the baseline methods are substantially more conservative, especially \texttt{Base-PS}. {Correspondingly, direct inclusion calibration also leads to lower decision costs and trajectories closer to the nominal route.} {Two baseline methods return empty subsets when $\alpha = 0.01$, whereas both \texttt{DISC} methods remain nonempty.} Figure~\ref{fig:trajectory_min_omega_sensitivity_main} examines the more stringent minimum-distance constraint function as the safety budget \(\omega\) increases. The \texttt{DISC} methods maintain inclusion rates near the nominal level while producing lower-cost nonempty decisions more frequently. {The empty-set ratios further show that the \texttt{DISC} methods yield nonempty feasible subsets more often than \texttt{Base-Add} across all values of $\omega$, and \texttt{Base-PS} is always infeasible.} These results demonstrate that the efficiency gain from direct inclusion calibration persists in a structured trajectory problem even when the \texttt{DISC} scores are computed through conservative certificates.

\section{Real Data Applications}\label{sec:real_data}

In this section, we evaluate \texttt{DISC} in two real-data applications: LLM reasoning-chain verification and safe object detection, with discrete decision spaces. 

\subsection{LLM reasoning-chain verification}\label{sec:exp_llm_reasoning}

We study reasoning-chain verification on the MATH benchmark \citep{hendrycksmath2021}. Qwen-Flash is used to generate a chain-of-thought response and to produce claim-level verifier confidence scores. 
Each record contains a problem statement, an official solution, a benchmark-provided difficulty label, one generated response, its ordered subclaims, an LLM-generated dependency graph, and claim-wise acceptability labels assigned by a fixed reference-judge pipeline. After preprocessing, the experiment contains $12{,}500$ valid problems. The inclusion target is defined relative to these fixed reference-judge labels and cannot be interpreted as a guarantee of objective mathematical correctness.

Following \citet{rubin2025conformal}, we preprocess each response as a reasoning graph; Appendix~\ref{appen:llm_graph_construction} gives the implementation details and prompts. For a prompt $X$, let $(\mathsf{LLM}_t(X))_{1\leq t\leq T(X)}$ denote the ordered $T(X)$ subclaims produced by the fixed LLM pipeline and let $G(X)$ be their dependency graph. For $\mathcal S\subseteq[T(X)]$, write $u_{\mathcal S}=(\mathsf{LLM}_t(X))_{t\in\mathcal S}$ for the output retaining the subclaims indexed by $\mathcal S$. {Let $\operatorname{Anc}(t;G(X))$ denote the set of strict ancestors of node $t$ in $G(X)$}. The set $\mathcal S$ is dependency closed if $\operatorname{Anc}(t;G(X))\subseteq\mathcal S$ for every $t\in\mathcal S$, and $\gU(X)$ collects all such outputs, including $u_{\emptyset}$. Let $Y=(Y_t)_{1\leq t\leq T(X)}$ be the reference-judge labels, where $Y_t=1$ indicates that $\mathsf{LLM}_t(X)$ is acceptable. {For $u_{\mathcal S}\in\gU(X)$, define $f(u_{\mathcal S};X,Y)=1-2\min_{t\in\mathcal S}Y_t$, with the convention $\min_{t\in\emptyset}Y_t=1$}. Hence, $A(X,Y)=\{u_{\mathcal S}\in\gU(X):f(u_{\mathcal S};X,Y)\leq0\}=\{u_{\mathcal S}\in\gU(X):Y_t=1\text{ for all }t\in\mathcal S\}$ is the oracle feasible set consisting of dependency-closed outputs whose retained subclaims are all acceptable. 

Given the prompt, generated subclaims, and dependency graph, the Qwen-Flash verifier produces a confidence score $\widehat Y_t(X)\in[0,1]$ for each subclaim $\mathsf{LLM}_t(X)$, with larger values indicating greater confidence that the subclaim is acceptable. Replacing $Y_t$ in the constraint function with $\widehat Y_t(X)$ gives the pretrained surrogate constraint function $\hat f(u_{\mathcal S};X)=1-2\min_{t\in\mathcal S}\widehat Y_t(X)$, where we retain the convention $\min_{t\in\emptyset}\widehat Y_t(X)=1$. Using $\hat f(u_{\mathcal S};X)$, we construct the additive residual-margin family $D(X;\lambda)=\{u_{\emptyset}\}\cup\{u\in\gU(X):\hat f(u;X)+\mu(u;X)+\lambda\leq0\}$, where $\mu(u;X)$ is an additive adjustment specified below.


{The constant-offset DISC method sets $\mu(u;X)\equiv0$ and recovers the rule of \citet{rubin2025conformal}; we refer to this method as \texttt{DISC-Const}}. The three adaptive DISC variants incorporate the benchmark difficulty label through $\mu(u;X)$, with \texttt{DISC-Mean} depending only on difficulty and \texttt{DISC-Linear} and \texttt{DISC-Sqrt} additionally depending on the retained-subset size. For comparison, the corresponding uniform-residual procedures are denoted by \texttt{Base-Const}, \texttt{Base-Mean}, \texttt{Base-Linear}, and \texttt{Base-Sqrt}. Their complete definitions and calibration scores are given in Appendix~\ref{appen:llm_cot_details}. For each $\alpha\in\{0.05,0.10\}$, we repeat the experiment over $100$ random splits with $10{,}000$ problems for calibration and $2{,}500$ for testing.

For each safe feasible set, the downstream output is chosen to retain an admissible dependency-closed subset of maximum cardinality. We report the empirical inclusion rate together with three retention metrics: mean retained length, mean retained fraction, and probability of retaining the final claim. Appendix~\ref{appen:llm_evaluation_metrics} specifies the exact definitions. As a reference, all subclaims in the unfiltered response are acceptable for $81.29\%$ of the problems under the same reference-judge pipeline.

\begin{table}[htbp]
    \centering
    \caption{LLM reasoning-chain verification on MATH at $\alpha=0.1$. Boldface and underlining indicate the largest and second-largest values among the retention metrics.}
    \label{tab:llm_reasoning_results}
    \begin{tabular}{lcccc}
        \toprule
        Method & Inclusion & Mean len $\uparrow$ & Mean frac $\uparrow$ & Final kept $\uparrow$ \\
        \midrule
        \texttt{DISC-Const} & 0.9036 & 9.89 & 0.8293 & 0.6786 \\
        \texttt{DISC-Mean} & 0.9032 & 9.90 & 0.8303 & 0.6796 \\
        \texttt{DISC-Linear} & 0.8994 & \textbf{10.57} & \textbf{0.8790} & \textbf{0.8770} \\
        \texttt{DISC-Sqrt} & 0.8996 & \underline{10.53} & \underline{0.8750} & \underline{0.8399} \\
        \bottomrule
    \end{tabular}
\end{table}

Table~\ref{tab:llm_reasoning_results} reports the results for $\alpha=0.1$; the results for $\alpha=0.05$ are provided in Table~\ref{tab:llm_reasoning_resultsp05}. All four displayed methods maintain empirical inclusion rates near the nominal target while producing nontrivial outputs. The fixed-family \texttt{DISC-Const} and the mean-adjusted \texttt{DISC-Mean} perform similarly. In contrast, \texttt{DISC-Linear} and \texttt{DISC-Sqrt} retain more claims and preserve the final claim more frequently. 
The uniform-residual baselines in Appendix~\ref{appen:llm_cot_details} are more conservative and return the empty output in nearly all repetitions. Thus, our method can produce more informative reasoning outputs in this application.


\subsection{Safe object detection}\label{sec:safe_obj_detect}

In this section, we study safe object detection using the COCO 2017 validation set \citep{lin2015microsoft} and a fixed pretrained YOLOv11n detector \citep{yolo11_ultralytics}. For an image $X_i$, let $Y_{ij}=(b_{ij},a_{ij})$, $j=1,\ldots,T(X_i)$, denote $T(X_i)$ target objects, where $b_{ij}\in[0,1]^4$ is the normalized ground-truth bounding box and $a_{ij}$ is the class label. We conduct calibration and evaluation at the object level, treating each pair $(X_i,Y_{ij})$ as one object-level unit. 
Here we consider pedestrians, vehicles, and traffic-related facilities and retain target objects with normalized area at least $0.001$, leaving $3{,}205$ images containing at least one retained object. Appendix~\ref{appen:safe_obj_detect} gives the complete implementation details.

For a generic object-level pair $(X,Y)$ with $Y=(b,a)$, the detector YOLOv11n returns at most $100$ target-class predictions $((\widehat b_k(X),\widehat c_k(X),\widehat a_k(X)))_{k=1}^{K(X)}$, where $K(X)\leq100$, $\widehat b_k(X)\in[0,1]^4$ is the $k$th predicted box, $\widehat c_k(X)\in[0,1]$ is its confidence score, and $\widehat a_k(X)$ is its predicted class label. The predictions are ordered so that $\widehat c_1(X)\geq\cdots\geq\widehat c_{K(X)}(X)$. For $k\in[K(X)]$, the decision $u_{[k]}(X)$ retains the first $k$ predictions, and $\gU(X)=\{u_{[k]}(X):k\in[K(X)]\}$ is the candidate-output family. For a prescribed coverage fraction $\gamma\in(0,1)$, define $\operatorname{cover}(u_{[k]}(X),Y)=\max_{\ell\leq k:\,\widehat a_\ell(X)=a}|b\cap\widehat b_\ell(X)|/|b|$, taking the maximum over an empty set as zero, where $|b|$ is the area of box $b$. The object-level constraint function is $f(u_{[k]}(X);X,Y)=\gamma-\operatorname{cover}(u_{[k]}(X),Y)$, and hence $A(X,Y)=\{u\in\gU(X):f(u;X,Y)\leq0\}$ is the oracle feasible set consisting of confidence-ordered prefixes that sufficiently cover the target object. 

Define a surrogate safety function by $\hat f(u_{[k]}(X);X)=\min_{\ell\leq k}\widehat c_\ell(X)=\widehat c_k(X)$.
For an offset $\mu(X)$, we consider $D(X;\lambda,\mu)=\{u_{[k]}(X)\in\gU(X):\hat f(u_{[k]}(X);X)+\mu(X)+\lambda\leq0\}\cup\{u_{[K(X)]}(X)\}$. This candidate feasible set is a tail of the ordered prefix family, and its shortest member is the reported output. Subject to object-level inclusion, shorter reported prefixes and smaller box-union areas are preferred. We compare \texttt{DISC} and \texttt{Base}, using the directed object-level inclusion score and a uniform residual score over the same prefixes, respectively. For each calibration strategy, we consider two offset specifications: the constant version sets $\mu(X)\equiv0$, whereas the trained version estimates $\widehat\mu(X)$ from detector-derived features on an independent pre-training split. The resulting methods are \texttt{DISC-Const}, \texttt{DISC-Train}, \texttt{Base-Const}, and \texttt{Base-Train}.

\begin{table}[htbp]
    \centering
    \caption{Safe object detection at $\gamma=0.6$ and $\alpha=0.1$, averaged over $100$ random splits. Boldface and underlining mark the smallest and second-smallest values in the retention and area columns.}
    \label{tab:object_detection_gamma06}
    \begin{tabular}{lcccccc}
        \toprule
        Method & Inclusion & Mean Ret $\downarrow$ & $Q_{0.25}^{\rm Ret}\ \downarrow$ & $Q_{0.50}^{\rm Ret}\ \downarrow$ & $Q_{0.75}^{\rm Ret}\ \downarrow$ & Area $\downarrow$ \\
        \midrule
        \texttt{DISC-Const} & 0.900 & \textbf{12.41} & 2.96 & \underline{6.15} & \underline{16.64} & \underline{0.379} \\
        \texttt{DISC-Train} & 0.900 & \underline{17.33} & \textbf{2.00} & \textbf{3.33} & \textbf{10.50} & \textbf{0.375} \\
        \texttt{Base-Const} & 0.952 & 82.79 & 72.18 & 100.00 & 100.00 & 0.576 \\
        \texttt{Base-Train} & 0.951 & 50.07 & \underline{2.26} & 41.03 & 100.00 & 0.447 \\
        \bottomrule
    \end{tabular}
\end{table}

The $3{,}205$ images are split into disjoint pre-training, calibration, and test sets in a $7{:}2{:}1$ ratio. All object-level pairs associated with the same image are assigned to the same split. We set $\alpha=0.1$ and average over $100$ random splits. Table~\ref{tab:object_detection_gamma06} reports the results at $\gamma=0.6$. The empirical inclusion rate is evaluated over eligible image--object pairs in the test images. In addition to object-level inclusion, we report the mean and quartiles of retained-prefix length and the normalized union area. Both \texttt{DISC} methods attain inclusion rates of $0.900$. The mean retained-prefix lengths of \texttt{DISC-Const} and \texttt{DISC-Train}, $12.41$ and $17.33$, are substantially smaller than those of the corresponding uniform-residual baselines, $82.79$ and $50.07$. The \texttt{DISC} procedures also produce 
smaller box-union areas. Although \texttt{DISC-Train} has a larger mean retained length than \texttt{DISC-Const}, it achieves lower retained-length quartiles and a slightly smaller union area, suggesting that it returns shorter prefixes for most images but has a heavier upper tail. Results for $\gamma=0.5$ and $0.7$ appear in Appendix~\ref{appen:safe_obj_detect}.

\section{Concluding Remarks}\label{sec:conclusion}

{This paper formulates feasible set inclusion as a setwise conformal inference problem and proposes \texttt{DISC} to control the probability of this target directly. For a fixed nested family, \texttt{DISC} converts set inclusion into a scalar critical threshold and yields a finite-sample, distribution-free marginal guarantee under exchangeability. In particular, we show that \texttt{DISC} is no more conservative than the corresponding label coverage and uniform-residual calibration rules for the prediction-set and additive residual-margin families, respectively. 
The numerical studies show that our method maintains the inclusion rate near the target safety level while producing larger or more useful decision regions and lower decision costs. 

Several extensions follow naturally from the present framework. First, Appendix~\ref{appen:optimal_safety_subset} characterizes a theoretically optimal safe feasible set, but turning this design into a practical method requires efficient estimation of conditional distributions. Also, the guarantees established here are marginal over the joint distribution of $(X,Y)$, and so it warrants further study of feasible conditional inclusion guarantees.} 






\phantomsection\label{supplementary-material}
\bigskip

\begin{center}

{\large\bf SUPPLEMENTARY MATERIAL}

\end{center}

The supplementary material contains the computation and optimization details of \texttt{DISC} scores, proofs of theoretical results, and deferred numerical settings and results.


\setstretch{.8}
\setlength{\bibsep}{4pt}
\putbib
\end{bibunit}

\newpage
\appendix
\begin{bibunit}
\setbibunitprefix{-appendix}
\setcounter{table}{0}
\renewcommand{\thetable}{S\arabic{table}}
\setcounter{figure}{0}
\renewcommand{\thefigure}{S\arabic{figure}}
\setcounter{lemma}{0}
\renewcommand{\thelemma}{S\arabic{lemma}}
\setcounter{theorem}{0}
\renewcommand{\thetheorem}{S\arabic{theorem}}

\numberwithin{equation}{section}
\allowdisplaybreaks
\spacingset{1.7}

\begin{center}
    {\LARGE\bf Supplementary Material for ``Conformalized Safe Feasible Sets in Uncertain Decision Systems''}
\end{center}

\section{Optimization and computation}

\subsection{Computation of prediction-set margins}
\label{appen:score_computation}

We give explicit formulations of
\(L_\lambda(u;X)=\sup_{y\in C(X;\lambda)}f(u;X,y)\).
Throughout this subsection, \(\gY=\sR^q\), the prediction sets under
consideration are nonempty and compact, and their fitted parameters
are fixed before calibration. Computational tractability depends on
both the geometry of \(C(X;\lambda)\) and the dependence of \(f\) on
the label \(y\).

\paragraph*{Affine and max-affine constraint functions.}
Suppose
\[
    f(u;X,y)=\max_{1\le k\le K}
    \{a_k(u;X)^\top y+b_k(u;X)\}.
\]
The affine case corresponds to \(K=1\). Define the support function
\(h_C(a)=\sup_{y\in C}a^\top y\). Interchanging the supremum with
the finite maximum gives
\begin{equation}
    L_\lambda(u;X)
    =\max_{1\le k\le K}
    \{b_k(u;X)+h_{C(X;\lambda)}(a_k(u;X))\}.
    \label{eq:ps_margin_support}
\end{equation}
Thus each affine branch requires one support-function evaluation.
The following expressions are standard robust counterparts
\citep{ben2002robust,boyd2004convex}; we suppress dependence on
\((u,X)\) in \(a_k,b_k\) and on \(X\) in the fitted set parameters.

\paragraph*{Ellipsoidal prediction sets.}
Let
\[
    C(X;\lambda)
    =\{y:(y-\mu)^\top\Sigma(y-\mu)\le r_\lambda^2\},
    \qquad \Sigma\succ0,
\]
where \(r_\lambda\ge0\) is nondecreasing in \(\lambda\). Substituting
\(y=\mu+r_\lambda\Sigma^{-1/2}z\), \(\|z\|_2\le1\), yields
\[
    L_\lambda(u;X)
    =\max_{1\le k\le K}
    \left\{b_k+a_k^\top\mu
    +r_\lambda\sqrt{a_k^\top\Sigma^{-1}a_k}\right\}.
\]
Hence \(D(X;\lambda)\) is obtained by requiring every expression
inside the maximum to be nonpositive. For the squared Mahalanobis
score \(s(X,y)=(y-\mu)^\top\Sigma(y-\mu)\),
\(r_\lambda=\sqrt\lambda\); for its square root,
\(r_\lambda=\lambda\), in both cases with \(\lambda\ge0\).

\paragraph*{Box prediction sets.}
For fixed coordinate scales \(w_j>0\), consider
\[
    C(X;\lambda)
    =\{y:|y_j-\mu_j|\le r_\lambda w_j,\ j=1,\ldots,q\}.
\]
Maximization separates across coordinates, giving
\[
    L_\lambda(u;X)
    =\max_{1\le k\le K}
    \left\{b_k+a_k^\top\mu
    +r_\lambda\sum_{j=1}^q w_j|a_{kj}|\right\}.
\]
The score \(s(X,y)=\max_j|y_j-\mu_j|/w_j\) produces this family
with \(r_\lambda=\lambda\ge0\).

\paragraph*{Polyhedral prediction sets.}
Let \(C(X;\lambda)=\{y:Hy\le h_\lambda\}\), where
\(H\in\sR^{p\times q}\) is fixed and \(h_\lambda\in\sR^p\)
is coordinatewise nondecreasing in \(\lambda\).
For each branch, linear-programming duality gives
\[
    h_{C(X;\lambda)}(a_k)
    =\max_{Hy\le h_\lambda}a_k^\top y
    =\min_{\substack{\nu_k\ge0\\H^\top\nu_k=a_k}}
    h_\lambda^\top\nu_k.
\]
The assumed nonemptiness and compactness ensure finite optimal
values and attainment. Consequently, \(u\in D(X;\lambda)\) is
equivalent to the existence of \(\nu_1,\ldots,\nu_K\) satisfying
\[
    \nu_k\ge0,\qquad
    H^\top\nu_k=a_k(u;X),\qquad
    b_k(u;X)+h_\lambda^\top\nu_k\le0,
    \quad k=1,\ldots,K.
\]
This formulation avoids enumerating the vertices of the prediction set.

\paragraph*{Nonlinear constraint functions and scope.}
If \(y\mapsto f(u;X,y)\) is continuous and concave and
\(C(X;\lambda)\) is convex, then, for fixed \((u,X,\lambda)\),
\[
    L_\lambda(u;X)
    =-\min_{y\in C(X;\lambda)}\{-f(u;X,y)\}
\]
is the negative optimal value of a convex optimization problem.
The same reasoning applies branchwise when \(f\) is a finite
maximum of concave functions of \(y\). For example, if
\(f(u;X,y)=\omega-\|z(u;X)-y\|_2^2\) and
\(C(X;\lambda)=\{y:\|y-\mu\|_2\le r_\lambda\}\), then
\[
    L_\lambda(u;X)
    =\omega-\bigl[\|z(u;X)-\mu\|_2-r_\lambda\bigr]_+^2.
\]

These results concern evaluation of the margin at a fixed decision.
Convexity of \(D(X;\lambda)\) additionally follows when \(\gU\)
is convex and \(u\mapsto f(u;X,y)\) is convex for every \(y\).
In particular, for fixed \(\lambda\), affine \(a_k(u;X)\) and
\(b_k(u;X)\), and polyhedral \(\gU\), the ellipsoidal case admits
a second-order-cone formulation, while the box and polyhedral cases
admit linear-programming formulations. Tractable evaluation of
\(L_\lambda\) alone does not guarantee tractability of the further
optimization over decisions required to compute the \texttt{DISC} score.

\subsection{Lagrangian conservative certificate}\label{appen:Lagrangian_certificate}
{Suppose that decision space $\gU$ is convex and $f_i$ is $\beta_i$-\emph{semiconcave} for some $\beta_i\geq 0$, meaning that $u\mapsto f_i(u)-\beta_i\lVert u\rVert_2^2/2$ is concave on $\gU$.}
We choose the margin function \(L_{i,\lambda}(u)\) to be uniformly \(\kappa_i\)-strongly convex in \(u\) with \(\kappa_i>0\). Then, for every multiplier \(\rho\geq\beta_i/\kappa_i\), the function $f_i(u)-\rho L_{i,\lambda}(u)$ is concave. This yields the Lagrangian conservative certificate: 
\begin{equation}\nonumber
    \widetilde g_i(\lambda)
    =
    \inf_{\rho\geq\beta_i/\kappa_i}
    \sup_{u\in\gU}
    \left\{
        f_i(u)-\rho L_{i,\lambda}(u)
    \right\}.
\end{equation}
Indeed, \(L_{i,\lambda}(u)\leq0\) implies
\(f_i(u)\leq f_i(u)-\rho L_{i,\lambda}(u)\), and weak duality \citep{boyd2004convex} ensures $\widetilde g_i(\lambda)\geq g_i(\lambda)$ by the definition \eqref{eq:certification_value}. For each fixed \(\rho\), the inner supremum in defining $\widetilde g_i(\lambda)$ is a convex optimization problem, since it maximizes a concave function over the convex set \(\gU\). Moreover, its optimal value is convex in the scalar variable \(\rho\), so the outer infimum can be computed through a one-dimensional convex minimization. Since \(L_{i,\lambda}(u)\) is nondecreasing in \(\lambda\), we know the conservative certificate \(\widetilde g_i(\lambda)\) is also nonincreasing, {and thus $\widetilde{V}_i$ can be located by bisection.}

\subsection{Sequential convex approximation of constrained ERM}\label{appen:SCA_erm}
Let \(m_1=n-k+1\) and \(m_0=n-k\). Denote $\widehat{\mathsf{CVaR}}_{m}(g) = m^{-1}\sum_{i=n-m+1}^n g_{(i)}$ for a vector $g\in \sR^n$. Since \(m_1\widehat{\mathsf{CVaR}}_{m_1}\) and \(m_0\widehat{\mathsf{CVaR}}_{m_0}\) are convex functions, the quantile constraint \(m_1\widehat{\mathsf{CVaR}}_{m_1}(g)-m_0\widehat{\mathsf{CVaR}}_{m_0}(g)\le 0\) is a difference-of-convex (DC) constraint. We solve it by sequential convex approximation. At iteration \(t\), given \((\theta^t,\lambda^t)\), set \(g^t=g(\lambda^t;\theta^t)\) and choose \(w^t\in\partial\{m_0\widehat{\mathsf{CVaR}}_{m_0}(g^t)\}\). When there are no ties, \(w_i^t=1\) if \(g_i^t\) is among the largest \(m_0\) components of \(g^t\), and \(w_i^t=0\) otherwise. By convexity,
\[
    m_0\widehat{\mathsf{CVaR}}_{m_0}(g)
    \ge
    m_0\widehat{\mathsf{CVaR}}_{m_0}(g^t)
    +
    \langle w^t,g-g^t\rangle .
\]
Thus, the DC constraint can be conservatively replaced by
\[
    m_1\widehat{\mathsf{CVaR}}_{m_1}(g)
    -
    m_0\widehat{\mathsf{CVaR}}_{m_0}(g^t)
    -
    \langle w^t,g-g^t\rangle
    \le 0 .
\]
Using the epigraph representation of \(m_1\widehat{\mathsf{CVaR}}_{m_1}\), the conservative surrogate subproblem at iteration \(t\) becomes
\begin{align}\label{eq:dc_cvar_sca_step}
    \min_{\theta,\lambda,u_i,\eta,\xi_i}
    \quad &
    \frac1n\sum_{i=1}^n\phi(u_i) \\
    \mathrm{s.t.}\quad
    &
    L_{\theta,\lambda}(u_i;X_i)\le 0,
    \quad i\in[n], \nonumber\\
    &
    \xi_i\ge g_i(\lambda;\theta)-\eta,
    \qquad
    \xi_i\ge 0,
    \quad i\in[n], \nonumber\\
    &
    m_1\eta+\sum_{i=1}^n\xi_i
    -
    m_0\widehat{\mathsf{CVaR}}_{m_0}(g^t)
    -
    \left\langle w^t,
    \vg(\lambda;\theta)-g^t
    \right\rangle
    \le 0 . \nonumber
\end{align}

Gradient information is available under suitable value-function regularity. Fix $i$, write $p=(\lambda,\theta)$, and consider a reference parameter $p_0$. Suppose that the inner problem admits a twice continuously differentiable nonlinear-programming formulation, with compact $\gU$ described by parameter-independent smooth constraints. Assume that the global maximizer at $p_0$ is unique and that its KKT point satisfies the linear independence constraint qualification, strict complementarity, and the second-order sufficient condition for the equivalent minimization of $-f$. Then $g_i$ is continuously differentiable near $p_0$. Since $f$ and $\gU$ are independent of $p$, \[\nabla g_i(p)=-\nu_i^*(p)\,\nabla_p L_{\theta,\lambda}(u_i^*(p);X_i),\] where $\nu_i^*(p)\ge 0$ is the unique multiplier associated with $L_{\theta,\lambda}(u;X_i)\le 0$. If this constraint is inactive, the gradient is zero \citep{bonnans2000perturbation,rockafellar1998variational}. 


If the inner value is replaced by a conservative upper bound $\bar g_i(p)\ge g_i(p)$ throughout the relevant parameter domain, enforcing $\bar g_i(p)\le 0$ remains sufficient for the exact certificate $g_i(p)\le 0$. Derivatives or generalized gradients used for optimization then correspond to $\bar g_i$. The required differentiability or local Lipschitz assumptions, and any subdifferential characterization, must be verified for the surrogate itself; they do not follow from the upper-bound property.

\subsection{Piecewise-constant augmented ERM maps}
\label{appen:piecewise_constant_augmented_erm}

The augmented-ERM construction in Section~\ref{sec:augmented_erm} appears to require recomputing \(\hat\theta^y\) for every hypothesized label \(y\). When the parameter class is finite, \(\Theta=\{\theta_1,\ldots,\theta_M\}\), the computation can be organized more explicitly. The key observation is that, for each fixed \(\theta_m\), the augmented quantile can be updated without re-sorting the calibration scores for every \(y\).

For each \(m\in[M]\), precompute the calibration scores \(c_{i,m}=V_{\theta_m}(X_i,Y_i)\) for \(i=1,\ldots,n\) and let \(c_{(1),m}\le\cdots\le c_{(n),m}\) be their order statistics. 
And we write $c_{(0),m} = -\infty$ and $c_{(n+1),m} = \infty$.
Set \(k=\lceil (n+1)(1-\alpha)\rceil\), \(a_m=c_{(k-1),m}\), and \(b_m=c_{(k),m}\).
For a candidate label \(y\), write \(v_m(y)=V_{\theta_m}(X_{n+1},y)\). Then the augmented threshold for the fixed parameter \(\theta_m\) is
\begin{align*}
    \lambda_m(y)
    :=
    \hat\lambda_{\theta_m}^y
    &=
    \mathsf{Q}_{1-\alpha}
    \left(
    \{V_{\theta_m}(X_i,Y_i)\}_{i=1}^n,
    V_{\theta_m}(X_{n+1},y)
    \right)\\
    &=
    \min\{b_m,\max\{v_m(y),a_m\}\}.
\end{align*}
Thus, after sorting the calibration scores once for each \(\theta_m\), evaluating the augmented threshold for a new label \(y\) only requires computing \(v_m(y)\) and clipping it to the interval \([a_m,b_m]\).

Next define the fixed-\(\theta_m\) value functions \(\Psi_{i,m}(\lambda) = \min_{u\in D_{\theta_m}(X_i;\lambda)}\phi(u)\) for \(i=1,\ldots,n+1\). The augmented objective evaluated at \(\theta_m\) is \(\widetilde R_m(y)=(n+1)^{-1}\sum_{i=1}^{n+1}\Psi_{i,m}(\lambda_m(y))\).
The augmented ERM therefore selects
\[
    \hat\theta^y=\theta_{m(y)},
    \qquad
    m(y)
    =
    \min\left\{
    m\in[M]:
    \widetilde R_m(y)=\min_{\ell\in[M]}\widetilde R_\ell(y)
    \right\},
\]
where the minimum index implements the deterministic tie-breaking rule. Equivalently, the selected-parameter cell of \(\theta_m\) is
\[
    \mathcal C_m
    =
    \left\{
    y\in\gY:
    \widetilde R_m(y)\le \widetilde R_\ell(y)\ \text{for all }\ell>m,
    \quad
    \widetilde R_m(y)< \widetilde R_\ell(y)\ \text{for all }\ell<m
    \right\}.
\]
On \(\mathcal C_m\), the selected parameter is constant, namely \(\hat\theta^y=\theta_m\). The threshold \(\lambda_m(y)=\hat\lambda_{\theta_m}^y\), however, can still vary with \(y\) through the clipped score \(v_m(y)\).

For a continuous label space, exact computation can be organized as a finite
cell-decomposition problem, provided that the score functions and the inner
value functions admit finite active-formula descriptions. The procedure is as
follows.

\begin{enumerate}
    \item \textit{Construct a common cell decomposition of the label space.}
    For each \(m\in[M]\), the fixed-\(\theta_m\) augmented threshold
    \(\lambda_m(y)=\min\{b_m,\max\{v_m(y),a_m\}\}\) changes formula only when
    \(v_m(y)\) crosses \(a_m\) or \(b_m\). We therefore first split \(\gY\) into
    the three regimes
    \[
        \begin{aligned}
            &\mathcal R_m^-=\{y:v_m(y)\le a_m\},\quad
            \mathcal R_m^0=\{y:a_m<v_m(y)<b_m\},\\
            &\mathcal R_m^+=\{y:v_m(y)\ge b_m,\ v_m(y)>a_m\}.
        \end{aligned}
    \]
    On these regimes,
    \[
        \lambda_m(y)=
        \begin{cases}
        a_m, & y\in\mathcal R_m^-,\\
        v_m(y), & y\in\mathcal R_m^0,\\
        b_m, & y\in\mathcal R_m^+.
        \end{cases}
    \]
    Taking the common refinement of these regimes over \(m=1,\ldots,M\) gives
    regions on which all augmented thresholds have fixed formulas. We then
    further refine these regions, if necessary, along the boundaries where the
    inner value functions \(\Psi_{i,m}(\lambda) = \min_{u\in D_{\theta_m}(X_i;\lambda)}\phi(u)\)   change their active form. Denote the resulting finite collection of refined
    cells by \(\mathfrak{H}\). Thus, on each \(H\in\mathfrak{H}\), every threshold
    \(\lambda_m(y)\) and every objective term
    \(\Psi_{i,m}(\lambda_m(y))\) has a fixed explicit formula. We denote these
    formulas by
    \[
        \lambda_{m,H}(y)
        \quad\text{and}\quad
        \widetilde R_{m,H}(y)
        =
        \frac1{n+1}
        \sum_{i=1}^{n+1}
        \Psi_{i,m}(\lambda_{m,H}(y)).
    \]

    \item \textit{Assign each cell to the selected parameter.}
    On a fixed cell \(H\in\mathfrak{H}\), the augmented ERM reduces to comparing
    the finite list of explicit objective formulas
    \(\widetilde R_{1,H}(y),\ldots,\widetilde R_{M,H}(y)\). With the
    minimum-index tie-breaking rule, the subset of \(H\) assigned to
    \(\theta_m\) is
    \[
        H^{(m)}
        =
        H
        \cap
        \bigcap_{\ell>m}
        \{y:\widetilde R_{m,H}(y)\le \widetilde R_{\ell,H}(y)\}
        \cap
        \bigcap_{\ell<m}
        \{y:\widetilde R_{m,H}(y)< \widetilde R_{\ell,H}(y)\}.
    \]
    The selected-parameter cell for \(\theta_m\) is therefore \(\mathcal C_m = \bigcup_{H\in\mathfrak{H}} H^{(m)}\). Hence \(\hat\theta^y=\theta_m\) for all \(y\in\mathcal C_m\), while the corresponding threshold \(\lambda_m(y)\) may still vary over
    \(\mathcal C_m\).

    \item \textit{Compute the largest threshold needed for each selected parameter.}
    Since \(D_{\theta_m}(x;\lambda)\) is nested decreasing in \(\lambda\), the
    intersection over all labels assigned to \(\theta_m\) is determined by the
    largest threshold appearing on \(\mathcal C_m\). Define
    \[
        \Lambda_m
        =
        \sup_{y\in\mathcal C_m}\lambda_m(y)
        =
        \sup_{H\in\mathfrak{H}:\,H^{(m)}\ne\emptyset}
        \sup_{y\in H^{(m)}}\lambda_{m,H}(y).
    \]
    The exact optimized DISC subset is then
    \[
        \widehat D^{\texttt{FO-DISC}}(X_{n+1})
        =
        \bigcap_{m:\mathcal C_m\ne\emptyset}
        D_{\theta_m}
        \left(
        X_{n+1};
        \Lambda_m
        \right),
    \]
    whenever the supremum is attained, or more generally when the nested family
    is closed under these suprema. Otherwise, the displayed subset is a
    conservative replacement.
\end{enumerate}

A simpler conservative implementation avoids the explicit construction of the selected cells. Since \(\lambda_m(y)\le b_m=c_{(k),m}\) whenever \(k\le n\), one may upper bound the cellwise supremum by \(b_m\) and output
\[
    \widehat D^{\rm cons}(X_{n+1})
    =
    \bigcap_{m=1}^M
    D_{\theta_m}(X_{n+1};b_m).
\]
This subset may be smaller than the exact optimized subset, but it is easy to compute: it only requires sorting the fixed-\(\theta_m\) calibration scores once for each \(m\), followed by a finite intersection over the candidate parameters.

\section{Preliminary lemmas and their proofs}


\begin{lemma}[Perturbation of critical thresholds]\label{lem:score_perturbation_selected}
For \(\theta\in\Theta\), define
\[
    g_\theta(\lambda;z)
    :=
    \sup_{u\in D_\theta(x;\lambda)} f(u;x,y),
    \qquad
    V_\theta(z)
    :=
    \inf\{\lambda\in\sR:g_\theta(\lambda;z)\le 0\},
    \qquad z=(x,y).
\]
Let \(\theta,\theta'\in\Theta\). Suppose that, on the localized interval
\([-\Lambda-t_0,\Lambda+t_0]\), there exists \(\delta\ge0\) such that
\[
    \sup_{\lambda\in[-\Lambda-t_0,\Lambda+t_0]}
    |g_\theta(\lambda;z)-g_{\theta'}(\lambda;z)|
    \le \delta .
\]
Assume further that both threshold maps cross zero with margin \(m_g>0\), in the sense that for every \(\vartheta\in\{\theta,\theta'\}\) and every \(s\ge0\) with
\(V_\vartheta(z)+s\in[-\Lambda-t_0,\Lambda+t_0]\),
\[
    g_\vartheta(V_\vartheta(z)+s;z)\le -m_gs .
\]
If \(|V_\theta(z)|\vee |V_{\theta'}(z)|\le \Lambda\) and \(\delta/m_g\le t_0\), then
\[
    |V_\theta(z)-V_{\theta'}(z)|
    \le
    \frac{\delta}{m_g}.
\]
\end{lemma}

\begin{proof}
Set \(\Delta=\delta/m_g\). We first show that
\(V_{\theta'}(z)\le V_\theta(z)+\Delta\). Since \(|V_\theta(z)|\le\Lambda\)
and \(\Delta\le t_0\), the point \(V_\theta(z)+\Delta\) lies in the localized
interval. By the crossing-margin condition,
\(g_\theta(V_\theta(z)+\Delta;z)\le -m_g\Delta=-\delta\). The perturbation
bound then gives
\(g_{\theta'}(V_\theta(z)+\Delta;z)\le g_\theta(V_\theta(z)+\Delta;z)+\delta
\le0\). Hence, by the definition of \(V_{\theta'}(z)\),
\(V_{\theta'}(z)\le V_\theta(z)+\Delta\). Interchanging the roles of
\(\theta\) and \(\theta'\) gives
\(V_\theta(z)\le V_{\theta'}(z)+\Delta\). Therefore
\(|V_\theta(z)-V_{\theta'}(z)|\le \Delta=\delta/m_g\).
\end{proof}

\begin{lemma}[Perturbation of sample quantiles]\label{lem:quantile_perturbation_selected}
Let \(1\le \lceil (n+1)(1-\alpha)\rceil\le n\). For two vectors \(a=(a_1,\ldots,a_n)\) and \(b=(b_1,\ldots,b_n)\), define
\[
    Q(a)=\mathsf{Q}_{1-\alpha}\left(\{a_1,\ldots,a_n\}, \infty\right),
    \qquad
    Q(b)=\mathsf{Q}_{1-\alpha}\left(\{b_1,\ldots,b_n\}, \infty\right).
\]
If \(\max_{j\in[n]}|a_j-b_j|\le \delta\), then
\[
    |Q(a)-Q(b)|\le \delta .
\]
\end{lemma}

\begin{proof}
Let \(k=\lceil (n+1)(1-\alpha)\rceil\). Since \(1\le k\le n\), the additional
element \(\infty\) does not affect the finite \(k\)-th order statistic, so
\(Q(a)=a_{(k)}\) and \(Q(b)=b_{(k)}\). The assumption
\(\max_j|a_j-b_j|\le\delta\) implies \(b_j-\delta\le a_j\le b_j+\delta\) for
all \(j\). Hence, for every \(t\in\sR\), if \(a_j\le t\), then
\(b_j\le t+\delta\). Taking \(t=Q(a)\), at least \(k\) components of \(b\) are
no larger than \(Q(a)+\delta\), and thus \(Q(b)\le Q(a)+\delta\). By symmetry,
\(Q(a)\le Q(b)+\delta\). Therefore \(|Q(a)-Q(b)|\le\delta\).
\end{proof}

\begin{lemma}[Replace-one stability]\label{prop:sufficient_stability_bound}
Suppose Assumption \ref{ass:primitive_stability_selected_rule} holds and let
\(k=\lceil(n+1)(1-\alpha)\rceil\le n\). Let
\(S=(Z_1,\ldots,Z_n)\), where \(Z_i=(X_i,Y_i)\), and let
\(S^{(i)}=(Z_1,\ldots,Z_{i-1},Z_i',Z_{i+1},\ldots,Z_n)\), where \(Z_i'\) is an independent copy of \(Z_i\). Define
\[
    \hat\lambda_{S}
    =
    \mathsf{Q}_{1-\alpha}
    \left(
    \{V_{\widehat\theta_S}(Z_j)\}_{j=1}^n,\infty
    \right),
\]
and
\[
    \hat\lambda_{S,i}'
    =
    \mathsf{Q}_{1-\alpha}
    \left(
    \{V_{\widehat\theta_{S^{(i)}}}(Z_j)\}_{j\neq i}
    \cup
    \{V_{\widehat\theta_S}(Z_i),\infty\}
    \right).
\]
Then
\[
    \frac1n\sum_{i=1}^n
    \E_{S,Z_i'}
    \left|
    \mathbbm{1}\left\{
    V_{\widehat\theta_S}(Z_i')\le \hat\lambda_S
    \right\}
    -
    \mathbbm{1}\left\{
    V_{\widehat\theta_{S^{(i)}}}(Z_i')\le \hat\lambda_{S,i}'
    \right\}
    \right|
    \le
    4c_{\mathrm{dens}}\varepsilon_{\rm stab} .
\]
\end{lemma}

\begin{proof}
Fix \(i\in[n]\). By Assumption~\ref{ass:primitive_stability_selected_rule}(ii),
\(\sup_z|V_{\widehat\theta_S}(z)-V_{\widehat\theta_{S^{(i)}}}(z)|
\le \varepsilon_{\rm stab}\). Applying Lemma~\ref{lem:quantile_perturbation_selected} to the
two quantile vectors, whose entries differ by at most \(\varepsilon_{\rm stab}\) except for
the common finite entry \(V_{\widehat\theta_S}(Z_i)\), gives
\(|\hat\lambda_S-\hat\lambda_{S,i}'|\le\varepsilon_{\rm stab}\).

Write \(a=V_{\widehat\theta_S}(Z_i')\), \(a'=V_{\widehat\theta_{S^{(i)}}}(Z_i')\),
\(b=\hat\lambda_S\), and \(b'=\hat\lambda_{S,i}'\). If
\(\mathbbm{1}\{a\le b\}\neq\mathbbm{1}\{a'\le b'\}\), then \(0\) lies between
\(a-b\) and \(a'-b'\). Consequently, \(|a-b|\le |(a-b)-(a'-b')|\le |a-a'|+|b-b'|\le 2\varepsilon_{\rm stab}\).
It follows that the absolute difference between the two indicators is bounded
by \(\mathbbm{1}\{|V_{\widehat\theta_S}(Z_i')-\hat\lambda_S|
\le 2\varepsilon_{\rm stab}\}\). Conditioning on \(S\), Assumption~\ref{ass:primitive_stability_selected_rule}(i)
implies \(\hat\lambda_S\in[-\Lambda,\Lambda]\), and \(t_0\ge2\varepsilon_{\rm stab}\)
ensures that the interval
\([\hat\lambda_S-2\varepsilon_{\rm stab},\hat\lambda_S+2\varepsilon_{\rm stab}]\) is contained in
\([-\Lambda-t_0,\Lambda+t_0]\). Hence, by the density upper bound in
Assumption~\ref{ass:primitive_stability_selected_rule}(iii),
\[
    \E_{Z_i'}\!\left[
    \left|
    \mathbbm{1}\{V_{\widehat\theta_S}(Z_i')\le\hat\lambda_S\}
    -
    \mathbbm{1}\{V_{\widehat\theta_{S^{(i)}}}(Z_i')\le\hat\lambda_{S,i}'\}
    \right|
    \,\middle|\,S
    \right]
    \le 4c_{\mathrm{dens}}\varepsilon_{\rm stab} .
\]
Averaging over \(S\) and over \(i\in[n]\) proves the claim.
\end{proof}

\begin{lemma}[Uniform error bound for conformal critical quantiles]\label{lem:lambda_error_bound}
Suppose Assumption \ref{ass:threshold_complexity_inclusion} and
Assumption \ref{ass:uniform_quantile_regularity}(i) hold, and let
\(k_n=\lceil (n+1)(1-\alpha)\rceil\le n\). Define
\[
    \Delta_n(\eta)
    :=
    \mathfrak C
    \sqrt{
    \frac{d_{\gH}\log(en/d_{\gH})+\log(4/\eta)}{n}
    }
    +
    \frac{2}{n}.
\]
Then, with probability at least \(1-\eta/2\), if
\(\Delta_n(\eta)\le c_q\epsilon_q\), we have
\[
    \sup_{\theta\in\Theta}
    |\hat\lambda_\theta-\lambda_\theta^*|
    \le
    \frac{\Delta_n(\eta)}{c_q}.
\]
\end{lemma}

\begin{proof}
For \(\theta\in\Theta\), write
\(F_\theta(t)=\Prob(V_\theta(X,Y)\le t)\) and
\(\widehat F_{n,\theta}(t)=n^{-1}\sum_{i=1}^n
\mathbbm{1}\{V_\theta(X_i,Y_i)\le t\}\). By the VC inequality applied to
\(\gH\), with probability at least \(1-\eta/2\),
\[
    \sup_{\theta\in\Theta,\ t\in\sR}
    |\widehat F_{n,\theta}(t)-F_\theta(t)|
    \le
    \varepsilon_{\gH, n}(\eta)
    :=
    \mathfrak C
    \sqrt{
    \frac{d_{\gH}\log(en/d_{\gH})+\log(4/\eta)}{n}
    } .
\]
Thus \(\Delta_n(\eta)=\varepsilon_{\gH, n}(\eta)+2/n\). The term \(2/n\) accounts for the
rank correction of the conformal quantile: with
\(k_n=\lceil (n+1)(1-\alpha)\rceil\), one has
\(\hat\lambda_\theta=\inf\{t:\widehat F_{n,\theta}(t)\ge k_n/n\}\) and
\(0\le k_n/n-(1-\alpha)\le2/n\). Let
\(\delta_n=\Delta_n(\eta)/c_q\), so that \(\delta_n\le\epsilon_q\).

For the upper bound, Assumption~\ref{ass:uniform_quantile_regularity}(i) gives
\(F_\theta(\lambda_\theta^*+\delta_n)\ge
1-\alpha+c_q\delta_n=1-\alpha+\varepsilon_{\gH, n}(\eta)+2/n\). On the VC event,
\(\widehat F_{n,\theta}(\lambda_\theta^*+\delta_n)\ge
1-\alpha+2/n\ge k_n/n\), and therefore
\(\hat\lambda_\theta\le\lambda_\theta^*+\delta_n\). For the lower bound,
the same regularity assumption gives
\(F_\theta(\lambda_\theta^*-\delta_n)\le
1-\alpha-c_q\delta_n=1-\alpha-\varepsilon_{\gH, n}(\eta)-2/n\). On the VC event,
\(\widehat F_{n,\theta}(\lambda_\theta^*-\delta_n)\le
1-\alpha-2/n<k_n/n\), and hence
\(\hat\lambda_\theta>\lambda_\theta^*-\delta_n\). Combining the two bounds
yields
\(|\hat\lambda_\theta-\lambda_\theta^*|\le\delta_n
=\Delta_n(\eta)/c_q\). Since the argument is uniform over \(\theta\), the proof
is complete.
\end{proof}

\section{Additional Theoretical Results}

\subsection{Stability analysis of inclusion gap}\label{appen:stability_bound}

The next result provides a sharper bound under a replace-one stability condition. For each $i\in[n]$, let $\hat\theta^{(i)}$ be the empirical minimizer of \eqref{eq:empirical_model_selection} computed after replacing $(X_i,Y_i)$ by an independent copy $(X_i',Y_i')$.

\begin{assumption}\label{ass:primitive_stability_selected_rule}
Assume there exist constants $\Lambda>0$, $c_{\mathrm{dens}}>0$, $t_0>0$, and a nonnegative number $\varepsilon_{\rm stab}$ such that: (i) for any $\theta\in\Theta$, $|V_\theta(X,Y)|\le \Lambda$ almost surely; (ii) the selected critical-threshold mapping is replace-one stable, in the sense that $\max_{i\in[n]}\sup_{(x,y)}|V_{\hat\theta}(x,y)-V_{\hat\theta^{(i)}}(x,y)|\le \varepsilon_{\rm stab}$ almost surely; and (iii) for any $\theta\in\Theta$, the distribution of $V_\theta(X,Y)$ admits a density bounded by $c_{\mathrm{dens}}$ on $[-\Lambda-t_0,\Lambda+t_0]$, where $t_0\ge 2\varepsilon_{\rm stab}$.
\end{assumption}

Assumption \ref{ass:primitive_stability_selected_rule} separates the two quantities needed for the refinement. The boundedness of the \texttt{DISC} score in $(i)$ is mild. The stability radius $\varepsilon_{\rm stab}$ in $(ii)$ measures the data-adaptivity cost of selecting $\hat\theta$ from the calibration data. The density upper bound is a local anti-concentration condition, ensuring that a perturbation of the critical score by at most $\varepsilon_{\rm stab}$ changes the corresponding probability mass by order $\varepsilon_{\rm stab}$. For stable implementations of \eqref{eq:empirical_model_selection} that have uniform stability \citep{feldman2018generalization}, $\varepsilon_{\rm stab}$ is of order $1/n$, and the next theorem leads to a sharper inclusion gap than that in Theorem~\ref{thm:direct_concentration_inclusion}.

\begin{theorem}\label{thm:stability_inclusion_coverage}
Suppose Assumption \ref{ass:primitive_stability_selected_rule} holds. Then $\Prob\!\left\{\widehat D^{\texttt{O-DISC}}(X_{n+1})\subseteq A(X_{n+1},Y_{n+1})\right\} \ge 1-\alpha-\frac1n-4c_{\mathrm{dens}}\cdot\varepsilon_{\rm stab}$.
\end{theorem}

For stable implementations of the parameter optimization step that have uniform stability \citep{feldman2018generalization}, $\varepsilon_{\rm stab}$ is of order $1/n$, leading to a sharper marginal inclusion gap than that in Theorem~\ref{thm:direct_concentration_inclusion}.

\begin{proof}
Let \(S=(Z_1,\ldots,Z_n)\), where \(Z_i=(X_i,Y_i)\), and write
\(\widehat D_S(x)=D_{\widehat\theta_S}(x;\hat\lambda_S)\), where
\(\hat\lambda_S\) is the conformal quantile of
\(\{V_{\widehat\theta_S}(Z_j)\}_{j=1}^n\) augmented with \(\infty\). Define
\(\ell_S(z)=\mathbbm{1}\{V_{\widehat\theta_S}(z)\le\hat\lambda_S\}\) and
\(L(S)=\Prob(\widehat D_S(X)\subseteq A(X,Y)\mid S)\) for an independent test
point \(Z=(X,Y)\). By attainability and the definition of the critical threshold,
\(\ell_S(z)=\mathbbm{1}\{\widehat D_S(x)\subseteq A(x,y)\}\), and therefore
\(L(S)=\E_Z[\ell_S(Z)]\) and
\(\Prob(\widehat D_S(X)\subseteq A(X,Y))=\E_S[L(S)]\).

For each \(i\), let \(S^{(i)}\) be obtained from \(S\) by replacing \(Z_i\) with
an independent copy \(Z_i'\). Define \(\hat\lambda_{S,i}\) as the conformal
quantile formed from \(\{V_{\widehat\theta_S}(Z_j)\}_{j\neq i}\), the replaced
score \(V_{\widehat\theta_{S^{(i)}}}(Z_i')\), and \(\infty\). Similarly, define
\(\hat\lambda_{S,i}'\) as the conformal quantile formed from
\(\{V_{\widehat\theta_{S^{(i)}}}(Z_j)\}_{j\neq i}\), the original score
\(V_{\widehat\theta_S}(Z_i)\), and \(\infty\). Set
\[
    \widehat L(S,Z_1',\ldots,Z_n')
    :=
    \frac1n\sum_{i=1}^n
    \mathbbm{1}\{
    V_{\widehat\theta_S}(Z_i)\le \hat\lambda_{S,i}
    \}.
\]
For every realization of \(S,Z_1',\ldots,Z_n'\), at least \(k-1\) of the
original scores \(\{V_{\widehat\theta_S}(Z_i)\}_{i=1}^n\) are no larger than
their corresponding perturbed quantiles, where
\(k=\lceil(n+1)(1-\alpha)\rceil\). Indeed, each of the \(k-1\) smallest
original scores has at most \(k-1\) elements below it after one replacement and
the extra \(\infty\) is added. Hence
\(\widehat L(S,Z_1',\ldots,Z_n')\ge(k-1)/n\ge1-\alpha-1/n\).

By exchangeability of \(Z_i\) and \(Z_i'\),
\[
    \E[\widehat L(S,Z_1',\ldots,Z_n')]
    =
    \frac1n\sum_{i=1}^n
    \E_{S,Z_i'}
    \mathbbm{1}\{
    V_{\widehat\theta_{S^{(i)}}}(Z_i')
    \le
    \hat\lambda_{S,i}'
    \}.
\]
On the other hand,
\[
    \E_S[L(S)]
    =
    \frac1n\sum_{i=1}^n
    \E_{S,Z_i'}
    \mathbbm{1}\{
    V_{\widehat\theta_S}(Z_i')
    \le
    \hat\lambda_S
    \}.
\]
By Lemma~\ref{prop:sufficient_stability_bound},
\[
    \left|
    \E_S[L(S)]
    -
    \E[\widehat L(S,Z_1',\ldots,Z_n')]
    \right|
    \le
    4c_{\mathrm{dens}}\varepsilon_{\rm stab} .
\]
Combining the last display with
\(\E[\widehat L(S,Z_1',\ldots,Z_n')]\ge1-\alpha-1/n\) yields
\[
    \Prob(
    \widehat D_S(X)\subseteq A(X,Y)
    )
    =
    \E_S[L(S)]
    \ge
    1-\alpha-\frac1n-4c_{\mathrm{dens}}\varepsilon_{\rm stab} .
\]
This completes the proof.
\end{proof}

\subsection{Verification of Assumption~\ref{ass:uniform_quantile_regularity}(ii)}\label{appen:verification_lipschitz}

The sufficient condition in the next proposition follows a standard
perturbation-analysis argument for parametric optimization problems
\citep{bonnans2000perturbation,dontchev2014implicit}.

\begin{proposition}[A sufficient condition for Lipschitz stability]\label{pro:lipschitz_stability_sufficient}
For \(\theta\in\Theta\), write \(D_\theta(x;\lambda)=\{u\in\gU:L_{\theta,\lambda}(u;x)\le 0\}\).
Suppose that the following conditions hold uniformly over \(\theta\in\Theta\)
and all \(\lambda\) satisfying \(|\lambda-\lambda_\theta^*|\le\epsilon_q\).
\begin{itemize}
    \item[(1)] \(D_\theta(X;\lambda)\) is nonempty and compact almost surely.

    \item[(2)] \(\phi\) is \(L_\phi\)-Lipschitz on \(\gU\).

    \item[(3)] There exists \(S_\lambda>0\) such that \(|L_{\theta,\lambda'}(u;x)-L_{\theta,\lambda}(u;x)|\le S_\lambda|\lambda'-\lambda|\) for all \(u\in\gU\), \(x\), \(\theta\in\Theta\), and all relevant \(\lambda,\lambda'\).

    \item[(4)] There exists \(\kappa>0\) such that, almost surely,
    \[
        \mathsf{dist}\bigl(u,D_\theta(X;\lambda)\bigr)
        \le
        \kappa [L_{\theta,\lambda}(u;X)]_+
        \qquad
        \text{for all }u\in\gU .
    \]

    \item[(5)] There exists \(B>0\) such that \(\left|\min_{u\in D_\theta(X;\lambda)}\phi(u)\right|\le B\) almost surely for all relevant \(\theta\) and \(\lambda\).
\end{itemize}
Then Assumption~\ref{ass:uniform_quantile_regularity}(ii) holds with
\(\Gamma=L_\phi\kappa S_\lambda\) and the constant \(B\) in condition (5).
\end{proposition}

\begin{proof}
For simplicity, write
\(M_{\theta,\lambda}(x)=\min_{u\in D_\theta(x;\lambda)}\phi(u)\). Fix
\(\theta\in\Theta\) and two local calibration levels \(\lambda,\lambda'\), and
let \(u_\lambda(x)\in\argmin_{u\in D_\theta(x;\lambda)}\phi(u)\), which exists
by compactness and continuity. Since \(u_\lambda(x)\in D_\theta(x;\lambda)\),
we have \(L_{\theta,\lambda}(u_\lambda(x);x)\le0\). By condition (3),
\(L_{\theta,\lambda'}(u_\lambda(x);x)\le S_\lambda|\lambda'-\lambda|\). The
error bound in condition (4) then gives
\[
    \mathsf{dist}\bigl(u_\lambda(x),D_\theta(x;\lambda')\bigr)
    \le
    \kappa S_\lambda|\lambda'-\lambda| .
\]
Thus there exists \(v_{\lambda'}\in D_\theta(x;\lambda')\) with
\(\|u_\lambda(x)-v_{\lambda'}\|\le\kappa S_\lambda|\lambda'-\lambda|\). By the
Lipschitz continuity of \(\phi\),
\[
    M_{\theta,\lambda'}(x)
    \le
    \phi(v_{\lambda'})
    \le
    M_{\theta,\lambda}(x)
    +
    L_\phi\kappa S_\lambda|\lambda'-\lambda|.
\]
Exchanging the roles of \(\lambda\) and \(\lambda'\) gives \(|M_{\theta,\lambda'}(x)-M_{\theta,\lambda}(x)|\le L_\phi\kappa S_\lambda|\lambda'-\lambda|\).
Taking expectation over \(X\) yields the Lipschitz condition in
Assumption~\ref{ass:uniform_quantile_regularity}(ii) with
\(\Gamma=L_\phi\kappa S_\lambda\). Condition (5) is exactly the boundedness
condition in Assumption~\ref{ass:uniform_quantile_regularity}(ii).
\end{proof}

\subsection{Verification of Assumption~\ref{ass:uniform_quantile_regularity}(iii)}\label{appen:verification_complexity}

\begin{proposition}[A sufficient condition for bounded \(\mathsf{Pdim}(\mathcal M_\delta)\)]
Suppose that $\gX \subseteq \sR^{d_{x}}$, \(\mathcal U\subseteq\sR^{d_u}\) and \(\Theta\subseteq\sR^{d_\theta}\)
are compact semi-algebraic sets. Assume that \(L_{\theta,\lambda}(u;x)\) and
\(\phi(u)\) are polynomial functions of their arguments with degree uniformly
bounded by a constant \(r\). Assume also that, for all relevant
\((\theta,\lambda)\), the feasible set
\(D_\theta(x;\lambda)=\{u\in\mathcal U:L_{\theta,\lambda}(u;x)\le0\}\) is
nonempty, and that the local calibration levels satisfying
\(|\lambda-\lambda_\theta^*|\le\epsilon_q\) lie in a bounded interval. Then,
for every \(0<\delta\le\epsilon_q\), the decision risk class
\[
    \mathcal M_\delta
    =
    \left\{
    x\mapsto
    \min_{u\in D_\theta(x;\lambda)}\phi(u):
    \theta\in\Theta,\ |\lambda-\lambda_\theta^*|\le\delta
    \right\}
\]
has finite pseudo-dimension. In particular, there exists a constant \(C\)
depending only on \(d_x,d_u,r\) and the number of defining constraints such
that
\[
    \mathsf{Pdim}(\mathcal M_\delta)
    \le
    C(d_\theta+1)\log(d_\theta+2).
\]
\end{proposition}

\begin{proof}
For any real-valued function class \(\mathcal F\),
\(\mathsf{Pdim}(\mathcal F)=\mathsf{VCdim}(\mathsf{sub}(\mathcal F))\), where
\(\mathsf{sub}(\mathcal F)=\{(x,t)\mapsto
\mathbbm{1}\{f(x)\le t\}:f\in\mathcal F\}\). Hence it suffices to bound the VC
dimension of \(\mathsf{sub}(\mathcal M_\delta)\).

Fix \(0<\delta\le\epsilon_q\). By hypothesis, all relevant \(\lambda\)'s belong
to a fixed bounded interval \(I\). Define the larger class
\[
    \widetilde{\mathcal M}
    =
    \left\{
    x\mapsto
    \min_{u\in D_\theta(x;\lambda)}\phi(u):
    \theta\in\Theta,\ \lambda\in I
    \right\}.
\]
Since \(\mathcal M_\delta\subseteq\widetilde{\mathcal M}\), it is enough to
bound \(\mathsf{VCdim}(\mathsf{sub}(\widetilde{\mathcal M}))\). For any
\(\theta,\lambda\), the minimum
\(M_{\theta,\lambda}(x)=\min_{u\in D_\theta(x;\lambda)}\phi(u)\) is attained by
compactness and continuity. Moreover,
\(M_{\theta,\lambda}(x)\le z\) if and only if there exists \(u\in\mathcal U\)
such that \(L_{\theta,\lambda}(u;x)\le0\) and \(\phi(u)\le z\). This is a
first-order formula over the reals with existential variable \(u\), free
variables \((x,z)\), and parameters \((\theta,\lambda)\). Since
\(\mathcal U\) and \(\Theta\) are semi-algebraic and
\(L_{\theta,\lambda}\) and \(\phi\) are polynomials of uniformly bounded degree,
the Tarski--Seidenberg theorem implies that this formula is equivalent to a
quantifier-free Boolean combination of polynomial inequalities whose degrees
and number depend only on the structural constants of the problem.

Therefore each subgraph indicator in
\(\mathsf{sub}(\widetilde{\mathcal M})\) is definable by a Boolean combination
of a fixed number of polynomial inequalities with parameter vector
\((\theta,\lambda)\in\sR^{d_\theta+1}\). By standard VC bounds for
semi-algebraic concept classes, there exists a constant \(C\), depending only
on \(d_x,d_u,r\) and the number of defining constraints, such that
\[
    \mathsf{VCdim}(\mathsf{sub}(\widetilde{\mathcal M}))
    \le
    C(d_\theta+1)\log(d_\theta+2).
\]
Since \(\mathsf{sub}(\mathcal M_\delta)\subseteq
\mathsf{sub}(\widetilde{\mathcal M})\), the same bound applies to
\(\mathsf{Pdim}(\mathcal M_\delta)\).
\end{proof}

\section{Proof of main results}

\subsection{Proof of Theorem \ref{thm:inclusion}}

\begin{proof}
Since the score $V_i=\inf\left\{\lambda\in\sR:\sup_{u\in D(X_i;\lambda)} f(u;X_i,Y_i)\leq 0\right\}$ is a measurable function of $(X_i,Y_i)$, and $(X_i,Y_i)$ are exchangeable for $i\in[n+1]$, the random variables $V_1,\ldots,V_{n+1}$ are also exchangeable.

By nestedness and the attainability of \(V_{n+1}\), the set of safe thresholds
is an upper set with critical value \(V_{n+1}\). Thus, for any \(\lambda\),
\(D(X_{n+1};\lambda)\subseteq A(X_{n+1},Y_{n+1})\) holds if and only if
\(V_{n+1}\le\lambda\). Taking \(\lambda=\hat\lambda\), the inclusion event
is exactly \(\{V_{n+1}\le\hat\lambda\}\). Since \(\hat\lambda\) is the
conformal quantile of \(\{V_i\}_{i=1}^n\) augmented with \(\infty\), and
\(\{V_i\}_{i=1}^{n+1}\) are exchangeable, the standard conformal argument gives
\(\Prob\{V_{n+1}\le\hat\lambda\}\ge1-\alpha\). If the scores are distinct
with probability one, the usual upper bound
\(\Prob\{V_{n+1}\le\hat\lambda\}\le 1-\alpha+1/(n+1)\) also follows.
\end{proof}

\subsection{Proofs of Propositions \ref{prop:dir_score_PS} and \ref{prop:dir_score_add}}

\begin{proof}
Recall the definition $A^c(X_i,Y_i) = \{u\in \gU: f(u;X_i,Y_i) > 0\}$.
    \begin{itemize}
        \item For a fixed \(u\in A^c(X_i,Y_i)\), exclusion from \(D(X_i;\lambda)\) is equivalent to \(\sup_{y\in C(X_i;\lambda)}f(u;X_i,y)>0\). Since \(C(X_i;\lambda)=\{y:s(X_i,y)\leq\lambda\}\), this occurs precisely when the prediction set contains some \(y\) satisfying \(f(u;X_i,y)>0\). The smallest such level is \(\inf_{y:f(u;X_i,y)>0}s(X_i,y)\). Taking the supremum over all decisions that are unsafe under \(Y_i\) gives \eqref{eq:dir_score_PS}.

        \item For a fixed unsafe decision \(u\in A^c(X_i,Y_i)\), exclusion from \(D(X_i;\lambda)\) is equivalent to \(\hat f(u;X_i)+\mu(u)+\lambda\sigma(u)>0\). Since \(\sigma(u)>0\), this holds precisely when \(\lambda>-\{\hat f(u;X_i)+\mu(u)\}/\sigma(u)\). Therefore, the smallest level at which all unsafe decisions are excluded is the supremum of these exclusion levels, which gives \eqref{eq:dir_score_add}.
    \end{itemize}
\end{proof}

\subsection{Proof of Proposition \ref{prop:ps_scp_comparison}}

\begin{proof}
If \(\lambda\ge V_i^{\texttt{Base-PS}}\), then \(Y_i\in C(X_i;\lambda)\).
Hence, for every \(u\in D(X_i;\lambda)\),
\(f(u;X_i,Y_i)\le \sup_{y\in C(X_i;\lambda)} f(u;X_i,y)\le0\), which implies
\(D(X_i;\lambda)\subseteq A(X_i,Y_i)\). Therefore
\(V_i\le V_i^{\texttt{Base-PS}}\). Equivalently, in the pointwise formula for
\(V_i\), whenever \(f(u;X_i,Y_i)>0\), the realized label \(Y_i\) is feasible
for the inner minimization, so \(\inf_{\{y\in \gY:f(u;X_i,y)>0\}}s(X_i,y)\le s(X_i,Y_i)=V_i^{\texttt{Base-PS}}\).
By monotonicity of the conformal quantile,
\(\hat\lambda\le\hat\lambda^{\texttt{Base-PS}}\). Since
\(D(X;\lambda)\) is decreasing in \(\lambda\), the inclusion of safe subsets follows.
\end{proof}

\subsection{Proof of Proposition \ref{prop:additive_unif_comparison}}

\begin{proof}
Let $R_i(u) = f(u;X_i,Y_i)-\hat f(u;X_i)$. For any \(u\in D(X_i;V_i^{\texttt{Base-Add}})\), we have
\(f(u;X_i,Y_i)=\hat f(u;X_i)+R_i(u)
\le \hat f(u;X_i)+\mu(u)+V_i^{\texttt{Base-Add}}\sigma(u)\le0\). Hence
\(D(X_i;V_i^{\texttt{Base-Add}})\subseteq A(X_i,Y_i)\), which implies
\(V_i\le V_i^{\texttt{Base-Add}}\). Equivalently, this follows from the
pointwise formula because, on the infeasible set \(f(u;X_i,Y_i)>0\),
\(-\{\hat f(u;X_i)+\mu(u)\}/\sigma(u)\le\{R_i(u)-\mu(u)\}/\sigma(u)\).
By monotonicity of the conformal quantile,
\(\hat\lambda\le\hat\lambda^{\texttt{Base-Add}}\). Since
\(D(X;\lambda)\) is decreasing in \(\lambda\), the inclusion of safe subsets follows.
\end{proof}

\subsection{Proof of Theorem \ref{thm:inclusion_certificate}}

\begin{proof}
If \(\widetilde g_i(t)\le0\) for all \(t\ge\lambda\), then
\(g_i(t)\le\widetilde g_i(t)\le0\) for all \(t\ge\lambda\). Hence every
feasible level for \(\widetilde V_i\) is feasible for \(V_i\), and therefore
\(\widetilde V_i\ge V_i\). By exchangeability of the conservatively computed
scores, the conformal quantile satisfies
\(\Prob\{\widetilde V_{n+1}\le\widetilde\lambda\}\ge1-\alpha\). On this event,
attainability of the conservative threshold gives
\(\widetilde g_{n+1}(\widetilde\lambda)\le0\), and hence
\(g_{n+1}(\widetilde\lambda)\le0\). By the definition of \(g_{n+1}\), this is
equivalent to
\(D(X_{n+1};\widetilde\lambda)\subseteq A(X_{n+1},Y_{n+1})\).
\end{proof}

\subsection{Proof of Theorem~\ref{thm:direct_concentration_inclusion}}

\begin{proof}
For any \(\theta\in\Theta\), let \(F_\theta(t)=\Prob\{V_\theta(Z)\le t\}\)
and
\(\widehat F_{n,\theta}(t)=n^{-1}\sum_{i=1}^n
\mathbbm{1}\{V_\theta(Z_i)\le t\}\), where \(Z=(X,Y)\) is an independent test
point and \(Z_i=(X_i,Y_i)\). By Assumption~\ref{ass:threshold_complexity_inclusion},
the class
\[
    \gH
    =
    \left\{
    \mathbbm{1}\{V_\theta(x,y)\le t\}:
    \theta\in\Theta,\ t\in\sR
    \right\}
\]
has VC dimension at most \(d_{\gH}\). Therefore, by the standard VC uniform
convergence inequality, with probability at least \(1-\eta\),
\[
    \sup_{\theta\in\Theta,\ t\in\sR}
    |\widehat F_{n,\theta}(t)-F_\theta(t)|
    \le
    \varepsilon_{\gH,n}(\eta),
    \qquad
    \varepsilon_{\gH,n}(\eta)
    =
    \mathfrak C
    \sqrt{
    \frac{d_{\gH}\log(en/d_{\gH})+\log(2/\eta)}{n}
    } .
\]
Denote this event by \(\mathcal E_n\). On \(\mathcal E_n\), the bound holds
uniformly and can be applied to the data-dependent choice \(\hat\theta\).
Since \(\hat\lambda_{\hat\theta}\) is the conformal quantile and
\(k=\lceil(n+1)(1-\alpha)\rceil\le n\), we have
\(\widehat F_{n,\hat\theta}(\hat\lambda_{\hat\theta})\ge k/n
\ge 1-\alpha\). Hence, on \(\mathcal E_n\),
\[
    F_{\hat\theta}(\hat\lambda_{\hat\theta})
    \ge
    \widehat F_{n,\hat\theta}(\hat\lambda_{\hat\theta})
    -
    \varepsilon_{\gH,n}(\eta)
    \ge
    1-\alpha-\varepsilon_{\gH,n}(\eta).
\]
Equivalently, conditional on \(Z_1,\ldots,Z_n\), the probability of
\(\{V_{\hat\theta}(Z_{n+1})\le\hat\lambda_{\hat\theta}\}\) is at
least \(1-\alpha-\varepsilon_{\gH,n}(\eta)\). By attainability and the
definition of the critical threshold, this event implies
\(D_{\hat\theta}(X_{n+1};\hat\lambda_{\hat\theta})
\subseteq A(X_{n+1},Y_{n+1})\). Since
\(\widehat D^{\texttt{O-DISC}}(X_{n+1})
=D_{\hat\theta}(X_{n+1};\hat\lambda_{\hat\theta})\), we obtain the
stated conditional guarantee on \(\mathcal E_n\).
\end{proof}

\subsection{Proof of Proposition~\ref{prop:dc_cvar_equivalence}}

\begin{proof}
Let \(m_1=n-k+1\) and \(m_0=n-k\), with the convention that
\(m_0\widehat{\mathsf{CVaR}}_{m_0}=0\) when \(m_0=0\). For any \(a\in\sR^n\),
\(m\widehat{\mathsf{CVaR}}_m(a)\) equals the sum of the largest \(m\)
components of \(a\). Hence
\(m_1\widehat{\mathsf{CVaR}}_{m_1}(a)
-m_0\widehat{\mathsf{CVaR}}_{m_0}(a)=a_{(k)}\). Applying this identity to
\(a=\vg(\lambda;\theta)\), the second constraint in \eqref{eq:dc_cvar_erm} is
equivalent to \(g_{(k)}(\lambda;\theta)\le0\), namely, at least \(k\) of the
verification values \(g_i(\lambda;\theta)\) are nonpositive.

By the closed upper-ray property of the inclusion certificate,
\(g_i(\lambda;\theta)\le0\) if and only if
\(V_\theta(X_i,Y_i)\le\lambda\). Therefore,
\(g_{(k)}(\lambda;\theta)\le0\) is equivalent to requiring at least \(k\) of
the critical thresholds \(V_\theta(X_i,Y_i)\) to be no larger than \(\lambda\),
which is exactly the sample quantile constraint
\(\hat\lambda_\theta\le\lambda\). Thus the difference-of-CVaR constraint
in \eqref{eq:dc_cvar_erm} is equivalent to the quantile constraint in
\eqref{eq:erm_quantile_constraint}.

It remains to compare the objectives. For fixed \((\theta,\lambda)\), let
\(M_i(\theta,\lambda)=\min\{\phi(u):u\in D_\theta(X_i;\lambda)\}\). Since \(D_\theta(X_i;\lambda)\)
shrinks as \(\lambda\) increases, \(M_i(\theta,\lambda)\) is nondecreasing in
\(\lambda\). Hence, for each fixed \(\theta\), the smallest feasible
calibration level \(\lambda=\hat\lambda_\theta\) attains the minimum over
all \(\lambda\ge\hat\lambda_\theta\), and the resulting objective is
precisely the empirical risk in \eqref{eq:empirical_model_selection}.
Minimizing over \(\theta\in\Theta\) therefore gives the same minimum objective
value and the same optimal family parameters. Finally, if the objective is
flat for some \(\lambda>\hat\lambda_\theta\), such a \(\lambda\) may also be
optimal in \eqref{eq:dc_cvar_erm}; nevertheless,
\(\lambda=\hat\lambda_\theta\) always extends any optimal
\(\hat\theta\) of \eqref{eq:empirical_model_selection} to an optimal
solution of \eqref{eq:dc_cvar_erm}.
\end{proof}

\subsection{Proof of Theorem~\ref{thm:excess_risk_bound}}

\begin{proof}[Proof of Theorem~\ref{thm:excess_risk_bound}]
For $\theta\in\Theta$ and $\lambda\in\sR$, write
$m_{\theta,\lambda}(x)
:=\min_{u\in D_\theta(x;\lambda)}\phi(u)$,
and let
$R(\theta,\lambda):=\E[m_{\theta,\lambda}(X)]$ and
$\widehat R_n(\theta,\lambda)
:=n^{-1}\sum_{i=1}^n m_{\theta,\lambda}(X_i)$. Denote $\theta^* \in \argmin_{\theta\in\Theta}R(\theta,\lambda_{\theta}^*)$.

By Assumption~\ref{ass:threshold_complexity_inclusion} and
Assumption~\ref{ass:uniform_quantile_regularity}(i), the uniform quantile bound in Lemma \ref{lem:lambda_error_bound} implies
\[
\sup_{\theta\in\Theta}
\left|
\hat\lambda_\theta-\lambda_\theta^*
\right|
\leq
\frac{\Delta_n(\eta)}{c_q}
\leq\epsilon_q
\]
with the required probability. Assumption~\ref{ass:uniform_quantile_regularity}(ii) and (iii), together with uniform convergence over
$\gM_{\Delta_n(\eta)/c_q}$, further give
\[
\sup_{\substack{\theta\in\Theta:\\
|\lambda-\lambda_\theta^*|
\leq\Delta_n(\eta)/c_q}}
\left|
\widehat R_n(\theta,\lambda)-R(\theta,\lambda)
\right|
\leq
\mathfrak{C}B
\sqrt{
\frac{
d_{\gM}\log(en/d_{\gM})+\log(4/\eta)
}{n}
}.
\]
Both bounds hold simultaneously with probability at least $1-\eta$.

Let the right-hand side of the last display be $\varepsilon_{\gM,n}$. By the
definition of $\hat\theta$ and the uniform bound above,
\begin{align*}
R(\hat\theta,\hat\lambda_{\hat\theta})
&\leq
\widehat R_n(\hat\theta,\hat\lambda_{\hat\theta})
+\varepsilon_{\gM,n}\\
&\leq
\widehat R_n(\theta^*,\hat\lambda_{\theta^*})
+\varepsilon_{\gM,n}\\
&\leq
R(\theta^*,\hat\lambda_{\theta^*})
+2\varepsilon_{\gM,n}\\
&\leq
R(\theta^*,\lambda_{\theta^*}^*)
+2\varepsilon_{\gM,n}
+\frac{\Gamma}{c_q}\Delta_n(\eta),
\end{align*}
where the last inequality follows from
Assumption~\ref{ass:uniform_quantile_regularity}(ii). Finally, since
$\widehat D^{\texttt{O-DISC}}(X)
=D_{\hat\theta}(X;\hat\lambda_{\hat\theta})$
and $X_{n+1}$ is independent of $\gD_n$, the first and last population
values in the preceding display equal the two expectations in the theorem.
\end{proof}

\subsection{Proof of Theorem \ref{thm:full_conformal_inclusion}}

\begin{proof}
Take \(y=Y_{n+1}\). The augmented ERM objective and deterministic tie-breaking
are permutation invariant in the \(n+1\) augmented observations, so
\(\hat\theta^{Y_{n+1}}\) is a symmetric function of the augmented sample.
Therefore the scores \(\{V_{\hat\theta^{Y_{n+1}}}(X_i,Y_i)\}_{i=1}^{n+1}\) are exchangeable. By the usual full conformal rank argument,
\[
    \Prob\left\{
    V_{\hat\theta^{Y_{n+1}}}(X_{n+1},Y_{n+1})
    \le
    \hat\lambda_{\hat\theta^{Y_{n+1}}}^{Y_{n+1}}
    \right\}
    \ge 1-\alpha .
\]
On this event, the attainability assumption gives
\[
    D_{\hat\theta^{Y_{n+1}}}
    \left(X_{n+1};\hat\lambda_{\hat\theta^{Y_{n+1}}}^{Y_{n+1}}\right)
    \subseteq A(X_{n+1},Y_{n+1}).
\]
The intersection in $\widehat D^{\texttt{FO-DISC}}(X_{n+1})$ contains the term corresponding to
\(y=Y_{n+1}\), so \(\widehat D^{\texttt{FO-DISC}}(X_{n+1})\) is contained in
the true-label-specific set above. This yields the stated probability bound.
\end{proof}

\section{The resource-allocation simulation}
\label{appen:resource_allocation_details}

\subsection{Data-generating process}
\label{appen:resource_dgp}

For the fixed-family benchmark, let \(X\sim\mathrm{Unif}([-1,1]^5)\) and \(Y_j=[m_j(X)+\varsigma_j(X)\xi_j]_+\), \(j=1,2\), where \(\xi_1,\xi_2\) are independent standardized \(t_3\) noises. The baseline location and scale functions are
\[
\begin{aligned}
    m_1(x)&=3.0+x_1+0.75\sin(\pi x_2)+0.25x_1x_2+0.25x_3^2-0.25x_4+0.25\sin(\pi x_5),\\
    m_2(x)&=3.2-0.75x_1+\cos(\pi x_2)+0.25x_1^2+0.25x_3x_4-0.25x_5,\\
    \varsigma_1(x)&=0.40+0.25|x_1|+0.15|x_3|+0.10\mathbbm 1\{x_4>0\},\\
    \varsigma_2(x)&=0.50+0.25|x_2|+0.15|x_5|+0.10\mathbbm 1\{x_1<0\}.
\end{aligned}
\]
An independent pre-training sample of size \(n_{\rm pre}=5000\) is used to fit \(\mathsf{ML}(X)=(\widehat m_1(X),\widehat m_2(X))\): a random forest with 300 trees, minimum leaf size 10, square-root feature subsampling, and nonnegative clipping of predictions. The predictor is held fixed across methods and repetitions. Each repetition uses independent fresh calibration and test samples, with \(n=500\) and \(n_{\rm test}=1000\) in the main fixed-family experiment.

\subsection{Decision-system setting}
\label{appen:resource_decision_setting}

On \(\mathcal U=[0,8]^2\), the two constraint functions are
\[
\begin{aligned}
    f_{\rm res}(u;X,Y)
    &=\max\{Y_1-u_1,\ Y_2-u_2,\ \rho(Y_1+Y_2)-(u_1+u_2)\},\\
    f_{\rm buf}(u;X,Y)
    &=\max\{f_{\rm res}(u;X,Y),\ u_2-Y_2-B\},
    \qquad \rho=1.5,\quad B=3.5.
\end{aligned}
\]
The fixed-family benchmark uses \(\phi_L(u)=c^\top u\), \(c=(1,1.25)\), and \(\phi_Q(u)=0.1(u_1^2+1.5u_2^2)\).

For \texttt{DISC-Add} and \texttt{Base-Add}, set
\(L_{\lambda,f}(u;X)=\hat f(u;X)+\lambda\sigma(u)\), where \(\hat f\) replaces \(Y\) by \(\mathsf{ML}(X)\) and \(\sigma(u)=1+0.1(u_1+u_2)\). Since \(\hat f\) is max-affine and \(\sigma\) is affine, \(D_f(X;\lambda)=\{u\in\mathcal U:L_{\lambda,f}(u;X)\le0\}\) is a polytope.

For \texttt{DISC-PS} and \texttt{Base-PS}, use \(C(X;\lambda)=\{y:\|y-\mathsf{ML}(X)\|_2\le\lambda\}\), \(\lambda\ge0\). Its support function gives
\[
\begin{aligned}
    D_{\rm res}^{\rm PS}(X;\lambda)=\{u\in\mathcal U:\quad
    &u_1\ge\widehat m_1(X)+\lambda,\quad
    u_2\ge\widehat m_2(X)+\lambda,\\
    &u_1+u_2\ge\rho(\widehat m_1(X)+\widehat m_2(X))+\rho\sqrt2\,\lambda\},\\
    D_{\rm buf}^{\rm PS}(X;\lambda)
    &=D_{\rm res}^{\rm PS}(X;\lambda)\cap\{u:u_2\le\widehat m_2(X)-\lambda+B\}.
\end{aligned}
\]

\subsection{Score computation}
\label{appen:resource_score_computation}

Write \(f_i(u)=f(u;X_i,Y_i)\), \(\widehat m_i=\mathsf{ML}(X_i)\), and let \(f\) denote either constraint function.

\paragraph*{Additive residual-margin family.}
The pointwise exclusion formula in \eqref{eq:dir_score_add} gives
\begin{equation}
    V_i^{\texttt{DISC-Add}}
    =\sup_{\substack{u\in\mathcal U\\f_i(u)>0}}
    \frac{-\hat f(u;X_i)}{\sigma(u)}.
    \label{eq:resource_direct_unsafe_score}
\end{equation}
For buffered safety, the unsafe set is the union of the halfspaces
\(u_1<Y_{i1}\), \(u_2<Y_{i2}\), \(u_1+u_2<\rho(Y_{i1}+Y_{i2})\), and \(u_2>Y_{i2}+B\), intersected with \(\mathcal U\); reserve safety omits the fourth halfspace. We enumerate each unsafe halfspace and active affine branch of \(\hat f\). On every nonempty resulting polygon, the objective is linear-fractional with positive denominator, so its supremum is obtained by checking vertices of the polygon's closure \citep{boyd2004convex}.
The baseline score is
\[
    V_i^{\texttt{Base-Add}}
    =\sup_{u\in\mathcal U}
    \frac{f_i(u)-\hat f(u;X_i)}{\sigma(u)}.
\]
It is evaluated by the same vertex calculation after enumerating the active branches of both \(f_i\) and \(\hat f\).

\paragraph*{Prediction-set margin family.}
Define the constraint-specific radius limits
\[
\begin{aligned}
    r_{i,\rm res}(u)
    &=\min\left\{u_1-\widehat m_{i1},\ u_2-\widehat m_{i2},\
    \frac{u_1+u_2-\rho(\widehat m_{i1}+\widehat m_{i2})}{\rho\sqrt2}\right\},\\
    r_{i,\rm buf}(u)
    &=\min\{r_{i,\rm res}(u),\ \widehat m_{i2}+B-u_2\}.
\end{aligned}
\]
Then \eqref{eq:dir_score_PS} yields, for \(k\in\{{\rm res},{\rm buf}\}\),
\[
    V_{i,k}^{\texttt{DISC-PS}}
    =\max\left\{0,\ \sup_{\substack{u\in\mathcal U\\f_k(u;X_i,Y_i)>0}}r_{i,k}(u)\right\}.
\]
For each nonempty unsafe-halfspace intersection, solve a linear program over its closure: maximize an auxiliary variable bounded above by every affine term defining \(r_{i,k}\). Maximize over these intersections and truncate below at zero. This accounts for all candidate constraints and the domain bounds, and gives zero when the unsafe set is empty. The baseline score is \(V_i^{\texttt{Base-PS}}=\|Y_i-\widehat m_i\|_2\).

\subsection{Optimization setup}
\label{appen:resource_family_opt}
\label{appen:resource_optimization}

The optimized-family experiment uses both constraint functions and the asymmetric linear cost \(c^\top u\), \(c=(1,10)\). The context distribution and functions \(m_j,\varsigma_j\) remain as in Appendix~\ref{appen:resource_dgp}, but demands follow
\[
    Y_j=\left[m_j(X)+\varsigma_j(X)
    \left\{0.45\,\frac{Z_j}{\sqrt2}+h_j(X)G_j\right\}\right]_+,
    \qquad j=1,2,
\]
where \(Z_j\sim t_4\), \(G_j\sim\mathrm{Gamma}(\text{shape}=2,\text{scale}=0.65)\), all four variables are mutually independent and independent of \(X\), and
\[
\begin{aligned}
    h_1(x)&=(1.2+0.8|x_2|)\mathbbm 1\{x_1+0.7x_3>0.35\},\\
    h_2(x)&=(1.4+0.9|x_5|)\mathbbm 1\{-x_1+0.8x_4>0.25\}.
\end{aligned}
\]
The Gamma terms add context-dependent positive shocks with different activation regions and amplitudes for the two resources, and the positive part enforces nonnegative demand. Thus \(m_j\) is a baseline location, not the resulting conditional mean. Pre-training, calibration, and test samples are drawn independently from this model.

For each family parameter \(\eta\), compute its directed or baseline scores, obtain the split-conformal threshold \(\hat\lambda_\eta\), and minimize the empirical decision loss
\[
    \frac1n\sum_{i=1}^n
    \min_{u\in D_\eta(X_i;\hat\lambda_\eta)}c^\top u.
\]
The directed and baseline methods share the parameter domain and loss; their calibration scores differ as follows.

\subsubsection{Additive residual-margin family}
The optimized family is
\[
    L_\lambda^\theta(u;X)=\hat f(u;X)+\theta^\top u+\lambda\sigma(u),
    \qquad \theta\in\Theta_{\rm disc}=\{-2,-1,0,1,2\}^2,
\]
with the same \(\sigma(u)=1+0.1(u_1+u_2)\); \(\theta=0\) recovers the fixed family. The scores are
\[
\begin{aligned}
    V_i^{\texttt{DISC-Add}}(\theta)
    &=\sup_{\substack{u\in\mathcal U\\f_i(u)>0}}
    \frac{-\hat f(u;X_i)-\theta^\top u}{\sigma(u)},\\
    V_i^{\texttt{Base-Add}}(\theta)
    &=\sup_{u\in\mathcal U}
    \frac{f_i(u)-\hat f(u;X_i)-\theta^\top u}{\sigma(u)}.
\end{aligned}
\]
Both use the active-branch vertex calculation in Appendix~\ref{appen:resource_score_computation}. The empirical loss is minimized by enumerating the 25 values of \(\theta\).

\subsubsection{Prediction-set margin family}
Use the positive-definite matrix
\[
    \Sigma(a,b)=\begin{pmatrix}e^a&b\\b&(1+b^2)e^{-a}\end{pmatrix},
    \qquad (a,b)\in[-2,2]\times[-3,3],\qquad \det\Sigma=1,
\]
and \(C_\Sigma(X_i;\lambda)=\{y:(y-\widehat m_i)^\top\Sigma(y-\widehat m_i)\le\lambda\}\), \(\lambda\ge0\). Here \(\lambda\) is a squared radius, whereas the fixed-family parameter above is a radius; \((a,b)=(0,0)\) gives the same Euclidean-ball family. Writing \(s_\Sigma(v)=\sqrt{v^\top\Sigma^{-1}v}\), the support function is \(v^\top\widehat m_i+\sqrt\lambda\,s_\Sigma(v)\). Consequently, for buffered safety,
\[
\begin{aligned}
    D_\Sigma(X_i;\lambda)=\{u\in\mathcal U:\quad
    &u_1\ge\widehat m_{i1}+\sqrt\lambda\,s_\Sigma(e_1),\\
    &u_2\ge\widehat m_{i2}+\sqrt\lambda\,s_\Sigma(e_2),\\
    &u_1+u_2\ge\rho(\widehat m_{i1}+\widehat m_{i2})
       +\rho\sqrt\lambda\,s_\Sigma(\mathbf 1),\\
    &u_2\le\widehat m_{i2}-\sqrt\lambda\,s_\Sigma(e_2)+B\}.
\end{aligned}
\]
Reserve safety omits the final constraint. The certificate
\(g_i(\lambda;\Sigma)=\max_{u\in D_\Sigma(X_i;\lambda)}f_i(u)\)
is evaluated at polygon vertices, with \(g_i=-\infty\) for an empty polygon. The directed and baseline scores are
\[
    V_i^{\texttt{DISC-PS}}(\Sigma)=\inf\{\lambda\ge0:g_i(\lambda;\Sigma)\le0\},
    \qquad
    V_i^{\texttt{Base-PS}}(\Sigma)=(Y_i-\widehat m_i)^\top\Sigma(Y_i-\widehat m_i).
\]
The two-dimensional empirical-loss optimization uses Powell's derivative-free method \citep{powell1964efficient}, accommodating the nonsmooth conformal quantile and polygon active-set changes.




\subsection{Additional simulation results}
\label{appen:add_results_resource}

Figure~\ref{fig:resource_alpha} varies the inclusion level in the fixed-family benchmark. Larger \(\alpha\) permits larger decision subsets and lower costs, with the directed methods retaining an efficiency advantage. Figure~\ref{fig:resource_fopt_full_appendix} compares full-conformal optimized calibration, with fixed families as references, to assess the effect of augmented-ERM full conformal inversion.

\begin{figure}[!htbp]
    \centering
    \includegraphics[width=\linewidth]{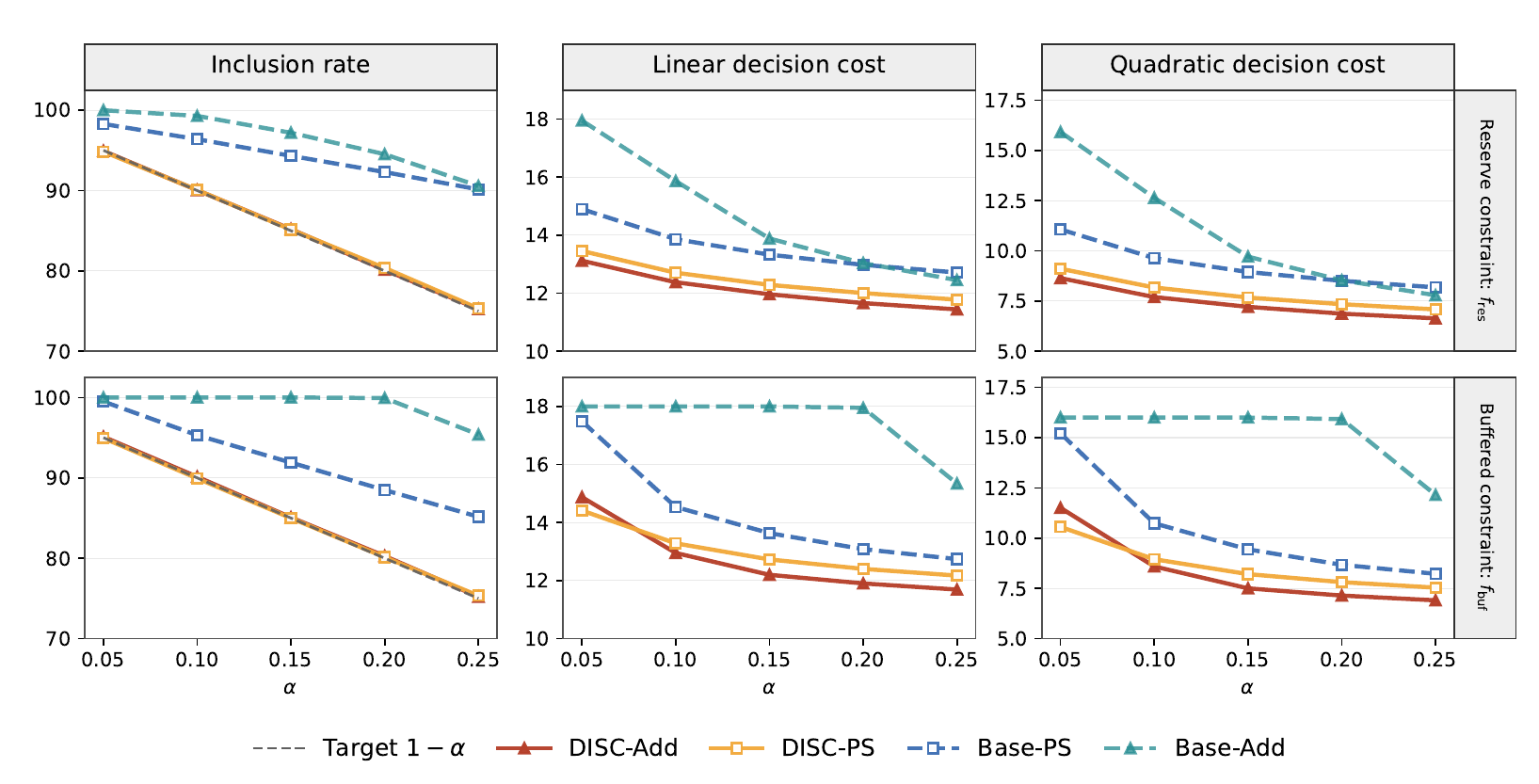}
    \caption{Sensitivity of each method to the inclusion level $ 1-\alpha$ in the fixed-family resource benchmark for $ f_{\rm res}$ and $ f_{\rm buf}$, with $ n=500$. Empty sets are counted as safe in the inclusion rate and are assigned the fallback allocation $ (8,8)$ when reporting decision cost.}
    \label{fig:resource_alpha}
\end{figure}

\begin{figure}[!htbp]
    \centering
    \includegraphics[width=\linewidth]{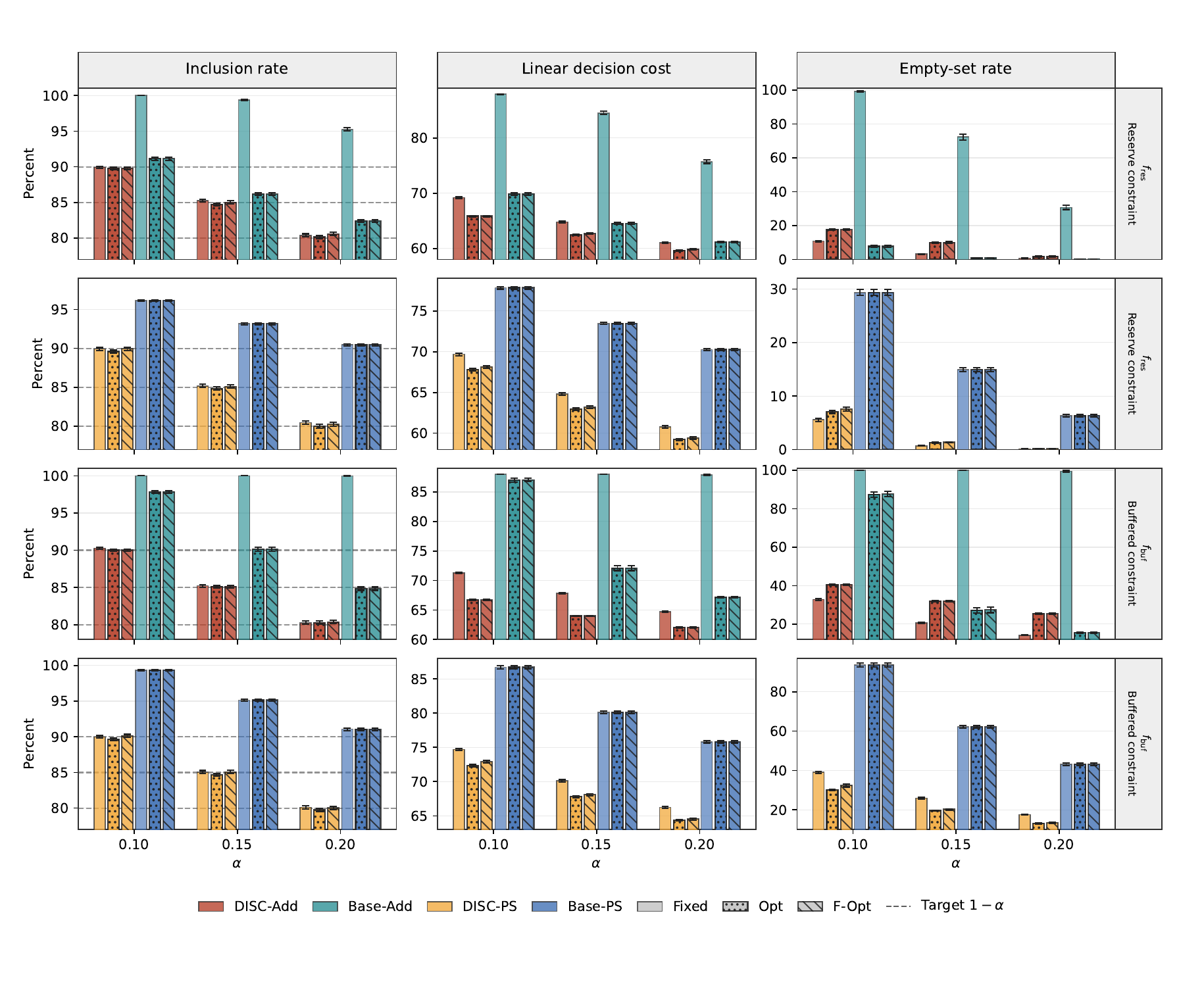}
    \caption{Full-conformal optimized-family benchmark for the resource-allocation experiment. The figure compares fixed, optimized, and full-conformal optimized versions of the additive residual-margin and prediction-set margin families under both \(f_{\rm res}\) and \(f_{\rm buf}\). Results are averaged over 100 repetitions with \(n=500\), \(n_{\rm test}=1000\), \(n_{\rm pre}=5000\), and linear decision cost \(c=(1,10)\).
    }
    \label{fig:resource_fopt_full_appendix}
\end{figure}

\section{The trajectory collision-avoidance simulation}
\label{appen:trajectory_collision_details}

\subsection{Data-generating process}
\label{appen:trajectory_dgp}

Let \(X\sim\mathrm{Unif}([-1,1]^3)\), \(T=10\), and \(\tau_t=(t-1)/(T-1)\). Conditional on \(X=x\), the obstacle trajectory is \(Y_t=m_t(x)+\varepsilon_t(x)+J_t\), where
\[
\begin{aligned}
    m_{t,1}(x)&=-1.2+2.4\tau_t+0.4x_1
        +0.2\sin\{2\pi(\tau_t+0.2x_2)\}+0.2x_3\tau_t,\\
    m_{t,2}(x)&=0.4\sin(\pi\tau_t)+0.6x_2
        +0.3x_3\cos(2\pi\tau_t)-0.2x_1\tau_t.
\end{aligned}
\]
With \(b_t=(1,\sin(\pi\tau_t),\cos(2\pi\tau_t))^\top\), the smooth noise is
\[
\begin{aligned}
    \varepsilon_{t,1}(x)&=s_1(x)a_1^\top b_t+\eta_{t,1},
    &s_1(x)&=0.10+0.10|x_1|,\\
    \varepsilon_{t,2}(x)&=s_2(x)a_2^\top b_t+\eta_{t,2},
    &s_2(x)&=0.15+0.15|x_2|+0.05\mathbbm 1\{x_3>0\}.
\end{aligned}
\]
The vectors \(a_1,a_2\in\mathbb R^3\) have independent \(t_4/\sqrt2\) entries, and \(\eta_t\sim N(0,0.04^2I_2)\) are independent local jitters. With probability \(0.08\), independently of \(X\), draw an equiprobable sign \(S\in\{-1,1\}\) and set \(J_t=(0,S\{0.50+0.50\sin(\pi\tau_t)\})\) for all \(t\); otherwise \(J_t=0\). An independent sample of size \(n_{\rm pre}=3000\) trains a GRU or LSTM predictor \(\widehat Y=\mathsf{ML}(X)\), held fixed across calibration-test repetitions and methods.

\subsection{Decision-system setting}
\label{appen:trajectory_decision_setting}

For \(u\in\mathcal U=[-3,3]^2\), the ego trajectory is \(Z_t(u;X)=a_t(X)+B_tu\), with \(h_t=\tau_t(1-\tau_t)\), \(s_t=\sin(\pi\tau_t)\), and
\[
    a_t(X)=\begin{pmatrix}
    -1.5+3.0\tau_t+0.2X_1+0.1X_3\tau_t\\
    -0.5+1.0\tau_t+0.2X_2-0.1X_1\tau_t
    \end{pmatrix},
    \qquad
    B_t=\begin{pmatrix}0.5h_t&0.3s_t\\0.2h_t&0.6s_t\end{pmatrix}.
\]
Define \(d_{it}(u)=\|Z_t(u;X_i)-Y_{it}\|_2^2\) and \(\widehat d_{it}(u)=\|Z_t(u;X_i)-\widehat Y_{it}\|_2^2\), and write
\[
\begin{aligned}
    d_{i,\rm avg}(u)&=T^{-1}\sum_t d_{it}(u),
    &d_{i,\min}(u)&=\min_t d_{it}(u),\\
    \widehat d_{i,\rm avg}(u)&=T^{-1}\sum_t\widehat d_{it}(u),
    &\widehat d_{i,\min}(u)&=\min_t\widehat d_{it}(u).
\end{aligned}
\]
For \(r\in\{{\rm avg},\min\}\), the constraint function is \(f_r(u;X_i,Y_i)=\omega-d_{i,r}(u)\), with baseline \(\omega=0.25\). The average distances are convex quadratics; the minimum distances are minima of timewise convex quadratics. Decision costs and empty-set handling follow Section~\ref{sec:simulation}.

The additive family is \(L_{\lambda,r}(u;X_i)=\omega-\widehat d_{i,r}(u)+\lambda\sigma(u)\). For the prediction-set family, use the RMS ball
\(C_{\rm avg}(X_i;\lambda)=\{y:(T^{-1}\sum_t\|y_t-\widehat Y_{it}\|_2^2)^{1/2}\le\lambda\}\)
or the sup-time ball
\(C_{\min}(X_i;\lambda)=\{y:\max_t\|y_t-\widehat Y_{it}\|_2\le\lambda\}\), respectively. For \(\lambda\ge0\), both induce
\[
    D_r^{\rm PS}(X_i;\lambda)
    =\{u\in\mathcal U:\sqrt{\widehat d_{i,r}(u)}\ge\sqrt\omega+\lambda\}.
\]

\subsection{Score computation}
\label{appen:trajectory_score_computation}

The additive scale is specified separately for the main-text and supplementary trajectory experiments:
\[
\begin{aligned}
    \sigma_{\rm main}(u)&=1+0.20(u_1+1)^2+0.05(u_2-0.5)^2,\\
    \sigma_{\rm old}(u)&=0.75+0.20(u_1+0.8)^2+0.045(u_2-0.4)^2.
\end{aligned}
\]
The main-text trajectory figures use \(\sigma=\sigma_{\rm main}\); the supplementary trajectory figures in this appendix use \(\sigma=\sigma_{\rm old}\). All additive scores and certificate bounds below use the scale corresponding to the experiment. The prediction-set families do not depend on \(\sigma\).

For either \(r\in\{{\rm avg},\min\}\), the exact directed scores are
\[
\begin{aligned}
    V_{i,r}^{\texttt{DISC-Add}}
    &=\sup_{\substack{u\in\mathcal U\\d_{i,r}(u)<\omega}}
    \frac{\widehat d_{i,r}(u)-\omega}{\sigma(u)},\\
    V_{i,r}^{\texttt{DISC-PS}}
    &=\left[\sup_{\substack{u\in\mathcal U\\d_{i,r}(u)<\omega}}
    \sqrt{\widehat d_{i,r}(u)}-\sqrt\omega\right]_+.
\end{aligned}
\]
For minimum-distance safety, the unsafe set is \(\bigcup_t\{u\in\mathcal U:d_{it}(u)<\omega\}\).
The baseline scores are
\[
\begin{aligned}
    V_{i,r}^{\texttt{Base-Add}}
    &=\sup_{u\in\mathcal U}\frac{\widehat d_{i,r}(u)-d_{i,r}(u)}{\sigma(u)},\\
    V_{i,\rm avg}^{\texttt{Base-PS}}
    &=\left(T^{-1}\sum_t\|Y_{it}-\widehat Y_{it}\|_2^2\right)^{1/2},
    \qquad
    V_{i,\min}^{\texttt{Base-PS}}=\max_t\|Y_{it}-\widehat Y_{it}\|_2.
\end{aligned}
\]
The directed scores used in the experiments are conservative upper bounds on the exact scores, computed as follows.

\subsection{Conservative box certificate}
\label{appen:trajectory_box_certificate}

Let \(\{B_m\}_{m=1}^M\) be the axis-aligned box cover of \(\mathcal U\) induced by the certificate grid, with default resolution \(321\times321\) unless a sensitivity experiment specifies otherwise. On each box, compute
\[
\begin{aligned}
    d^-_{i,\rm avg,m}&\le\inf_{u\in B_m}d_{i,\rm avg}(u),
    &\widehat d^+_{i,\rm avg,m}&\ge\sup_{u\in B_m}\widehat d_{i,\rm avg}(u),\\
    d^-_{itm}&\le\inf_{u\in B_m}d_{it}(u),
    &\widehat d^+_{itm}&\ge\sup_{u\in B_m}\widehat d_{it}(u),\\
    0<\sigma^-_m&\le\inf_{u\in B_m}\sigma(u),
    &\sigma^+_m&\ge\sup_{u\in B_m}\sigma(u).
\end{aligned}
\]
These are bounds on convex quadratics. Their box minima are obtained by checking feasible stationary points, clipped one-dimensional edge minimizers, and corners; their maxima occur at vertices. To handle both safety criteria, define
\[
\begin{aligned}
    \ell_{im}^{\rm avg}&=d^-_{i,\rm avg,m},
    &H_{im}^{\rm avg}&=\widehat d^+_{i,\rm avg,m},\\
    \ell_{im}^{\min}&=\min_t d^-_{itm},
    &H_{im}^{\min}&=\min_t\widehat d^+_{itm}.
\end{aligned}
\]
Thus \(\ell_{im}^r\le d_{i,r}(u)\) and \(\widehat d_{i,r}(u)\le H_{im}^r\) throughout \(B_m\), including for the minimum of timewise quadratics. Retain the potentially unsafe boxes \(\mathcal I_i^r=\{m:\ell_{im}^r\le\omega\}\). The two certificates are
\[
\begin{aligned}
    \widetilde V_{i,r}^{\texttt{DISC-Add}}
    &=\max_{m\in\mathcal I_i^r}
    \begin{cases}
    (H_{im}^r-\omega)/\sigma^-_m,&H_{im}^r\ge\omega,\\
    (H_{im}^r-\omega)/\sigma^+_m,&H_{im}^r<\omega,
    \end{cases}\\
    \widetilde V_{i,r}^{\texttt{DISC-PS}}
    &=\left[\max_{m\in\mathcal I_i^r}\sqrt{H_{im}^r}-\sqrt\omega\right]_+.
\end{aligned}
\]
If \(\mathcal I_i^r=\varnothing\), use \(-\infty\) and \(0\), respectively. The retained boxes cover every unsafe control, and the sign-dependent denominator gives an upper bound on the additive ratio. Hence both certificates satisfy \(\widetilde V_{i,r}\ge V_{i,r}\), so Theorem~\ref{thm:inclusion_certificate} applies to their conformal calibration.

\begin{figure}[!htbp]
    \centering
    \includegraphics[width=\linewidth]{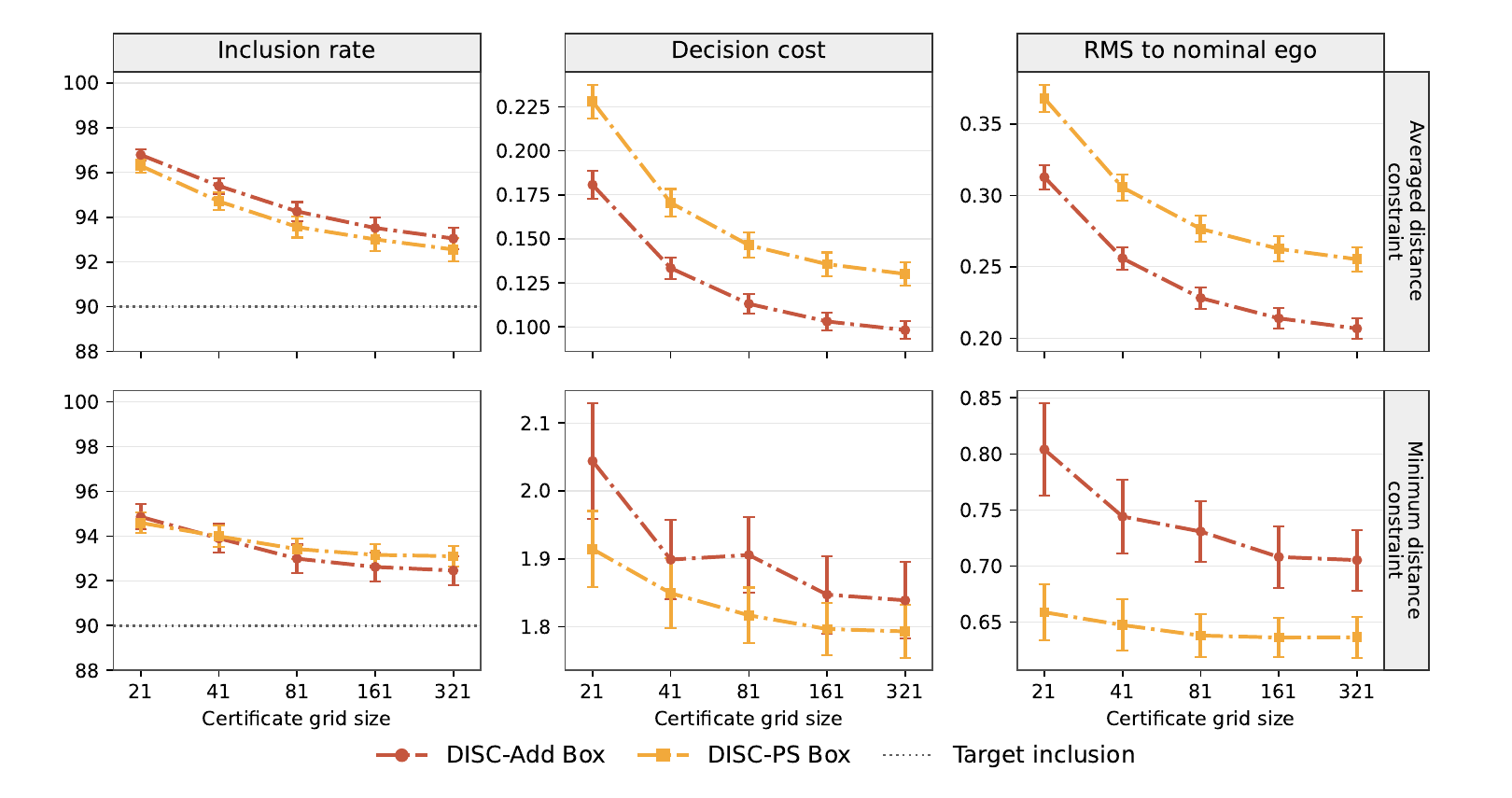}
    \caption{Sensitivity of the box certificate to the grid resolution. Rows correspond to the average-distance and minimum-distance safety criteria, both at \(\alpha=0.1\), and columns report inclusion rate, quadratic decision cost, and RMS distance to the nominal ego trajectory as the certificate grid is refined over \(21,41,81,161,321\).}
    \label{fig:trajectory_certificate_sensitivity_appendix}
\end{figure}

Figure~\ref{fig:trajectory_certificate_sensitivity_appendix} varies the certificate resolution while holding the statistical setting fixed. Inclusion remains close to the target across the tested resolutions; decision costs and RMS deviations stabilize as finer boxes reduce numerical conservatism.

\subsection{Additional simulation results}
\label{appen:trajectory_evaluation}

Figures~\ref{fig:trajectory_alpha_sensitivity_gru_appendix} and \ref{fig:trajectory_min_omega_sensitivity_gru_appendix} give GRU risk-level and safety-budget sweeps corresponding to the LSTM studies in Figures~\ref{fig:collision_trajectory_barplot} and \ref{fig:trajectory_min_omega_sensitivity_main}, with the supplementary additive scale specified above. They retain the qualitative pattern of inclusion near the target and lower decision costs for directed calibration.

\begin{figure}[H]
    \centering
    \includegraphics[width=\linewidth]{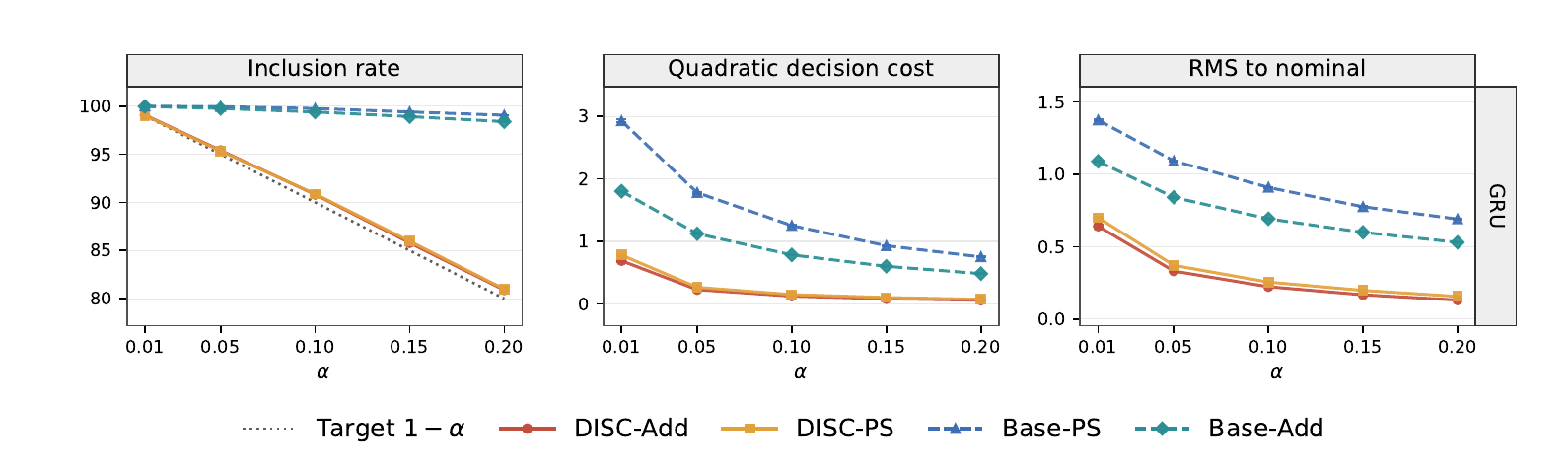}
    \caption{Sensitivity to the safety level \(\alpha\) in the average-distance safety setting with the GRU trajectory predictor, using \(n=500\), safety budget \(\omega=0.25\), and 100 repetitions. Columns report inclusion rate, decision cost, and RMS distance to the nominal ego trajectory. Empty decision sets are assigned the fallback decision cost \(4.5\); RMS summaries are computed over nonempty decision sets.}
    \label{fig:trajectory_alpha_sensitivity_gru_appendix}
\end{figure}

\begin{figure}[H]
    \centering
    \includegraphics[width=.95\linewidth]{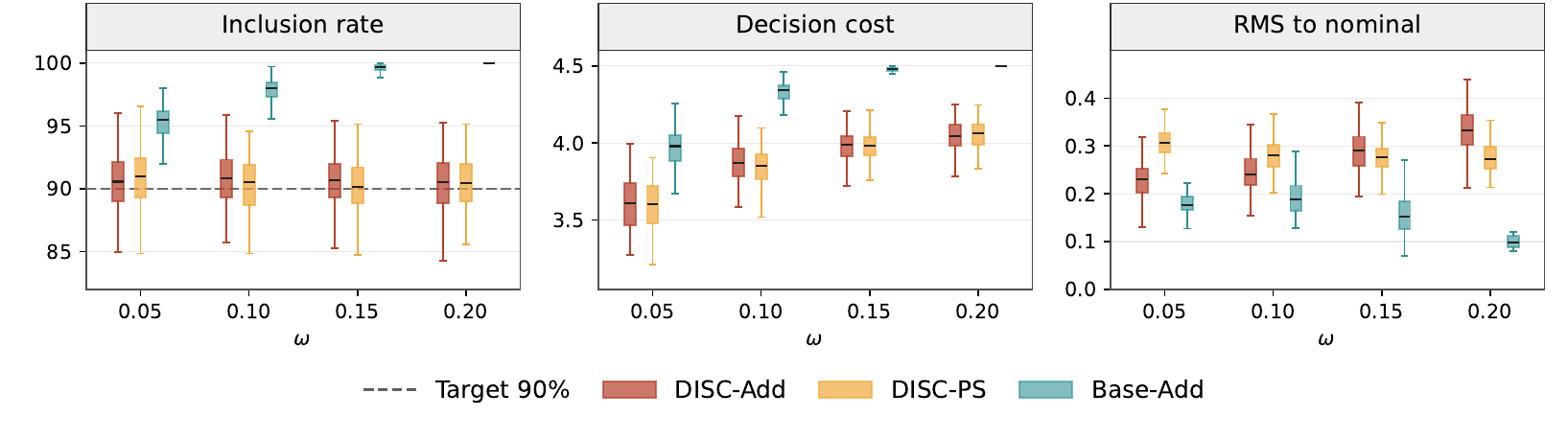}
    \caption{Sensitivity to the safety budget \(\omega\) in the minimum-distance safety setting with the GRU trajectory predictor, using \(n=200\), \(\alpha=0.1\), and 100 repetitions. Columns report inclusion rate, decision cost, and RMS distance to the nominal ego trajectory for \(\omega\in\{0.05,0.10,0.15,0.20\}\).}
    \label{fig:trajectory_min_omega_sensitivity_gru_appendix}
\end{figure}

The horizon sweep varies \(T\in\{5,10,15,20,25\}\), and the calibration-size sweep varies \(n\in\{50,100,200,400,600\}\). Figures~\ref{fig:trajectory_T_sensitivity_appendix} and \ref{fig:trajectory_n_cal_sensitivity_appendix} show that the directed methods retain their cost and RMS advantages across these settings; larger calibration samples mainly stabilize the repetition-level distributions.

\begin{figure}[!htbp]
    \centering
    \includegraphics[width=\linewidth]{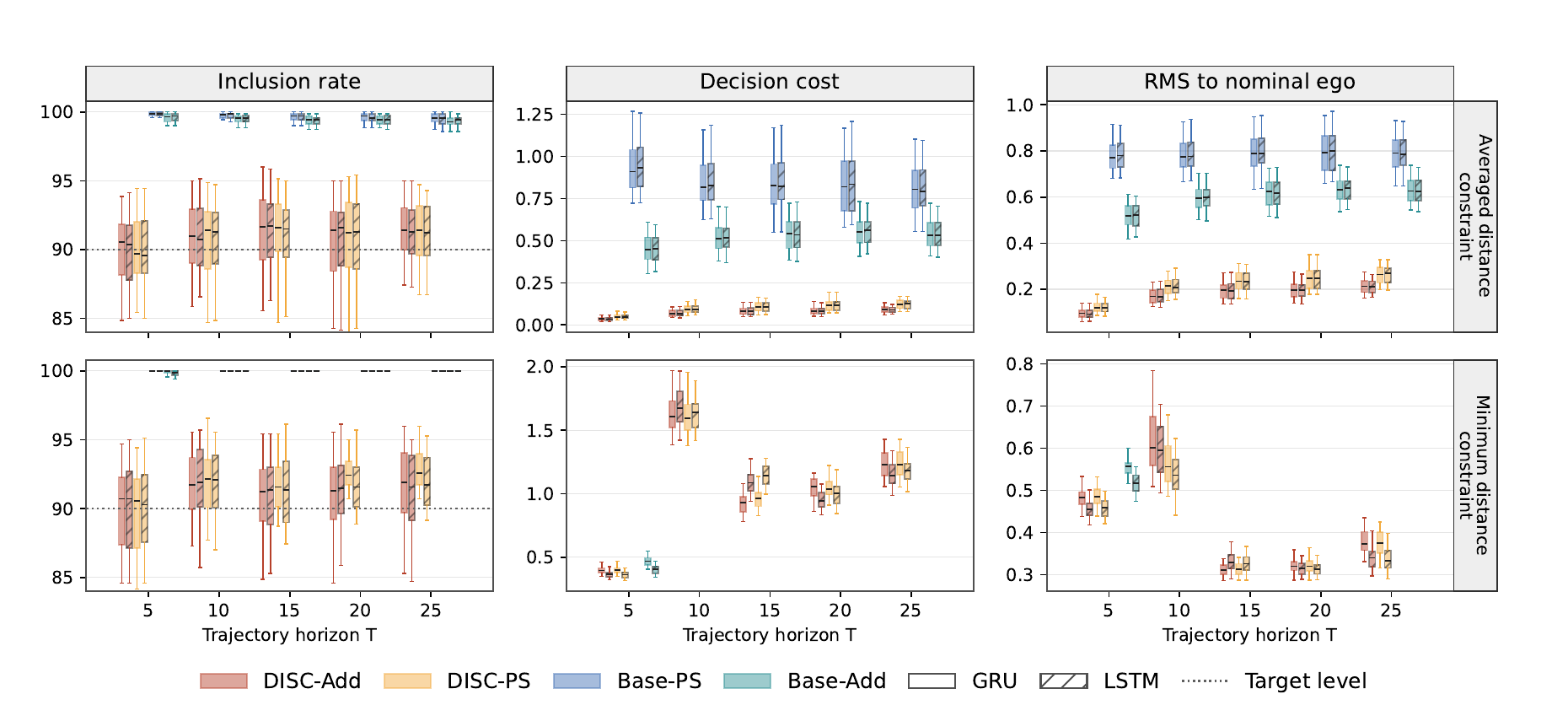}
    \caption{Sensitivity to the trajectory horizon \(T\), using \(\omega=0.25\), \(n=100\), \(\alpha=0.1\), and 100 repetitions.}
    \label{fig:trajectory_T_sensitivity_appendix}
\end{figure}

\begin{figure}[!htbp]
    \centering
    \includegraphics[width=\linewidth]{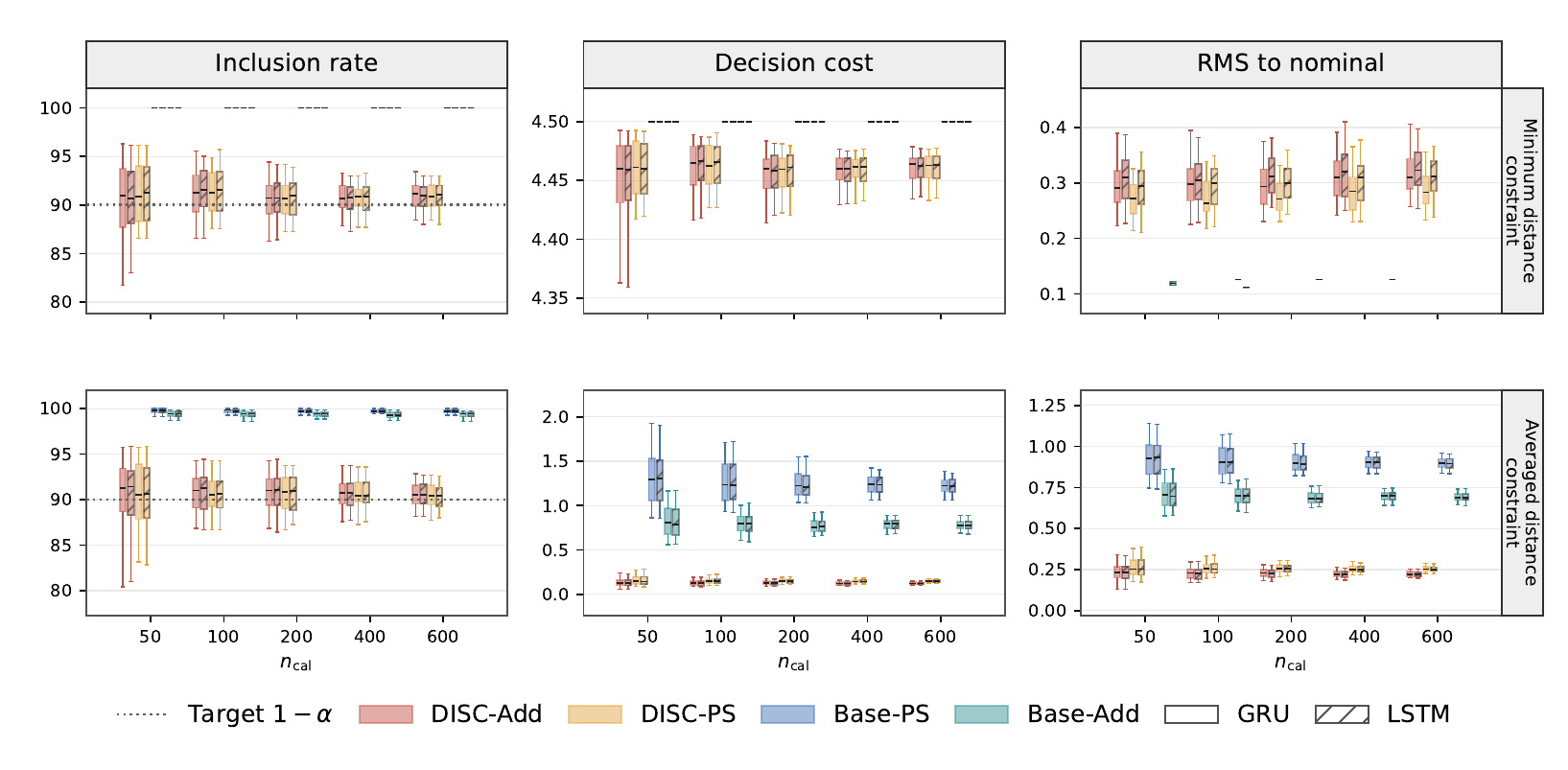}
    \caption{Sensitivity to the calibration size \(n\), using \(\omega=0.25\), \(\alpha=0.1\), grid size \(41\), and 100 repetitions.}
    \label{fig:trajectory_n_cal_sensitivity_appendix}
\end{figure}

\section{The LLM reasoning-chain verification experiment}
\label{appen:llm_cot_details}

This appendix describes implementation details of experiments in Section~\ref{sec:exp_llm_reasoning}. We use additive residual-margin families throughout, since there is no tractable conformal prediction set for the latent reasoning chain that can be propagated through the constraint function in closed form. We compare four directed-calibration procedures, \texttt{DISC-Const}, \texttt{DISC-Mean}, \texttt{DISC-Linear}, and \texttt{DISC-Sqrt}, with their uniform-residual counterparts, \texttt{Base-Const}, \texttt{Base-Mean}, \texttt{Base-Linear}, and \texttt{Base-Sqrt}.

\subsection{Detailed settings}

\phantomsection\label{appen:llm_graph_construction}
\paragraph*{Generation, decomposition, and candidate output family.}
We use Qwen-Flash for response generation, claim decomposition, dependency graph induction, claim confidence scoring, and reference-based claim judging, with stage-specific prompts. Following \citet{rubin2025conformal}, we represent each generated response as a reasoning graph. For each prompt $X$, a base LLM first generates one complete response by a fixed ``solve step by step'' instruction. We then decompose the response into the ordered subclaims $(\mathsf{LLM}_t(X))_{1\leq t\leq T(X)}$, where $T(X)$ is the number of subclaims. Next, an LLM induces a directed dependency graph $G(X)=([T(X)],E(X))$ over these subclaims, where $[T(X)]$ is the set of all subclaim indices and $E(X)$ is the set of directed edges. An edge $(s,t)\in E(X)$, written $s\to t$, means that subclaim $\mathsf{LLM}_t(X)$ directly depends on subclaim $\mathsf{LLM}_s(X)$. The graph-induction prompt and its in-context examples are reported in Appendix~\ref{appen:prompt_graph_induction}.
For each $t\in[T(X)]$, define its strict ancestor set by
\begin{gather*}
\operatorname{Anc}(t;G(X))=
\left\{s\in[T(X)]\setminus\{t\}:\begin{array}{l}
\text{there exist }m\geq1\text{ and }v_0=s,v_1,\ldots,v_m=t\\[-2pt]
\text{such that }(v_{\ell-1},v_\ell)\in E(X)\text{ for every }\ell\in[m]
\end{array}\right\}.
\end{gather*}
Thus, $s\in\operatorname{Anc}(t;G(X))$ exactly when a directed path from $\mathsf{LLM}_s(X)$ to $\mathsf{LLM}_t(X)$ exists; direct and transitive dependencies are both included. For $\mathcal S\subseteq[T(X)]$, write $u_{\mathcal S}=(\mathsf{LLM}_t(X))_{t\in\mathcal S}$ for the output retaining the subclaims indexed by $\mathcal S$ in their original order. Since the space of all possible reasoning chains is too large to search, we restrict attention to the dependency-closed retained outputs
\begin{gather*}
\gU(X)=
\left\{u_{\mathcal S}:\mathcal S\subseteq[T(X)],\ 
\operatorname{Anc}(t;G(X))\subseteq\mathcal S\ \text{for every }t\in\mathcal S \right\}.
\end{gather*}
Thus, the decision is not to edit individual subclaims, but to select a dependency-closed subset of the generated reasoning chain. The complete response is $u_{[T(X)]}$.

\paragraph*{Reference-based claim labels and the oracle feasible set.}
The MATH data provide a problem statement and one official solution, but they do not contain a unique ground-truth chain of thought. We therefore obtain subclaim-level reference labels from a fixed reference-judge LLM that receives the problem, the official solution, the ordered subclaims, and their dependency graph; Appendix~\ref{appen:prompt_claim_labeling} gives the complete prompt. Let $Y=(Y_t)_{1\leq t\leq T(X)}$ denote the resulting labels, where $Y_t=1$ means that $\mathsf{LLM}_t(X)$ is judged mathematically acceptable and $Y_t=0$ otherwise. For $u_{\mathcal S}\in\gU(X)$, define the constraint function by
\begin{gather*}
f(u_{\mathcal S};X,Y)=1-2\min_{t\in\mathcal S}Y_t\in\{-1,1\}.
\end{gather*}
We adopt the convention $\min_{t\in\emptyset}Y_t=1$, so the empty retained output contains no unacceptable subclaim and has $f(u_\emptyset;X,Y)=-1$. Hence $f(u_{\mathcal S};X,Y)\leq0$ if and only if every retained subclaim receives an acceptable label from the fixed reference-judge pipeline. The oracle feasible set is therefore $A(X,Y)=\{u_{\mathcal S}\in\gU(X):f(u_{\mathcal S};X,Y)\leq0\}$,
namely, the set of dependency-closed retained outputs that contain no unacceptable claim.

\paragraph*{Predicted claim scores and subset scores.}
On the prediction side, a verifier LLM receives the problem, the ordered subclaims, and the dependency graph, and returns a confidence score $\widehat Y_t(X)\in[0,1]$ for each $\mathsf{LLM}_t(X)$, where larger values indicate greater confidence that the subclaim is acceptable. To reduce overconfident scoring, the verifier is instructed to reserve values near one for clearly correct subclaims and values near zero for clearly incorrect subclaims, while using intermediate values for uncertain or underspecified steps. Replacing $Y_t$ in the constraint function with $\widehat Y_t(X)$ gives the pretrained surrogate constraint function
\begin{gather*}
\hat f(u_{\mathcal S};X)=1-2\min_{t\in\mathcal S}\widehat Y_t(X)\in[-1,1].
\end{gather*}
We similarly take the minimum over the empty set to be one, so that $\hat f(u_\emptyset;X)=-1$. Thus, a retained output receives a large surrogate constraint value as soon as one of its constituent subclaims is judged unreliable.

\paragraph*{Directed additive residual-margin families.}
We use the common additive residual-margin family
\begin{gather*}
    D(X;\lambda,\mu)=\{u_\emptyset\}\cup\{u_{\mathcal S}\in\gU(X):\hat f(u_{\mathcal S};X)+\mu(u_{\mathcal S};X)+\lambda\leq0\},
\end{gather*}
where $\mu(u_{\mathcal S};X)$ is an additive adjustment. The constant specification is $\mu_{\rm const}(u_{\mathcal S};X)\equiv0$. Directed calibration with this specification gives \texttt{DISC-Const}, which is equivalent to the thresholding rule of \citet{rubin2025conformal}. The three adaptive specifications used by \texttt{DISC-Mean}, \texttt{DISC-Linear}, and \texttt{DISC-Sqrt} are defined below.

\paragraph*{Safe feasible set construction and calibration.}
For calibration data $\{(X_i,Y_i)\}_{i=1}^{n}$, let $T_i=T(X_i)$ and write $Y_i=(Y_{it})_{1\leq t\leq T_i}$. For sample $i$, the notation $u_{\mathcal S}$ refers to $(\mathsf{LLM}_t(X_i))_{t\in\mathcal S}$. As in Section~\ref{sec:method}, the directed calibration score for an adjustment $\mu$ is
\begin{align*}
V_i^{\texttt{DISC}}(\mu)
&=\inf\{\lambda\in\sR:D(X_i;\lambda,\mu)\subseteq A(X_i,Y_i)\}\\
&=\inf\left\{\lambda\in\sR:
\begin{array}{l}
\min_{t\in\mathcal S}Y_{it}=1\ \text{for every }u_{\mathcal S}\in\gU(X_i)\text{ such that}\\[-2pt]
\displaystyle \min_{t\in\mathcal S}\widehat Y_t(X_i)\geq\frac{1+\mu(u_{\mathcal S};X_i)+\lambda}{2}
\end{array}
\right\}.
\end{align*}
For the constant specification $\mu_{\rm const}$, let $\mathcal B_i=\{t\in[T_i]:Y_{it}=0\}$ be the index set of unacceptable subclaims, and let $\mathcal C_i(t)=\operatorname{Anc}(t;G(X_i))\cup\{t\}$ be the minimal dependency-closed index set containing $t$. 
If $\mathcal B_i=\emptyset$, every candidate output is acceptable and hence $V_i^{\texttt{DISC}}(\mu_{\rm const})=-\infty$. 
Otherwise, every unsafe candidate output $u_{\mathcal S}$ has $\mathcal C_i(t)\subseteq\mathcal S$ for at least one $t\in\mathcal B_i$. Since adding subclaims can only decrease the bottleneck confidence score, the largest exclusion threshold is attained by one of these minimal unsafe outputs. Therefore,
\begin{gather*}
V_i^{\texttt{DISC}}(\mu_{\rm const})=\max_{t\in\mathcal B_i}\left\{2\min_{s\in\mathcal C_i(t)}\widehat Y_s(X_i)-1\right\},\qquad \mathcal B_i\neq\emptyset.
\end{gather*}
For the adaptive \texttt{DISC} variants, let $d(X)\in[n_d]$ denote the benchmark-provided difficulty level, where level $1$ is the easiest and level $n_d$ is the hardest. We use the mean of the finite directed exclusion thresholds $V_i^{\texttt{DISC}}(\mu_{\rm const})$ within each difficulty level to construct a difficulty-specific additive adjustment. Specifically, define
\begin{gather*}
\overline\mu_g=\dfrac{\sum_{i=1}^{n}V_i^{\texttt{DISC}}(\mu_{\rm const})\mathbbm{1}\left(V_i^{\texttt{DISC}}(\mu_{\rm const})>-\infty,\ d(X_i)=g\right)}{\sum_{i=1}^{n}\mathbbm{1}\left(V_i^{\texttt{DISC}}(\mu_{\rm const})>-\infty,\ d(X_i)=g\right)},\qquad g\in[n_d],
\end{gather*}
where we adopt the convention $(-\infty)\cdot0=0$. 
If a difficulty group contains no finite constant score $V_i^{\texttt{DISC}}(\mu_{\rm const})$, we use the mean over all finite \texttt{DISC-Const} scores as $\overline\mu_g=\frac{\sum_{i=1}^{n}V_i^{\texttt{DISC}}(\mu_{\rm const})\mathbbm{1}\left( V_i^{\texttt{DISC}}(\mu_{\rm const})>-\infty \right)}{\sum_{i=1}^{n}\mathbbm{1}\left( V_i^{\texttt{DISC}}(\mu_{\rm const})>-\infty\right)}$. We then consider three choices of $\mu$. The first uses only the difficulty group:
\begin{gather*}
\mu_{\rm mean}(u_{\mathcal S};X)=\overline\mu_{d(X)}.
\end{gather*}
The second adds a linearly centered length correction:
\begin{gather*}
\mu_{\rm lin}(u_{\mathcal S};X)=\overline\mu_{d(X)}-\left|\overline\mu_{d(X)}\right|\left(\frac{|\mathcal S|}{T(X)}-\frac12\right),
\end{gather*}
and the third replaces the linear term by a square-root correction:
\begin{gather*}
\mu_{\rm sqrt}(u_{\mathcal S};X)=\overline\mu_{d(X)}-\left|\overline\mu_{d(X)}\right|\left(\sqrt{\frac{|\mathcal S|}{T(X)}}-\sqrt{\frac12}\right).
\end{gather*}
Both length-dependent adjustments are nonincreasing in the retained fraction and equal $\overline\mu_{d(X)}$ at $|\mathcal S|/T(X)=1/2$, so they relax the margin condition for longer outputs relative to the mean adjustment at a fixed $\lambda$.
Because $\overline\mu_g$ and the final conformal quantile are estimated from the same calibration sample, these three variants have the same-sample adaptive form covered by the analysis in Theorem~\ref{thm:direct_concentration_inclusion}; \texttt{DISC-Const} instead uses a family fixed before calibration. For a subset-dependent adjustment, we check all unsafe dependency-closed outputs to determine the score. Thus, $V_i^{\texttt{DISC}}(\mu)=-\infty$ when $\mathcal B_i=\emptyset$, while
\begin{gather*}
V_i^{\texttt{DISC}}(\mu)=\max_{u_{\mathcal S}\in\gU(X_i):\mathcal S\cap\mathcal B_i\neq\emptyset}\left\{2\min_{t\in\mathcal S}\widehat Y_t(X_i)-1-\mu(u_{\mathcal S};X_i)\right\},\qquad \mathcal B_i\neq\emptyset.
\end{gather*}
The directed threshold is $\hat\lambda_{\texttt{DISC},\mu}=\mathsf{Q}_{1-\alpha}\LRs{\{V_i^{\texttt{DISC}}(\mu)\}_{i=1}^{n},\infty}$. For an unsafe output $u_{\mathcal S}$, exclusion from $D(X_i;\lambda,\mu)$ requires $\lambda>2\min_{t\in\mathcal S}\widehat Y_t(X_i)-1-\mu(u_{\mathcal S};X_i)$. Thus, the infimum defining the score need not itself exclude all unsafe outputs. We remove outputs on the margin boundary and return
\begin{gather*}
\widehat D^{\texttt{DISC}}(X;\mu)
=\{u_\emptyset\}\cup
\left\{u_{\mathcal S}\in\gU(X):
\hat f(u_{\mathcal S};X)+\mu(u_{\mathcal S};X)+\hat\lambda_{\texttt{DISC},\mu}<0
\right\}.
\end{gather*}
The strict inequality implements the boundary-exclusion convention while retaining the empty output.

For \texttt{DISC-Const} and \texttt{DISC-Mean}, the adjustment does not depend on the retained-subset size, so the score can be computed by checking the minimal dependency closure associated with each unacceptable subclaim. For the length-dependent \texttt{DISC-Linear} and \texttt{DISC-Sqrt} variants, the score is instead computed over all unsafe dependency-closed candidate outputs, as specified above.
The inclusion check is carried out over dependency-closed candidate outputs induced by the claim graph, using the boundary-exclusion convention above.

\paragraph*{Uniform-residual baselines.}
For comparison, we calibrate the same four margin specifications using the uniform residual score
\begin{gather*}
V_i^{\texttt{Base}}(\mu)=\sup_{u_{\mathcal S}\in\gU(X_i)}\left\{f(u_{\mathcal S};X_i,Y_i)-\hat f(u_{\mathcal S};X_i)-\mu(u_{\mathcal S};X_i)\right\}.
\end{gather*}
Using $\mu_{\rm const}$ gives \texttt{Base-Const}, while $\mu_{\rm mean}$, $\mu_{\rm lin}$, and $\mu_{\rm sqrt}$ give \texttt{Base-Mean}, \texttt{Base-Linear}, and \texttt{Base-Sqrt}, respectively. For a given $\mu$, the threshold and safe feasible set are
\begin{gather*}
\hat\lambda_{\texttt{Base},\mu}=\mathsf{Q}_{1-\alpha}\LRs{\{V_i^{\texttt{Base}}(\mu)\}_{i=1}^{n},\infty},
\qquad
\widehat D^{\texttt{Base}}(X;\mu)=D(X;\hat\lambda_{\texttt{Base},\mu},\mu).
\end{gather*}
Substituting the definitions of $f$ and $\hat f$, the base score can be written explicitly as
\begin{gather}
V_i^{\texttt{Base}}(\mu) = \max_{u_{\mathcal S}\in\gU(X_i)} \left\{ 2\min_{t\in\mathcal S}\widehat Y_t(X_i) - 2\min_{t\in\mathcal S}Y_{it} - \mu(u_{\mathcal S};X_i) \right\}.\label{eq:llm_base_score}
\end{gather}
Thus, unlike the directed score, the base score maximizes over both safe and unsafe dependency-closed outputs. For any unsafe $u_{\mathcal S}\in\gU(X_i)$, the expression in \eqref{eq:llm_base_score} reduces to $2\min_{t\in\mathcal S}\widehat Y_t(X_i)-\mu(u_{\mathcal S};X_i)$. Hence $V_i^{\texttt{Base}}(\mu)\geq V_i^{\texttt{DISC}}(\mu)+1$ when $\mathcal B_i\neq\emptyset$. When $\mathcal B_i=\emptyset$, the directed score is $-\infty$. Consequently, if the two procedures use the same fitted adjustment $\mu$ and calibration sample, their finite thresholds satisfy $\hat\lambda_{\texttt{Base},\mu}\geq\hat\lambda_{\texttt{DISC},\mu}+1$. Every nonempty output $u_{\mathcal S}\in\widehat D^{\texttt{Base}}(X;\mu)$ therefore satisfies
\begin{gather*}
\hat f(u_{\mathcal S};X)+\mu(u_{\mathcal S};X)+\hat\lambda_{\texttt{DISC},\mu}
\leq \hat\lambda_{\texttt{DISC},\mu}-\hat\lambda_{\texttt{Base},\mu}
\leq-1<0.
\end{gather*}
Thus, it also satisfies the strict margin condition for $\widehat D^{\texttt{DISC}}(X;\mu)$. Since both subsets contain $u_\emptyset$, we obtain $\widehat D^{\texttt{Base}}(X;\mu)\subseteq\widehat D^{\texttt{DISC}}(X;\mu)$. The same containment holds at infinite thresholds: a directed threshold of $-\infty$ retains all candidate outputs, while a base threshold of $+\infty$ retains only $u_\emptyset$.

\paragraph*{Method comparison.}
The experiment compares directed and uniform-residual calibration under four shared adjustment specifications:
\begin{itemize}
    \item \texttt{DISC-Const} and \texttt{Base-Const} use $\mu_{\rm const}$;
    \item \texttt{DISC-Mean} and \texttt{Base-Mean} use $\mu_{\rm mean}$;
    \item \texttt{DISC-Linear} and \texttt{Base-Linear} use $\mu_{\rm lin}$;
    \item \texttt{DISC-Sqrt} and \texttt{Base-Sqrt} use $\mu_{\rm sqrt}$.
\end{itemize}
These comparisons assess the effects of the calibration rule and the choice of adjustment. Table~\ref{tab:llm_reasoning_resultsp05} reports the four directed methods at $\alpha=0.05$. The uniform-residual methods are omitted from the table because they return empty outputs and achieve an empirical inclusion rate of 1.

\begin{table}[htbp]
    \centering
    \caption{LLM reasoning-chain verification on MATH at $\alpha=0.05$. Boldface and underlining indicate the largest and second-largest values among the retention metrics.}
    \label{tab:llm_reasoning_resultsp05}
    \begin{tabular}{lcccc}
        \toprule
        Method & Inclusion & Mean len $\uparrow$ & Mean frac $\uparrow$ & Final kept $\uparrow$ \\
        \midrule
        \texttt{DISC-Const} & 0.9638 & 6.36 & 0.5622 & 0.3697 \\
        \texttt{DISC-Mean} & 0.9505 & \textbf{7.68} & \textbf{0.6614} & 0.4733 \\
        \texttt{DISC-Linear} & 0.9497 & 7.06 & 0.6210 & \textbf{0.7840} \\
        \texttt{DISC-Sqrt} & 0.9496 & \underline{7.26} & \underline{0.6412} & \underline{0.7519} \\
        \bottomrule
    \end{tabular}
\end{table}

\phantomsection\label{appen:llm_evaluation_metrics}
\paragraph*{Evaluation metrics.}
For a safe feasible set $\widehat D(X)$, define the inclusion indicator and a maximum-retention output by
\begin{align*}
I_{\rm inc}(X,Y)&=\mathbbm{1}\{\widehat D(X)\subseteq A(X,Y)\},\\
\widehat{\mathcal S}(X)&\in\arg\max_{\mathcal S}\left\{|\mathcal S|:u_{\mathcal S}\in\widehat D(X)\right\},
\qquad \widehat u(X)=u_{\widehat{\mathcal S}(X)}.
\end{align*}
For the four adjustment specifications above, $\mu(u_{\mathcal S};X)$ is nonincreasing in $|\mathcal S|$ at a fixed $X$. Taking the union of two admissible index sets preserves both dependency closure and the margin condition. Hence the maximum-cardinality index set $\widehat{\mathcal S}(X)$ is unique.
The retained length and fraction are $L(X)=|\widehat{\mathcal S}(X)|$ and $R(X)=L(X)/T(X)$. The final-claim indicator is $I_{\rm final}(X)=\mathbbm{1}\{T(X)\in\widehat{\mathcal S}(X)\}$. Thus, `Inclusion', `Mean len.', `Mean frac.', and `Final kept' are the test-sample averages of $I_{\rm inc}$, $L$, $R$, and $I_{\rm final}$, respectively.
All reported entries are then averaged over the repeated random splits.

\subsection{Prompt templates}

This subsection records the prompt templates used in the LLM pipeline described above. We report the stable system prompt together with the user-side template for each stage. Placeholders such as \verb|{problem}|, \verb|{response}|, \verb|{numbered_claims}|, and \verb|{edge_text}| denote the corresponding instance-specific fields inserted at runtime.

\subsubsection{CoT generation}\label{appen:prompt_cot_generation}
We first prompt a base LLM to produce one complete step-by-step solution.

\begin{quote}
\footnotesize\setstretch{1}
\begin{verbatim}
[System]
You are a careful math solver. Solve the user's problem step by step, 
keep the reasoning explicit, compact and brief, and end with the final answer.

[User]
Problem:
{problem}
\end{verbatim}
\end{quote}

\subsubsection{Claim decomposition}\label{appen:prompt_claim_decomposition}
The generated response is then decomposed into an ordered list of short atomic subclaims.

\begin{quote}
\footnotesize\setstretch{1}
\begin{verbatim}
[System]
You are a precise decomposition assistant. Split a math solution into an 
ordered list of atomic subclaims. Preserve order. Do not add new 
information. Each subclaim should contain one mathematical assertion, 
transformation, or conclusion. Keep each subclaim short and direct, 
ideally one sentence. Return JSON only in the form 
{"claims":[{"id":1,"text":"..."},{"id":2,"text":"..."}]}.

[User]
Decompose the following response. Return JSON only.

Problem:
{problem}

Response:
{response}
\end{verbatim}
\end{quote}

\subsubsection{Dependency graph induction}\label{appen:prompt_graph_induction}
Next, we induce a dependency graph over the ordered claims. The user prompt includes two in-context examples.

\begin{quote}
\footnotesize\setstretch{1}
\begin{verbatim}
[System]
You create dependency graphs over ordered subclaims. If claim j depends on 
claim i, include edge i->j. Include implicit dependencies. Keep 
dependencies even if the reasoning is wrong. A priori claims may have no 
ancestors. Never include self-loops. Output only an edge list like 
[1->2,2->4,3->4], or [] if there are no edges.

[User]
Build the dependency graph for these subclaims. Return only the edge list.

Example 1:
Question: How many vertical asymptotes does the graph of 
y = x / (x^2 + 1) have?
NUM = 4
Subclaims:
1. A function has vertical asymptotes where its denominator equals zero.
2. We must solve x^2 + 1 = 0.
3. For all real x, x^2 + 1 > 0.
4. Therefore the function has no vertical asymptotes.
Desired Output:
[1->2,2->3,2->4,3->4]

Example 2:
Question: Consider y = x^2 + 2x + 15. What is the sum of the zeroes?
NUM = 5
Subclaims:
1. The zeroes are the x-values of the x-intercepts.
2. Set x^2 + 2x + 15 = 0.
3. Factor it as (x + 3)(x - 5) = 0.
4. Therefore the zeroes are -3 and 5.
5. Therefore the sum is 2.
Desired Output:
[1->2,2->3,3->4,4->5]

Question:
{problem}

NUM = {n}
Subclaims:
{numbered_claims}

Output:
\end{verbatim}
\end{quote}

\subsubsection{Claim confidence scoring}\label{appen:prompt_claim_scoring}
We then ask a verifier LLM to assign a seperate confidence score in $[0,1]$ to each claim.

\begin{quote}
\footnotesize\setstretch{1}
\begin{verbatim}
[System]
You are a cautious math verifier. Given a problem and a student's 
decomposed claims, assign each claim an independent confidence score in 
[0,1] for how likely it is to be correct relative to the problem. Use the 
dependency edges only as context. Avoid overconfidence: reserve very high 
scores for claims that are clearly correct, reserve very low scores for 
claims that are clearly wrong, and use mid-range scores when a claim is 
uncertain, incomplete, or underspecified. Return JSON only in the form 
{"scores":[{"id":1,"score":0.73},{"id":2,"score":0.15}]}.

[User]
Below is a math problem and a student's answer decomposed into claims. 
The dependency edges are provided only as context.

Problem:
{problem}

Claims:
{numbered_claims}

Edges:
{edge_text}

Return JSON only.
\end{verbatim}
\end{quote}

\subsubsection{Claim label judging}\label{appen:prompt_claim_labeling}
Finally, we construct claim-level reference labels by judging each generated claim against the official solution.

\begin{quote}
\footnotesize\setstretch{1}
\begin{verbatim}
[System]
You are a strict math judge. Given a problem, an official solution, a 
student's decomposed claims, and their dependency edges, judge each claim 
independently as acceptable or not. Use the dependency edges only as 
context. Label 1 only if the claim is mathematically correct and 
compatible with the problem and the reference solution. Label 0 if the 
claim is wrong, unsupported, or too vague to accept as a mathematical 
step. Return JSON only in the form 
{"labels":[{"id":1,"label":1},{"id":2,"label":0}]}.

[User]
Below is a math problem, an official solution, and a student's decomposed 
claims. Judge each claim as acceptable (1) or not acceptable (0).

Problem:
{problem}

Official solution:
{reference_solution}

Claims:
{numbered_claims}

Edges:
{edge_text}

Return JSON only.
\end{verbatim}
\end{quote}

In implementation, when a JSON-formatted response appeared truncated or malformed, we used a short retry instruction requesting compact valid JSON and, if needed, one additional JSON-repair call. These recovery prompts are implementation details for robustness and are not part of the main task definitions above.

\section{The safe object detection experiment}
\label{appen:object_detection_details}
\label{appen:safe_obj_detect}

We use the experimental setup and notation of Section~\ref{sec:safe_obj_detect}. This appendix details detector preprocessing, offset training, calibration scores, and evaluation metrics, followed by additional results for $\gamma\in\{0.5,0.6,0.7\}$, including raw detector baselines.

\subsection{Implementation details}

This subsection gives the implementation details for the preceding object-level experiment. We use a fixed pretrained detector and construct candidate feasible sets from confidence-ordered prefixes using a scalar nestedness level $\lambda$. The comparison includes four conformal methods, namely \texttt{DISC-Const}, \texttt{DISC-Train}, \texttt{Base-Const}, and \texttt{Base-Train}, together with five raw detector baselines obtained by direct confidence thresholding.

\paragraph*{Dataset, detector, and object-level samples.}
We use the COCO 2017 validation set together with a fixed pretrained YOLOv11n detector. We focus on person, vehicle, and traffic-facility classes, corresponding to COCO category ids $\{1\}\cup\{2,3,4,5,6,7,8,9\}\cup\{10,11,13,14,15\}$. We retain only target-class ground-truth objects whose normalized area is at least $0.001$, and discard images containing no such object after this filtering. Let $N_{\rm img}=3205$ denote the resulting number of eligible images. Within each repetition, $\mathcal I_{\rm pre}$, $\mathcal I_{\rm cal}$, and $\mathcal I_{\rm test}$ denote the disjoint image-index sets used for pre-training, calibration, and testing, respectively. Splitting is performed at the image level, so all objects from the same image remain in the same split.

For each such image $X$, the detector is first run with confidence threshold $0$ and maximum output size $300$. We then restrict its outputs to the target classes, sort them by confidence, and retain at most the first $100$ candidates. Let $K(X)\leq100$ denote the resulting number of candidates, written as $(\widehat b_k(X),\widehat c_k(X),\widehat a_k(X))_{k=1}^{K(X)}$, 
where $\widehat b_k(X)\in[0,1]^4$ is the normalized predicted box, $\widehat c_k(X)\in[0,1]$ is its confidence score, and $\widehat a_k(X)$ is its predicted class label. The candidates satisfy $\widehat c_1(X)\geq\cdots\geq\widehat c_{K(X)}(X)$, with ties broken deterministically.
For image $X_i$, the retained ground-truth objects are $Y_{ij}=(b_{ij},a_{ij})$, $j=1,\ldots,T(X_i)$. They induce the object-level pairs $(X_i,Y_{ij})$ used for calibration and evaluation.

\paragraph*{Safety event for one object.}
Fix one object-level pair $(X,Y)$ with $Y=(b,a)$. For $k\in[K(X)]$, let $u_{[k]}(X)=(\widehat b_\ell(X),\widehat c_\ell(X),\widehat a_\ell(X))_{\ell=1}^{k}$, and let $\gU(X)=\{u_{[k]}(X):k\in[K(X)]\}$. The same-class coverage of $Y$ is
\begin{gather*}
\operatorname{cover}(u_{[k]}(X),Y)
=\max_{\ell\leq k:\,\widehat a_\ell(X)=a}
\frac{|b\cap\widehat b_\ell(X)|}{|b|},
\end{gather*}
where $|b|$ is the area of the ground-truth box and the maximum over an empty set is zero. For $\gamma\in(0,1)$, the constraint function is $f(u_{[k]}(X);X,Y)=\gamma-\operatorname{cover}(u_{[k]}(X),Y)$, and the oracle feasible set is
\begin{gather*}
A(X,Y)=\{u\in\gU(X):f(u;X,Y)\leq0\}.
\end{gather*}
Thus, a prefix is safe for $Y$ if at least one retained prediction of the correct class covers at least a $\gamma$-fraction of $b$.

\paragraph*{Prefix family and reported output.}
For each prefix, define the surrogate constraint function
\begin{gather}
\hat f(u_{[k]}(X);X)=\min_{\ell\leq k}\widehat c_\ell(X)=\widehat c_k(X).
\end{gather}
The last equality follows from the decreasing confidence order. The detector provides confidence scores but no direct estimate of object-level coverage. We therefore use the boundary confidence as a surrogate score, rather than a numerical estimate of the true constraint function. As the prefix grows, both the boundary confidence and the true constraint function are nonincreasing.
We use an image-dependent offset $\mu(X)$. For $\lambda\in\mathbb R$, define
\begin{align}
D(X;\lambda,\mu)&=
\left\{u_{[k]}(X)\in\gU(X):
\hat f(u_{[k]}(X);X)+\mu(X)+\lambda\leq0
\right\}\cup\{u_{[K(X)]}(X)\}\\
&=\left\{u_{[k]}(X)\in\gU(X):
\widehat c_k(X)+\mu(X)+\lambda\leq0
\right\}\cup\{u_{[K(X)]}(X)\}.
\end{align}
Including the full prefix ensures that the candidate feasible set is nonempty. To report a single detection output, we select the shortest prefix in $D(X;\lambda,\mu)$, which retains the fewest predicted boxes among its members. Its length, denoted by $k_\lambda(X;\mu)$, determines how many leading detections are retained and is given by
\begin{gather}
k_\lambda(X;\mu)=
\min\left(\left\{k\in[K(X)]:\widehat c_k(X)+\mu(X)+\lambda\leq0\right\}\cup\{K(X)\}\right).
\end{gather}
The reported output is $u^{\rm out}(X;\lambda,\mu)=u_{[k_\lambda(X;\mu)]}(X)$. If no prefix satisfies the margin inequality, the subset contains only the full prefix $u_{[K(X)]}(X)$, which is then reported.
Since confidence scores are nonincreasing in $k$, $D(X;\lambda,\mu)$ is a tail of the prefix family. Object coverage is nondecreasing as boxes are added, so
\begin{gather}
D(X;\lambda,\mu)\subseteq A(X,Y)
\quad\Longleftrightarrow\quad
u^{\rm out}(X;\lambda,\mu)\in A(X,Y).
\end{gather}

\paragraph*{Constant and trained families.}
The constant family is obtained by setting $\mu_{\rm const}(X)\equiv 0$. This yields the method \texttt{DISC-Const} when combined with the directed score below, and \texttt{Base-Const} when combined with the base score. The trained family replaces $\mu_{\rm const}$ by an estimated offset $\widehat\mu(X)$ learned once from the pre-training images indexed by $\mathcal I_{\rm pre}$. The detector itself is fixed throughout; only this post-processing offset is trained.

For both calibration strategies, the constant methods pool object-level scores from images indexed by $\mathcal I_{\rm pre}\cup\mathcal I_{\rm cal}$, whereas the trained methods use only those from $\mathcal I_{\rm cal}$. Each method obtains $\hat\lambda$ by applying the empirical $(1-\alpha)$-quantile rule with an appended $+\infty$ from Section~\ref{sec:dir_inclusion} to its corresponding pooled scores.

The trained family aims to reduce the number of retained boxes while targeting an overall object-level inclusion rate of $1-\alpha$. Covering difficult objects in some images may require retaining many additional low-confidence detections. An image-dependent offset allows the retained prefix length to adapt to this variation. We therefore use the pre-training set to determine how many objects to target for coverage in each image, allowing approximately an $\alpha$ fraction of the pooled objects to be left uncovered. This allocation targets coverage across the pre-training set as a whole, rather than requiring the same coverage fraction in every image.

To measure the prefix length required to cover an object-level sample $(X,Y)$, define
\begin{gather}
L(X,Y) = \begin{cases}
\min\{k\in[K(X)]:u_{[k]}(X)\in A(X,Y)\}, & \text{if such a prefix exists},\\
K(X)+1, & \text{if $u_{[K(X)]}(X)\notin A(X,Y)$.}
\end{cases}
\end{gather}
Thus, an object that cannot be covered by the full detector output is ordered after all coverable objects in the training-target construction.
For each pre-training image $X_i$, $i\in\mathcal I_{\rm pre}$, sort the objectwise lengths $\{L(X_i,Y_{ij})\}_{j=1}^{T(X_i)}$ as $L_{i,(1)}\le\cdots\le L_{i,(T(X_i))}$. This orders its objects from those requiring the shortest prefixes to those requiring the longest. For a coverable object at position $r$, a prefix of length $L_{i,(r)}$ covers at least the first $r$ objects in this order.

The incremental costs of increasing the number of objects targeted for coverage are
\begin{gather*}
\Delta_{i,r}=L_{i,(r)}-L_{i,(r-1)},\qquad r=1,\ldots,T(X_i),
\end{gather*}
where $L_{i,(0)}$ is taken to be $1$ when the smallest safe prefix length is positive, and $0$ otherwise.
For coverable objects, $\Delta_{i,r}$ is the additional prefix length required when the coverage target increases from the first $r-1$ objects to the first $r$ objects. Conversely, relaxing the target from $r$ to $r-1$ saves this many retained boxes. This saving is the basis of the greedy allocation below.

We initialize the target number of objects in each image as $r_i^{\rm pre}=T(X_i)$ and reduce these targets greedily. At each step, we compare the current terminal costs $\Delta_{i,r_i^{\rm pre}}$ across images with $r_i^{\rm pre}>0$, select an image with the largest such cost, and decrease its $r_i^{\rm pre}$ by one. Only the last object in the current ordered sequence is removed from the coverage target; the preceding increment then becomes the new terminal cost for that image. We repeat this step until approximately an $\alpha$ fraction of all pre-training object-level requirements has been removed. The procedure thus seeks shorter prefixes by allocating the allowed uncovered objects to the largest currently available savings. These removals apply only to the pre-training coverage targets, not to the objects used for subsequent calibration or evaluation.

The resulting $r_i^{\rm pre}$ specifies how many objects image $X_i$ should target for coverage. We convert this target into a training response for the image-dependent offset:
\begin{gather}
\tau_i^{\rm pre}=
\begin{cases}
-1, & r_i^{\rm pre}=0,\\
-\widehat c_{\,L_{i,(r_i^{\rm pre})}}(X_i),
& r_i^{\rm pre}>0,\quad L_{i,(r_i^{\rm pre})}\le K(X_i),\\
0, & r_i^{\rm pre}>0,\quad L_{i,(r_i^{\rm pre})}=K(X_i)+1.
\end{cases}
\end{gather}
For a coverable target, the response is the negative confidence score at the last detection in the shortest required prefix. At $\lambda=0$, setting $\mu(X_i)=\tau_i^{\rm pre}$ makes this prefix satisfy the margin condition at equality. If no object remains in the coverage target, we set $\tau_i^{\rm pre}=-1$; since confidence scores lie in $[0,1]$, all candidate prefixes then satisfy the margin condition at $\lambda=0$, and the reported output is the shortest, single-box prefix. If the target includes an object that cannot be covered even by the full detector output, we set $\tau_i^{\rm pre}=0$.

We fit a random forest regressor to the pairs $(X_i,\tau_i^{\rm pre})$, $i\in\mathcal I_{\rm pre}$, to obtain $\widehat\mu(X)$. Its inputs consist only of detector-derived features that are available at test time. Specifically, we use the first $15$ target-class confidence scores; the proportion of target-class detections among all predicted boxes; the corresponding area and confidence-score proportions; analogous proportions among the top-$10$, top-$20$, and top-$50$ highest-confidence detections; the number and area fractions of the target-class candidates retained after truncation; and their average confidence score. Here, truncation refers only to retaining at most the top $100$ target-class detector predictions, not to the ground-truth-based greedy reduction of pre-training coverage targets. Ground-truth boxes and objectwise safe-prefix lengths are used only to construct the pre-training response $\tau_i^{\rm pre}$, and are not included among the predictor inputs. Once $\widehat\mu$ has been fitted using the pre-training images, it is kept fixed while the conformal quantile is estimated using the object-level pairs from images indexed by $\mathcal I_{\rm cal}$.

\paragraph*{Directed score used by \texttt{DISC}.}
For a fixed offset $\mu$, the object-level directed score is
\begin{gather}
V^{\texttt{DISC}}(X,Y;\mu) = \inf \left\{ \lambda\in\mathbb R: D(X;\lambda,\mu)\subseteq A(X,Y) \right\}.
\end{gather}
Both $A(X,Y)$ and $D(X;\lambda,\mu)$ are defined over $\gU(X)$, and $D(X;\lambda,\mu)$ is a tail of this prefix family. Because object coverage is nondecreasing along this family, inclusion is equivalent to requiring its shortest member to have length at least $L(X,Y)$. 
Consequently,
\begin{gather}
V^{\texttt{DISC}}(X,Y;\mu) = \begin{cases}
\displaystyle \max_{1\le k<L(X,Y)} \left\{ -\widehat c_k(X)-\mu(X) \right\}, & 2\le L(X,Y)\le K(X), \\
-\infty, & L(X,Y)=1, \\
+\infty, & L(X,Y)=K(X)+1.
\end{cases}
\end{gather}
The first case gives the critical boundary above which every prefix shorter than the shortest safe prefix is excluded from $D(X;\lambda,\mu)$; exclusion requires $\lambda$ to strictly exceed this boundary. The value $-\infty$ means that every candidate prefix is already safe, whereas $+\infty$ means that even the full prefix is unsafe. Because the margin family uses a non-strict boundary, an infinitesimal upward perturbation, implemented through a fixed numerical tolerance, is applied at exact equality.

\paragraph*{Base score.}
The base methods use the uniform residual construction in \eqref{eq:uniform_score}, with the same constraint function $f$ and surrogate $\hat f$ defined above.
Then, for one object-level sample $(X,Y)$, the base score is
\begin{gather}
V^{\texttt{Base}}(X,Y;\mu) = \sup_{u\in\gU(X)} \left\{ f(u;X,Y)-\hat f(u;X)-\mu(X) \right\}.
\end{gather}
Since $\gU(X)$ is finite, this becomes
\begin{gather}
V^{\texttt{Base}}(X,Y;\mu) = \max_{1\le k\le K(X)} \left\{ \gamma - \operatorname{cover}(u_{[k]}(X),Y) - \widehat c_k(X) - \mu(X) \right\}.
\end{gather}
In the implementation, the quantity inside the braces is precomputed for every target object and every prefix length, and the score is obtained by taking the maximum over $k$.

\paragraph*{Evaluation metrics.}
For each test-image index $i\in\mathcal I_{\rm test}$, let $u_i^{\rm out}$ denote the reported output of the method being evaluated, and let $|u_i^{\rm out}|$ denote its number of retained boxes. The primary metric is the empirical object-level inclusion rate
\begin{gather*}
\widehat{\mathrm{Inc}}_{\rm obj}=
\frac{
\sum_{i\in\mathcal I_{\rm test}}\sum_{j=1}^{T(X_i)}
\mathbbm{1}\{u_i^{\rm out}\in A(X_i,Y_{ij})\}
}{
\sum_{i\in\mathcal I_{\rm test}}T(X_i)
}.
\end{gather*}
For the conformal methods, safety of the reported shortest prefix is equivalent to inclusion of the entire prefix tail in $A(X_i,Y_{ij})$, so this metric is also the empirical set-inclusion rate.
We also report the mean retained prefix length $\widehat{\mathrm{Ret}}=\frac{1}{|\mathcal I_{\rm test}|}\sum_{i\in\mathcal I_{\rm test}}|u_i^{\rm out}|$, together with the empirical retained-length quantiles $\widehat Q_{0.25}^{\rm Ret}$, $\widehat Q_{0.50}^{\rm Ret}$, and $\widehat Q_{0.75}^{\rm Ret}$ across test images. Finally, we report the mean union area of the retained boxes,
\begin{gather*}
\widehat{\mathrm{Area}}=
\frac{1}{|\mathcal I_{\rm test}|}
\sum_{i\in\mathcal I_{\rm test}}
\left|\bigcup_{\ell=1}^{|u_i^{\rm out}|}\widehat b_\ell(X_i)\right|.
\end{gather*}
Raw confidence thresholding may return an empty output; in that case, every object in the image has inclusion indicator zero, and the retained length and union area are both zero. Since the box coordinates are normalized to $[0,1]^2$, the union area always lies in $[0,1]$.

\paragraph*{Additional results.}
Tables~\ref{tab:object_detection_gamma05}, \ref{tab:object_detection_gamma06_full}, and~\ref{tab:object_detection_gamma07} report results for $\gamma=0.5$, $0.6$, and $0.7$, respectively. Each table includes five raw detector baselines that retain all target-class predictions with confidence strictly above $0.01$, $0.02$, $0.05$, $0.1$, or $0.2$. These raw baselines use neither the offset $\mu$ nor the calibration level $\lambda$; their results at $\gamma=0.6$ supplement the conformal-method comparison in the main text.

For each value of $\gamma$, we compute the metrics on the test set of each split and average the resulting values over $100$ random splits. In particular, the reported quartiles are averages of the within-split retained-length quartiles, not quartiles of pooled outputs across splits. For table marking, we use an inclusion cutoff of $0.898$, while the nominal inclusion target remains $1-\alpha=0.9$. Among methods whose mean inclusion exceeds this cutoff, boldface and underlining mark the smallest and second-smallest values in the retention and area columns. A superscript $-$ marks mean inclusion below the cutoff.

\begin{table}[t]
    \centering
    \caption{Safe object detection at $\gamma=0.5$ and $\alpha=0.1$. Boldface and underlining mark the smallest and second-smallest values in the retention and area columns.}
    \label{tab:object_detection_gamma05}
    \begin{tabular}{lcccccc}
        \toprule
        Method & Inclusion & Mean Ret $\downarrow$ & $Q_{0.25}^{\rm Ret}\ \downarrow$ & $Q_{0.50}^{\rm Ret}\ \downarrow$ & $Q_{0.75}^{\rm Ret}\ \downarrow$ & Area $\downarrow$ \\
        \midrule
        \texttt{DISC-Const} & 0.900 & \underline{10.40} & 2.83 & 5.45 & \underline{13.84} & 0.374 \\
        \texttt{DISC-Train} & 0.901 & 14.36 & \textbf{2.00} & \textbf{3.16} & \textbf{8.81} & \underline{0.368} \\
        \texttt{Base-Const} & 0.959 & 82.79 & 72.18 & 100.00 & 100.00 & 0.576 \\
        \texttt{Base-Train} & 0.956 & 44.56 & \underline{2.02} & 10.50 & 100.00 & 0.431 \\
        Raw $>0.01$ & 0.943 & 27.14 & 4.00 & 12.00 & 37.00 & 0.404 \\
        Raw $>0.02$ & 0.929 & 18.13 & 3.00 & 8.00 & 25.00 & 0.385 \\
        Raw $>0.05$ & 0.899 & \textbf{10.28} & \textbf{2.00} & \underline{5.00} & 14.00 & \textbf{0.363} \\
        Raw $>0.10$ & $0.863^{-}$ & 6.69 & 1.00 & 3.00 & 9.00 & 0.347 \\
        Raw $>0.20$ & $0.805^{-}$ & 4.42 & 1.00 & 3.00 & 6.00 & 0.330 \\
        \bottomrule
    \end{tabular}
\end{table}

\begin{table}[t]
    \centering
    \caption{Safe object detection at $\gamma=0.6$ and $\alpha=0.1$.}
    \label{tab:object_detection_gamma06_full}
    \begin{tabular}{lcccccc}
        \toprule
        Method & Inclusion & Mean Ret $\downarrow$ & $Q_{0.25}^{\rm Ret}\ \downarrow$ & $Q_{0.50}^{\rm Ret}\ \downarrow$ & $Q_{0.75}^{\rm Ret}\ \downarrow$ & Area $\downarrow$ \\
        \midrule
        Raw $>0.01$ & 0.935 & 27.14 & 4.00 & 12.00 & 37.00 & 0.404 \\
        Raw $>0.02$ & 0.920 & 18.13 & 3.00 & 8.00 & 25.00 & 0.385 \\
        Raw $>0.05$ & $0.888^{-}$ & 10.28 & 2.00 & 5.00 & 14.00 & 0.363 \\
        Raw $>0.10$ & $0.849^{-}$ & 6.69 & 1.00 & 3.00 & 9.00 & 0.347 \\
        Raw $>0.20$ & $0.789^{-}$ & 4.42 & 1.00 & 3.00 & 6.00 & 0.330 \\
        \bottomrule
    \end{tabular}
\end{table}

\begin{table}[H]
    \centering
    \caption{Safe object detection at $\gamma=0.7$ and $\alpha=0.1$. Boldface and underlining mark the smallest and second-smallest values in the retention and area columns.}
    \label{tab:object_detection_gamma07}
    \begin{tabular}{lcccccc}
        \toprule
        Method & Inclusion & Mean Ret $\downarrow$ & $Q_{0.25}^{\rm Ret}\ \downarrow$ & $Q_{0.50}^{\rm Ret}\ \downarrow$ & $Q_{0.75}^{\rm Ret}\ \downarrow$ & Area $\downarrow$ \\
        \midrule
        \texttt{DISC-Const} & 0.901 & \textbf{18.04} & 3.48 & 8.53 & \textbf{24.76} & 0.393 \\
        \texttt{DISC-Train} & 0.901 & 23.49 & \textbf{2.00} & \textbf{3.74} & \underline{24.96} & \underline{0.387} \\
        \texttt{Base-Const} & 0.938 & 82.79 & 72.18 & 100.00 & 100.00 & 0.576 \\
        \texttt{Base-Train} & 0.938 & 56.09 & \underline{3.00} & 75.87 & 100.00 & 0.467 \\
        Raw $>0.01$ & 0.918 & 27.14 & 4.00 & 12.00 & 37.00 & 0.404 \\
        Raw $>0.02$ & 0.901 & \underline{18.13} & \underline{3.00} & \underline{8.00} & 25.00 & \textbf{0.385} \\
        Raw $>0.05$ & $0.865^{-}$ & 10.28 & 2.00 & 5.00 & 14.00 & 0.363 \\
        Raw $>0.10$ & $0.824^{-}$ & 6.69 & 1.00 & 3.00 & 9.00 & 0.347 \\
        Raw $>0.20$ & $0.763^{-}$ & 4.42 & 1.00 & 3.00 & 6.00 & 0.330 \\
        \bottomrule
    \end{tabular}
\end{table}

Increasing the raw confidence threshold shortens the retained output but lowers object-level inclusion. Some raw thresholds are competitive at individual values of $\gamma$, but they are not calibrated to a prespecified inclusion target and their inclusion rates change with the required coverage fraction.

\section{Theoretically Optimal Design of Safe Feasible Set}\label{appen:optimal_safety_subset}

\subsection{Maximum-Volume Safe Feasible Set}

\paragraph*{Problem setup.}
Let $(X,Y)\sim P$, where $X\in\mathcal X$ is the covariate and $Y\in\mathcal Y$ represents the uncertain outcome. For each realization
$(x,y)\in\mathcal X\times\mathcal Y$, define the oracle feasible set $A(x,y)=\{u\in\mathcal U:f(u,x,y)\le 0\}$. We seek a measurable set-valued function $D:\mathcal X\to 2^{\mathcal U}$ with maximal expected volume subject to the marginal inclusion constraint
\begin{gather*}
    \max_{D:\mathcal X\to 2^{\mathcal U}}\mathbb E[\mu(D(X))]\quad\text{s.t.}\quad \mathbb P\{D(X)\subseteq A(X,Y)\}\ge 1-\alpha .
\end{gather*}
For any set-valued function $D$, define its \emph{inclusion event} by
$S_D=\{(x,y)\in\mathcal X\times\mathcal Y:D(x)\subseteq A(x,y)\}$.
Then $\mathbb P\{D(X)\subseteq A(X,Y)\}= \Prob\{(X,Y)\in S_D\}.$

\paragraph*{Scenario-induced covariate-dependent safe feasible sets.}
Let $\mathcal E = \{E\subseteq\mathcal X\times\mathcal Y: \Prob\{(X,Y)\in E\}\ge 1-\alpha\}.$ For each $E\in\mathcal E$, define the $x$-section
$E_x=\{y\in\mathcal Y:(x,y)\in E\}$ and the \emph{induced set-valued function}
$D_E(x)=\bigcap_{y\in E_x}A(x,y)=\{u\in\mathcal U:f(u,x,y)\le 0\text{ for all }y\in E_x\}$.
For every $(x,y)\in E$, we have $y\in E_x$, hence
$D_E(x)\subseteq A(x,y)$. Therefore $E\subseteq S_{D_E}$ and
\begin{gather*}
    \mathbb P\{D_E(X)\subseteq A(X,Y)\}=\Prob\{(X,Y)\in S_{D_E}\}\ge \Prob\{(X,Y)\in E\}\ge 1-\alpha .
\end{gather*}
Thus, every scenario-induced function $D_E$ is feasible. In general, the equality
$\Prob\{D_E(X)\subseteq A(X,Y)\}=\Prob\{(X,Y)\in E\}$ does not necessarily hold; it
holds only when $E=S_{D_E}$ up to a $P$-null set.

\paragraph*{Reduction to scenario-induced functions.}
Conversely, let $D$ be any feasible set-valued function, so that $\mathbb P\{D(X)\subseteq A(X,Y)\}\ge 1-\alpha$. Since $D(x)\subseteq A(x,y)$ for every $(x,y)\in S_D$, we have $D(x)\subseteq \bigcap_{y\in \gY:(x,y)\in S_D}A(x,y)=D_{S_D}(x)$ for each $x$. Moreover, $\Prob(S_D)=\mathbb P\{D(X)\subseteq A(X,Y)\}\ge 1-\alpha$, so $S_D\in\mathcal E$. Hence every feasible $D$ is pointwise contained in a scenario-induced feasible function $D_{S_D}$, that is, $D(x)\subseteq D_{S_D}(x)$ for all $x$. Consequently, $\mu(D(x))\le \mu(D_{S_D}(x))$ for all $x$, and therefore $\mathbb E[\mu(D(X))] \le \mathbb E[\mu(D_{S_D}(X))]$. Thus, when maximizing expected volume, it is sufficient to optimize over scenario-induced set-valued functions of the form $D_E(x)=\cap_{y\in E_x}A(x,y)$.

\paragraph*{Optimal characterization.}
The optimal covariate-dependent safe feasible set is characterized by
\begin{gather*}
    E^*\in\argmax_{E\subseteq\mathcal X\times\mathcal Y}\mathbb E\left[\mu\left(\bigcap_{y\in \gY:(X,y)\in E}A(X,y)\right)\right]\quad\text{s.t.}\quad \Prob\{(X,Y)\in E\}\ge 1-\alpha,
\end{gather*}
and
$D^*(x)=\bigcap_{y\in \gY:(x,y)\in E^*}A(x,y)$. Equivalently, $D^*(x)=\{u\in\mathcal U:\sup_{y\in \gY:(x,y)\in E^*} f(u,x,y)\le 0\}$.

\paragraph*{Proof of optimality.}
Let $D$ be any feasible set-valued function. Its inclusion event $S_D$ satisfies $P(S_D)\ge 1-\alpha$, and the reduction above gives $D(x)\subseteq D_{S_D}(x)$ for all $x$. Therefore $\mathbb E\{\mu(D(X))\}\le \mathbb E\{\mu(D_{S_D}(X))\}$. Since $E^*$ maximizes the expected volume among all events $E\subseteq\mathcal X\times\mathcal Y$ with $P(E)\ge 1-\alpha$, we also have $\mathbb E\{\mu(D_{S_D}(X))\}\le \mathbb E\{\mu(D_{E^*}(X))\}$. Combining the two inequalities yields $\mathbb E\{\mu(D(X))\}\le \mathbb E\{\mu(D_{E^*}(X))\}=\mathbb E\{\mu(D^*(X))\}$.
Therefore no feasible set-valued function has larger expected volume than
$D^*$, and hence $D^*$ solves the original marginal chance-constrained problem.

\paragraph*{Remark on conditional and marginal inclusion.}
The constraint above is marginal: $\mathbb P\{D(X)\subseteq A(X,Y)\}\ge 1-\alpha$. It does not require the stronger conditional guarantee $\mathbb P\{D(x)\subseteq A(x,Y)\mid X=x\}\ge 1-\alpha$ for every $x$. If the stronger conditional guarantee is desired, then one would instead solve, for each fixed $x$, $D^*(x)=\bigcap_{y\in E_x^*}A(x,y)$ with $E_x^*\in\arg\max_{E_x\subseteq\mathcal Y}\mu\left(\bigcap_{y\in E_x}A(x,y)\right)$ subject to $\Prob_{Y\mid X}(E_x\mid X=x)\ge 1-\alpha$.
This conditional formulation is different from the marginal formulation,
because the marginal constraint allows probability mass to be allocated
unevenly across different values of $X$.

\paragraph*{Generalized analysis of the optimal safe feasible set.}
The characterization above is based on scenario-induced sets. A complementary view, parallel to the decision-aware analysis, is to ask how much local safety budget should be assigned to each covariate value, and how much safe volume can be retained once that local budget is fixed. 
For $x\in\mathcal X$ and $v\in[0,1]$, define
\begin{gather*}
    \gV(v,x):=\sup\left\{\mu\left(\bigcap_{y\in E_x}A(x,y)\right):E_x\subseteq\mathcal Y,\ \Prob_{Y\mid X}(E_x\mid X=x)\ge v\right\}.
\end{gather*}
At covariate value $x$, $\gV(v,x)$ is the supremum of the retained volume under a local inclusion requirement of at least $v$. A larger $v$ imposes a stricter requirement and cannot increase this supremum.

\paragraph*{Reduction to local budget allocation.}
For an event $E\subseteq\mathcal X\times\mathcal Y$, define its local probability budget by $b_E(x):=\Prob\{(X,Y)\in E\mid X=x\}$. If $D_E(x)=\bigcap_{y\in E_x}A(x,y)$ is the rule induced by $E$, then the event $E$ allocates local budget $b_E(x)$ to covariate value $x$, and its marginal budget is $\mathbb E\{b_E(X)\}=\Prob\{(X,Y)\in E\}$. By the definition of $\gV$, the induced rule satisfies $\mu(D_E(x))\le \gV(b_E(x),x)$ for every $x$.

Conversely, for any local budget function $b:\mathcal X\to[0,1]$, an oracle rule $D_b$ can be chosen pointwise to attain the best available volume at that budget, namely $\mu(D_b(x))=\gV(b(x),x)$ whenever the supremum is attained. Therefore the original optimal design problem
\begin{gather*}
    \max_{D:\mathcal X\to 2^{\mathcal U}}\mathbb E\{\mu(D(X))\}\quad \text{s.t.}\quad \Prob\{D(X)\subseteq A(X,Y)\}\ge 1-\alpha
\end{gather*}
reduces, at the oracle level, to the local budget allocation problem
\begin{gather*}
    \max_{b:\mathcal X\to[0,1]}\mathbb E\{\gV(b(X),X)\}\quad \text{s.t.}\quad \mathbb E\{b(X)\}\ge 1-\alpha.
\end{gather*}
Here $b(x)$ is a covariate-dependent local safety budget. Unlike the decision-aware formulation, this objective rewards retained safe volume.

\paragraph*{Lagrangian characterization.}
If $v\mapsto\gV(v,x)$ is differentiable on the interior of $[0,1]$, then the corresponding Lagrangian first-order condition gives a constant marginal value across covariate values. Specifically, there exists a non-negative multiplier $\lambda^*\ge 0$ such that any interior optimizer $b^*$ satisfies
\begin{gather*}
    \frac{\partial}{\partial v}\gV(v,x)\big|_{v=b^*(x)}=\lambda^*,\qquad x\in\mathcal X\text{ whenever } b^*(x)\in(0,1),
\end{gather*}
with $\mathbb E\{b^*(X)\}=1-\alpha$ when the constraint is active. Boundary points satisfy the usual KKT inequalities. In words, the optimal allocation equalizes the local marginal value of the safety budget across covariate values. This yields a threshold-type policy: $b^*(x)$ is determined by comparing the local marginal value $\partial_v\gV(v,x)$ with a common cutoff. The negative threshold notation $\lambda_-$ used below is only a sign convention for the estimated rule written in terms of $-\widehat{\phi}$.

\paragraph*{Data-driven implication.}
This characterization suggests a direct plug-in construction. Suppose estimators of the local marginal value $\partial\gV(v,x)/\partial v$ and of the map $b\mapsto D_b$ are available, denoted by $\widehat{\phi}(v,x)$ and $\widehat D_b(x)$. For a threshold level $\lambda_-\le 0$, the estimated budget selected at covariate value $x$ is $\sup\{b\in[0,1]:-\widehat{\phi}(b,x)\le \lambda_-\}$. The corresponding threshold-indexed set collects decisions from the estimated sets $\widehat D_b(x)$ satisfying $-\widehat{\phi}(b,x)\le \lambda_-$:
\begin{gather*}
    \widehat{D}_{\lambda_-}(x) = \left\{ u\in\widehat{D}_b(\cdot) : -\widehat{\phi}(b,x) \leq \lambda_- \right\}\,.
\end{gather*}
This is the data-driven threshold analogue of the oracle policy above. 
To calibrate this family, define
\begin{gather*}
    \lambda_{i-}:=\sup\{\lambda_-\le 0:\widehat D_{\lambda_-}(X_i)\subseteq A(X_i,Y_i)\}
\end{gather*}
for each calibration observation, and set
\begin{gather*}
    \hat\lambda_-:=\inf\left\{\lambda_-:\frac{1}{n+1}\sum_{i=1}^n\mathbbm{1}\{\lambda_{i-}\le \lambda_-\}\ge 1-\alpha\right\}.
\end{gather*}
The final safe feasible set is $D(X_{n+1})=\widehat D_{\hat\lambda_-}(X_{n+1})$. 
As in the decision-aware analysis below, the task is to allocate the marginal safety budget across covariate values. Here the objective is retained set volume rather than the cost of a deployed action.

\subsection{Decision-Aware Design of $L^*(u)$}

To better understand the optimization of the safety margin family, it is useful to step back from the parametric family in Section \ref{sec:family_optimization} and describe an oracle formulation that isolates the statistical tradeoff between downstream decision cost and safety.

\paragraph*{Problem setup.}
Let $(X,Y)\sim P=P_X\times P_{Y\mid X}$, where $P_X$ is the marginal distribution of $X$ and $P_{Y\mid X}$ is the conditional distribution of $Y$ given $X$. For each realization $(x,y)\in\mathcal X\times\mathcal Y$, define the oracle feasible set by $A(x,y)=\{u\in\mathcal U:f(u,x,y)\le 0\}$. We seek a measurable set-valued rule $D:\mathcal X\to 2^{\mathcal U}$ satisfying the marginal inclusion constraint $\Prob\{D(X)\subseteq A(X,Y)\}\ge 1-\alpha$.

\paragraph*{Scenario-induced decision rules.}
To describe feasible rules, it is convenient to begin from events in the joint space $(X,Y)$ rather than from decision rules directly. Let $\mathcal E=\{E\subseteq\mathcal X\times\mathcal Y:\Prob\{(X,Y)\in E\}\ge 1-\alpha\}$. For each $E\in\mathcal E$, define the $x$-section $E_x=\{y\in\mathcal Y:(x,y)\in E\}$ and the induced set-valued rule $D_E(x)=\bigcap_{y\in E_x}A(x,y)$. In words, $D_E(x)$ keeps those decisions that are safe simultaneously for all responses $y$ included in the section $E_x$. We write $\breve{\mathcal D}=\{D_E:E\in\mathcal E\}$ for the family of all such induced rules. If $(x,y)\in E$, then $y\in E_x$ and hence $D_E(x)\subseteq A(x,y)$. Therefore every $D_E\in\breve{\mathcal D}$ is feasible, since 
\begin{gather*}
    \Prob\{D_E(X)\subseteq A(X,Y)\}\ge \Prob\{(X,Y)\in E\}\ge 1-\alpha\,.
\end{gather*}

\paragraph*{Reduction to point decisions.}
Now introduce a downstream decision cost $\phi(u,x)$. A set-valued rule $D$ specifies which decisions are declared safe at covariate value $x$, while the deployed action minimizes the decision cost over that set. Thus, for any set-valued rule $D$, define its induced point decision rule by $g_D(x)\in\argmin_{u\in D(x)}\phi(u,x)$. More generally, write $g:\mathcal X\to\mathcal U$ for an arbitrary measurable point decision rule.

To quantify the conditional safety level of a decision $u$ at covariate value $x$, define
\begin{gather*}
    \Lambda(u,x):=\sup\left\{\Prob_{Y\mid X}(E_x\mid X=x):E_x\subseteq\mathcal Y,\ u\in\bigcap_{y\in E_x}A(x,y)\right\}.
\end{gather*}
The role of $\Lambda(u,x)$ is to assign a single scalar value to each candidate decision $u$: larger $\Lambda(u,x)$ means that $u$ stays safe over a larger portion of the conditional distribution of $Y$ given $X=x$. Equivalently, $\Lambda(u,x)=\Prob\{u\in A(x,Y)\mid X=x\}$, so $\Lambda(u,x)$ is the conditional safety level of $u$ at $x$.

For any induced rule $D_E$, the deployed action $g_{D_E}(x)$ belongs to $D_E(x)$, and therefore $\Prob\{D_E(X)\subseteq A(X,Y)\mid X=x\}\le \Lambda(g_{D_E}(x),x)$. Thus, the conditional safety level of the deployed action is at least the conditional inclusion probability of the whole set. This motivates constructing an induced set-valued rule from a point decision rule. Given a measurable point decision rule $g$, define $E_g=\{(x,y)\in\mathcal X\times\mathcal Y:g(x)\in A(x,y)\}$ and let $D_g:=D_{E_g}$. Here $E_g$ is the exact event on which the point decision $g(x)$ is safe, while $D_g$ is the induced set-valued rule generated by that event. Then $g(x)\in D_g(x)$ for every $x$, and the marginal inclusion constraint is controlled by the average conditional safety level of $g$. Hence the oracle decision-aware problem can be written as
\begin{gather*}
    g^*\in\argmin_{g:\mathcal X\to\mathcal U}\mathbb E\{\phi(g(X),X)\}\quad \text{s.t.}\quad \mathbb E\{\Lambda(g(X),X)\}\ge 1-\alpha.
\end{gather*}
As in the local budget formulation above, the marginal constraint is expressed through an average conditional quantity, leaving a local tradeoff at each covariate value.

\paragraph*{Lagrangian characterization.}
Introduce a Lagrange multiplier $\eta\ge 0$. The corresponding Lagrangian is  
\begin{gather*}
    \mathcal L(g,\eta)=\mathbb E\{\phi(g(X),X)-\eta\Lambda(g(X),X)\}+\eta(1-\alpha)\,.
\end{gather*}
For a fixed $\eta$, the optimizing rule therefore satisfies $g_\eta(x)\in\argmin_{u\in\mathcal U}\{\phi(u,x)-\eta\Lambda(u,x)\}$. Thus the formal Lagrangian tradeoff is additive: one pays downstream decision cost $\phi(u,x)$ and receives a reward proportional to the conditional safety value $\Lambda(u,x)$.

The ratio $\phi(u,x)/\Lambda(u,x)$ enters as a convenient ranking score suggested by this tradeoff. When $\Lambda(u,x)$ is interpreted as the amount of conditional safety carried by decision $u$, the ratio measures downstream decision cost per unit of safety value. In this sense, the ratio rule should be viewed as a practical surrogate motivated by the Lagrangian tradeoff, rather than as the literal first-order optimality condition.

\paragraph*{Data-driven implementation.}
Suppose an estimator $\hat\lambda(u,x)$ of $\Lambda(u,x)$ is available. The oracle ranking rule $\phi(u,x)/\Lambda(u,x)$ then suggests the empirical surrogate $\phi(u,x)/\hat\lambda(u,x)$. We therefore consider the one-parameter family 
\begin{gather*}
    \widehat D_\lambda(x)=\{u\in\mathcal U:\phi(u,x)/\hat\lambda(u,x)\le \lambda\}\,.
\end{gather*}
The nestedness level $\lambda$ controls the tradeoff between downstream decision cost and estimated safety value: smaller $\lambda$ retains only decisions with smaller decision cost relative to their estimated safety level, while larger $\lambda$ admits more decisions.

To calibrate this family, define for each calibration point $\lambda_i=\inf\{\lambda\ge 0:\widehat D_\lambda(X_i)\subseteq A(X_i,Y_i)\}$. In words, $\lambda_i$ is the smallest nestedness level for which the candidate feasible set becomes safe on the $i$th sample. We then set
\begin{gather*}
    \hat\lambda:=\sup\left\{\lambda:\frac{1}{n+1}\sum_{i=1}^n\mathbbm{1}\{\lambda_i\le \lambda\}<\alpha\right\}.
\end{gather*}
The resulting safe feasible set is $D(X_{n+1})=\widehat D_{\hat\lambda}(X_{n+1})$, and the deployed action is the cost-minimizing element of this set, namely $g_D(X_{n+1})\in\argmin_{u\in D(X_{n+1})}\phi(u,X_{n+1})$.

\setstretch{.8}
\putbib
\end{bibunit}

\end{document}